\documentclass[12pt]{article}
\usepackage{amsmath}
\usepackage{graphicx}
\usepackage{enumerate}
\usepackage{natbib}
\usepackage{url} 

\RequirePackage{amsthm,amsmath,mathrsfs,multirow,morefloats,booktabs,subfigure,pdfpages,color,algorithm,array,hhline,makecell,epstopdf,algorithm}
\RequirePackage{algorithm,algpseudocode}
\algrenewcommand\algorithmicrequire{\textbf{Input:}}
\algrenewcommand\algorithmicensure{\textbf{Output:}}
\usepackage{amsmath,amssymb}

\theoremstyle{plain}
\newtheorem{thm}{Theorem}[section]
\newtheorem{defin}[thm]{Definition}
\newtheorem{lem}[thm]{Lemma}

\newtheorem{rem}[thm]{Remark}
\newtheorem{cor}[thm]{Corollary}
\allowdisplaybreaks[4]
\newcommand{\blind}{1}

\begin{document}

\def\spacingset#1{\renewcommand{\baselinestretch}%
{#1}\small\normalsize} \spacingset{1}


\if1\blind
{
  \title{\bf Impact of regularization on spectral clustering under the mixed membership stochasticblock model}
  \author{Huan Qing\\ 
    School of Mathematics, China University of Mining and Technology\\
    and \\
    Jingli Wang \\
    School of Statistics and Data Science, Nankai University}
  \maketitle
} \fi

\if0\blind
{
  \bigskip
  \bigskip
  \bigskip
  \begin{center}
    {\LARGE\bf Impact of regularization on spectral clustering under the mixed membership stochasticblock model}
\end{center}
  \medskip
} \fi

\bigskip
\begin{abstract}
Mixed membership community detection is a challenge problem in network analysis. To estimate the memberships and study the impact of regularized spectral clustering under the mixed membership stochastic block (MMSB) model, this article proposes two efficient spectral clustering approaches based on regularized Laplacian matrix, Simplex Regularized Spectral Clustering (SRSC) and Cone Regularized Spectral Clustering (CRSC). SRSC and CRSC methods are designed based on the ideal simplex structure and the ideal cone structure in the variants of the eigen-decomposition of the population regularized Laplacian matrix. 
We show that these two approaches SRSC and CRSC are asymptotically consistent under mild conditions by providing error bounds for the inferred membership vector of each node under MMSB. Through the theoretical analysis, we give the upper and lower bound for the regularizer $\tau$. By introducing a parametric convergence probability, we can directly see that when $\tau$ is large these two methods may still have low error rates but with a smaller probability.  Thus we give an empirical optimal choice of $\tau$ is $O(\mathrm{log}(n))$ with $n$ the number of nodes to detect sparse networks. The proposed two approaches are successfully applied to synthetic and empirical networks with encouraging results compared with some benchmark methods. 
\end{abstract}

\noindent%

\section{Introduction}
\label{sec:intro}

 Detecting the memberships (or community detection, or clustering) in a network has a long history \citep{Francois1971, White1976, SBM, B1984The, Wasserman1994, Le2016}.
Many methods are well developed to detect communities. In these studies some may focus on the (non-mixed membership) community detection problem in which  one node/individual only belongs to one community in a network, such as \cite{SBM, SCORE,papadopoulos2012community,RSC,SLIM}. Some may be interested in the mixed membership community detection in which some vertices can belong to many communities \citep{MMSB,goldenberg2010a,mixedSCORE,mao2020estimating,OCCAM}, and such case is more realistic. In  this paper, we study the problem of mixed membership community detection. 

The stochastic blockmodel (SBM) \citep{SBM} is perhaps the most popular model for community detection. In SBM, it is assumed that there are $K$ disjoint communities, i.e, no mixed membership nodes. And edges only depend on the memberships of nodes, thus  the average degree of connectivity for nodes in a same community is much higher than in different communities. Therefore, SBM  assumes that the nodes in a same community have the same probability to connect with others. The mixed memberships stochastic blockmodel (MMSB) \citep{MMSB} extended the SBM to mixed membership networks by allowing each node to have different degrees among all communities.  We intend to use MMSB to generate mixed membership networks in this paper.

Spectral clustering is a classical and attractive method to identify communities due to its computational tractability in network analysis. It was first introduced by \cite{1973Lower} and \cite{1973Algebraic} for graph partitions, and then it was extended and developed for different problems \citep{1991Partitioning,1995An,1996Spectral,ng2001spectral}.    \cite{von2007tutorial} provided a nice tutorial for spectral clustering. It is well known that spectral clustering method is benefit from the normalization \citep{von2007tutorial,amini2013Pseudo,sarkar2015role}.  \cite{von2008consistency} studied the consistency of the spectral clustering method and showed that the normalized spectral clustering is consistent under general conditions, while the un-normalized spectral clustering method is consistent under some very specific conditions which may not be satisfied in practice. \cite{Bickel21068} provided a general framework for the analysis of consistency of community detection methods.  \cite{lei2015consistency} also studied the consistency of spectral clustering for very sparse networks  even when the order of the maximum expected degree is as small as $log(n)$.  And \cite{sarkar2015role} theoretically studied the impact of normalization of spectral clustering for SBM. \cite{RSC} proposed an efficient regularized spectral clustering (RSC) algorithm for community detection under Degree Corrected Stochastic Block Model (DCSBM) \citep{DCSBM} by considering the regularized Laplacian matrix instead directly using the adjacency matrix.  \cite{joseph2016impact} focused on how the regularization influence the performance of spectral clustering method even when the minimum degree is of constant order, and they found a large regularizer may might be helpful for relaxing the constrain for the minimum degree. They proposed a data-driven methodology for selecting the regularization parameter and suggested  that moderate values of the regularizer may lead to better clustering performance.  Based on the work of \cite{abbe2020entrywise}, \cite{su2019strong} showed the strong consistency of spectral clustering with regularized Laplacian for the SBM and DCSBM. Under the framework of MMSB, \cite{mao2020estimating} developed a spectral clustering algorithm called SPACL based on the leading eigenvectors' simplex structure of the population adjacency matrix and  provided uniform rates of convergence for the inferred community membership vector of each node. By considering  the degree heterogeneity,  \cite{mixedSCORE} modified the Spectral Clustering On Ratios-of-Eigenvectors (SCORE) method \citep{SCORE}, which was designed for non-mixed community detection, to mixed membership problem by considering a vertex  hunting procedure and a membership reconstruction step, and called it as Mixed-SCORE. There are some more related works for spectral clustering method such as \cite{OCCAM,ZHANG2007ident,Chin15,zhou2019analysis}.

In this paper, we provide an attempt at studying the impact of regularization on spectral clustering by constructing two efficient spectral clustering algorithms for mixed membership community detection problem under the MMSB model. We also propose a reasonable explanation on the choice of the optimal regularization parameter. Below are the four main contributions of this paper.
\begin{itemize}
	\item By carefully analyzing the variants of the eigen-decomposition of the population regularized Laplacian matrix under MMSB, we find that there exist ideal simplex structure and ideal cone structure. Based on this finding, to recover the mixed memberships under MMSB, we propose two efficient algorithms: simplex regularized spectral clustering (SRSC for short) and cone regularized spectral clustering (CRSC for short). Empirically, for the simplex structure which generates the designing of SRSC, we apply the successive projection (SP) algorithm to find the corners; for the cone structure which inspires us to design CRSC, we use the SVM-cone algorithm developed in \cite{MaoSVM} to find the corners by applying the one-class SVM  to the normalized rows of the data matrix.
	\item By providing the equivalence algorithms of SRSC and CRSC for the convenience of theoretical analysis, we obtain the node-wise error bounds of SRSC and CRSC, where we take the advantage of Theorem 10 in \cite{cai2013sparse} to obtain the row-wise eigenvector deviation of the regularized Laplacian matrix.
	\item We study the regularization for MMSB using a parametric probability (The parametric probability is the convergence probability involving some parameters.). By carefully analyzing the step of obtaining the spectral norm difference between the sample and population regularized Laplacian matrix, we obtain the theoretical upper bound of the regularization parameter with a parametric probability. After obtaining the node-wise error bounds of SRSC and CRSC, we obtain the theoretical optimal choice of the regularization parameter $\tau$ as $O(\mathrm{log}(\rho n))$, where $\rho$ is the sparsity parameter. Especially, for the sparest network, the theoretical and empirical optimal choice of $\tau$ is $O(\mathrm{log}(n))$.
	\item Since the parametric probability is closely related with the sparsity of a network, under mild conditions, we obtain the optimal regularization parameter for the sparse network with the order of the maximum excepted degree as small as $log(n)$. With the help of  the parametric probability,  it is easy to comprehend the trade-offs between the sparsity of a network and the probability of successfully detecting mixed memberships under MMSB. Meanwhile, the parametric probability is also useful in explaining the conclusion in \cite{joseph2016impact} that a large regularizer may lead to good results but  a moderate regularization parameter is preferred.
\end{itemize}

The following notations will be used throughout the paper: $\|\cdot\|_{F}$ for a matrix denotes the Frobenius norm,  $\|\cdot\|$ for a matrix denotes the spectral norm, $\|\cdot\|_{1}$ for a vector denotes the $l_{1}$ norm and $|C|$ means the absolute value of number $C$. For any matrix $X$, $\|X\|_{2\rightarrow\infty}$ denotes the maximum $l_{2}$-norm of all the rows of $X$, and $\|X\|_{\infty}=\mathrm{max}_{i}\sum_{j}|X(i,j)|$. For any matrix $X$, set the matrix $\mathrm{max}(X,0)$ such that its $(i,j)$-th entry is $\mathrm{max}(X(i,j),0)$. For convenience, when we say ``leading eigenvalues'' or ``leading eigenvectors'', we are comparing the \emph{magnitudes} of the eigenvalues and their respective eigenvectors with unit-norm. For two positive sequences $\{a_{n}\}$ and $\{b_{n}\}$. We say $a_{n}\asymp b_{n}$ if there are two constants $c_{2}>c_{1}>0$ such that $c_{1}a_{n}\leq b_{n}\leq c_{2}a_{n}$. For any matrix or vector $X$, $X'$ denotes the transpose of $X$. Unless specified, let $\lambda_{k}(X)$ denote the $k$-th leading eigenvalue of the matrix $X$. $X(i,:)$ and $X(:,j)$  denote the $i$-th row and the $j$-th column of matrix $X$, respectively. $X(S_{r},:)$ and $X(:,S_{c})$ denote the rows and columns in the index sets $S_{r}$ and $S_{c}$ of matrix $X$, respectively. For any vector $x$, we use $x_{i}$ or $x(i)$ to denote the $i$-th entry of it occasionally. For any matrix $X\in\mathbb{R}^{m\times m}$, let $\mathrm{diag}(X)$ be the $m\times m$ diagonal matrix whose $i$-th diagonal entry is $X(i,i)$. $\mathbf{1}$ and $\mathbf{0}$ are  column vectors with all entries being ones and zeros, respectively.  $e_{i}$ is a column vector whose $i$-th entry is 1 while other entries are zero. In this paper, $C$ is a positive constant which may be different occasionally.

\section{Mixed Membership Stochasticblock Model}\label{sec2}
Consider an undirected and unweighted network $\mathcal{N}$ and assume that there are $K$ disjoint blocks $V^{(1)}, V^{(2)}, \ldots, V^{(K)}$ where $K$ is assumed to be known in this paper. Let $A$ be its adjacency matrix such that $A_{ij}=1$ if there is an edge between node $i$ and $j$, $A_{ij}=0$ otherwise. 

The mixed membership stochasticblock (MMSB) model \citep{MMSB} allows us to measure the probability of that each node belongs to a certain community. It is assumed that  each node $i$ belongs to cluster $V^{(k)}$ with probability $\pi_{i}(k)$ and $\sum_{k=1}^{K}\pi_{i}(k)=1$, i.e., there is a Probability Mass Function (PMF) $\pi_{i}=(\pi_{i}(1), \pi_{i}(2), \ldots, \pi_{i}(K))$ such that
\begin{align*}
\mathrm{Pr}(i\in V^{(k)})=\pi_{i}(k), \qquad 1\leq k\leq K, 1\leq i\leq n.
\end{align*}
We call node $i$ ``pure" if $\pi_i$ is degenerate such that there is one element of $\pi_i$ is 1, and the remaining $K-1$ entries are 0; and call node $i$ ``mixed" otherwise. Furthermore,  $\underset{1\leq k\leq K}{\mathrm{max}}\pi_{i}(k)$ is taken as the \textit{purity} of node $i$, for $1\leq i\leq n$.
For mixed membership community detection, the main aim is to  estimate $\pi_{i}$ for all nodes $i\in\{1,2,\ldots, n\}$.

For any fixed pair of  $(i,j)$, MMSB assumes that
\begin{align*}
\mathrm{Pr}(A(i,j)=1|i\in V^{(k)}, j\in V^{(l)})=\rho \tilde{P}(k,l),
\end{align*}
where $\tilde{P}$ is a  $K\times K$ symmetric non-negative, non-singular and irreducible matrix, $\mathrm{max}_{1\leq i,j\leq K}\tilde{P}(i,j)=1$, and the parameter $\rho$ controls the sparsity of the generated network. 
For convenience, set $P=\rho \tilde{P}$. This model assumption indicates  that when we know $i\in V^{(k)}$ and $j\in V^{(l)}$, the probability that there is an edge between nodes $i$ and $j$ is $P(k,l)$. For $1\leq i<j\leq n$, $A(i,j)$ are independent  Bernoulli random variables, satisfying
\begin{align}\label{A}
\mathrm{Pr}(A(i,j)=1)=\sum_{k=1}^{K}\sum_{l=1}^{K}\pi_{i}(k)\pi_{j}(l)P(k,l).
\end{align}
Let $E[A]=\Omega$ such that $\Omega(i,j)=\mathrm{Pr}(A(i,j)=1), 1\leq i<j\leq n$, then we have
\begin{align}\label{Omega}
\Omega=\Pi P \Pi',
\end{align}
where $\Pi$  is an $n\times K$ membership matrix such that the $i$-th row of $\Pi$ (denoted as $\Pi(i,:)$) is $\pi_{i}$ for all $i\in \{1,2,\ldots, n\}$.

Given $(n,P, \Pi)$, we can generate the random adjacency matrix $A$ under MMSB, hence we denote the MMSB model as $MMSB(n,P, \Pi)$ for convenience in this paper. The primary goal for mixed membership community detection is to estimate the membership matrix $\Pi$ with given $(A, K)$.

As studied in \cite{mao2020estimating}, to make the model identifiable, in this paper we assume that
\begin{itemize}
	\item (I1) $\mathrm{rank}(P)=K$.
	\item (I2) Each community has at least one pure node.
\end{itemize}
For convenience, in this article, we treat the two conditions as default.

\section{Methodologies}\label{sec3}
In this section, to design algorithms designed based on the regularized Laplacian matrix for mixed membership community detection problem, we start by the oracle case where $\Omega$ is given, and then we extend what we have in the oracle case to the empirical case.

We start with introducing the population regularized Laplacian matrix:
\begin{align}\label{populationL} \mathscr{L}_{\tau}=\mathscr{D}^{-1/2}_{\tau}\Omega\mathscr{D}^{-1/2}_{\tau},
\end{align}
where $\mathscr{D}_{\tau}=\mathscr{D}+\tau I$,  $\mathscr{D}$ is an $n\times n$ diagonal matrix whose $i$-th diagonal entry is $\mathscr{D}(i,i)=\sum_{j=1}^{n}\Omega(i,j)$, and $\tau$ is a nonnegative regularizer.  By (\ref{Omega}), we have  $\mathscr{L}_{\tau}=\mathscr{D}^{-1/2}_{\tau}\Pi P\Pi'\mathscr{D}^{-1/2}_{\tau}$.
By basic algebra, we have the rank of $\mathscr{L}_{\tau}$ is $K$, thus $\mathscr{L}_{\tau}$ has $K$ nonzero eigenvalues. Denote $\{\lambda_{i},\eta_{i}\}_{i=1}^{K}$ as the leading $K$ eigenvalues and their respective eigenvectors with unit-norm.

In next two subsections, we will give two ideal algorithms based on  properties of the population regularized Laplacian matrix.

\subsection{The Ideal Simplex (IS) and the Ideal SRSC algorithm}
By studying the eigenvalue decomposition of $\mathscr{L}_{\tau}$, we have the following lemma which guarantees the existence of the Ideal Simplex (to be defined later). 
\begin{lem}\label{IdealSimplex}
	Under $MMSB(n,P, \Pi)$, let $\mathscr{L}_{\tau}=VEV'$ be the compact eigenvalue decomposition of $\mathscr{L}_{\tau}$ such that $V=[\eta_{1},\eta_{2}, \ldots, \eta_{K}], E=\mathrm{diag}(\lambda_{1}, \lambda_{2}, \ldots, \lambda_{K})$ and $V'V=I$.
	Set $V_{\tau,1}=\mathscr{D}^{1/2}_{\tau}V$, we have $V_{\tau,1}=\Pi V_{\tau,1}(\mathcal{I},:)$, where $\mathcal{I}$ is the indices of rows corresponding to $K$ pure nodes, one from each community.  Meanwhile, for any two distinct nodes $i,j$, we have $V_{\tau,1}(i,:)=V_{\tau,1}(j,:)$ when $\Pi(i,:)=\Pi(j,:)$.
\end{lem}
\begin{rem}\label{IdifferentButVtauSame}
	Though the index set $\mathcal{I}$ may be various since we can choose different nodes from a certain cluster, $V_{\tau,1}(\mathcal{I},:)$ is always the same due to the fact that $V_{\tau,1}(i,:)=V_{\tau,1}(j,:)$ if pure nodes $i$ and $j$ come from the same cluster.
\end{rem}
Let $v_{1}, v_{2}, \ldots, v_{K}$ be the $K$ rows of $V_{\tau,1}(\mathcal{I},:)$. By  the form $V_{\tau,1}=\Pi V_{\tau,1}(\mathcal{I},:)$, we can find that the rows of $V_{\tau}$ form a $K$-simplex in $\mathbb{R}^{K}$ which we call the Ideal Simplex (IS), with $v_{1}, v_{2}, \ldots, v_{K}$ being the vertices. Denoting the simplex by $\mathcal{S}^{\mathrm{ideal}}(v_{1}, v_{2}, \ldots, v_{K})$, by Lemma \ref{IdealSimplex}, we have
\begin{itemize}
	\item [(1)] Each row $V_{\tau,1}(i,:)$ is a convex linear combination of $v_{1}, \ldots, v_{K}$ such that
	\begin{align*}
	V_{\tau,1}(i,:)=\sum_{k=1}^{K}\Pi(i,k)v_{k}.
	\end{align*}
	\item [(2)] A pure row (row $i$ of $V_{\tau,1}$ is pure if node $i$ is pure and is mixed otherwise) falls on one of the $K$ vertices of $\mathcal{S}^{\mathrm{ideal}}(v_{1}, v_{2}, \ldots, v_{K})$, and a mixed row falls in the interior of $\mathcal{S}^{\mathrm{ideal}}(v_{1}, v_{2}, \ldots, v_{K})$.
\end{itemize}
Since the $K$ rows of $V_{\tau,1}(\mathcal{I},:)$ are the $K$ vertices of the simplex, we call $V_{\tau,1}(\mathcal{I},:)$ as the corner matrix for convenience.

In fact, \cite{mixedSCORE} and \cite{mao2020estimating} also showed the existence of the ideal simplex based on the adjacency matrix. However, in this paper, the ideal simplex is constructed based on the population regularized Laplacian matrix.

By conditions (I1) and (I2), we have $\mathrm{rank}(P)=K$ and $ \mathrm{rank}(\Pi)=K$, which give that $\mathrm{rank}(V_{\tau,1}(\mathcal{I},:))=K$. Since $V_{\tau,1}(\mathcal{I},:)\in\mathbb{R}^{K\times K}$, we see $V_{\tau,1}(\mathcal{I},:)$ is an non-singular matrix. Then by Lemma \ref{IdealSimplex}, we have $\Pi=V_{\tau,1}V^{-1}_{\tau,1}(\mathcal{I},:)$. Since $V_{\tau,1}=\mathscr{D}^{1/2}_{\tau}V$, we have
\begin{align}\label{spV1}
Z_{1}\equiv\mathscr{D}^{-1/2}_{\tau}\Pi=VV^{-1}_{\tau,1}(\mathcal{I},:).
\end{align}
As $\mathscr{D}^{1/2}_{\tau}$ is a diagonal matrix, we can obtain that $\Pi(i,:)=\frac{Z_{1}(i,:)}{\|Z_{1}(i,:)\|_{1}}$. Therefore, if $\Pi$ and $P$ are unknown but $\Omega$ and $K$ are given, then we can compute $\mathscr{L}_{\tau}, V$ and $V_{\tau,1}$, thus we can obtain $\Pi$ by normalizing each rows of $Z_{1}$ to have unit $l_{1}$ norm, as long as we can find the index set $\mathcal{I}$.  
Hereafter, the only difficulty is in finding $\mathcal{I}$.  The successive projection (SP) algorithm \cite{gillis2015semidefinite} (see Algorithm SP in the supplementary material for detail) can be applied to the Ideal Simplex to find an index set. 

The above analysis gives rise to the following three-stage algorithm which we call Ideal Simplex Regularized Spectral Clustering (Ideal SRSC for short). Input: $\Omega, K$. Output: $\Pi$.
\begin{itemize}
	\item \texttt{RSC step}.
	\begin{itemize}
		\item Obtain $\mathscr{D}_{\tau}=\mathscr{D}+\tau I$.
		\item Obtain $\mathscr{L}_{\tau}$ such that $\mathscr{L}_{\tau}=\mathscr{D}^{-1/2}_{\tau}\Omega\mathscr{D}^{-1/2}_{\tau}$ and let $V\in\mathbb{R}^{n\times K}$ be the matrix  of the leading $K$ eigenvectors with unit-norm of $\mathscr{L}_{\tau}$.
		\item Obtain $V_{\tau,1}$ such that $V_{\tau,1}=\mathscr{D}^{1/2}_{\tau}V$.
	\end{itemize}
	\item \texttt{Corners Hunting (CH) step.}
	\begin{itemize}
		\item Run SP algorithm with inputs $V_{\tau,1}$ and $K$ to obtain the corner matrix $V_{\tau,1}(\mathcal{I},:)$.
	\end{itemize}
	\item \texttt{Membership Reconstruction (MR) step.}
	\begin{itemize}
		\item Recover $Z_{1}$ by setting $Z_{1}=VV^{-1}_{\tau,1}(\mathcal{I},:)$.
		\item Recover $\Pi(i,:)$ by setting $\Pi(i,:)=\frac{Z_{1}(i,:)}{\|Z_{1}(i,:)\|_{1}}$ for $1\leq i\leq n$.
	\end{itemize}
\end{itemize}
The above analysis shows that the Ideal SRSC exactly recovers the membership matrix $\Pi$.

To demonstrate that $V_{\tau,1}$ has the ideal simplex structure, we drew panel (a) of Figure \ref{PlotVstar} when $K=3$. Panel (a) of Figure \ref{PlotVstar} shows that all mixed rows of $V_{\tau}$ are located  inside of the simplex  formed by the $K$ pure rows of $V_{\tau,1}$.  Meanwhile, the SP algorithm can exactly return the corner matrix $V_{\tau,1}(\mathcal{I},:)$ from $V_{\tau,1}$, for detailed explanation of this statement, refer to Remark 10 in the supplementary material.
The data used for panel (a) is generated from MMSB with $n=800, K=3$. Among the 800 nodes, 600 are pure nodes with each cluster has 200 pure nodes. For node $j$ among the 200 mixed nodes, we set $\Pi(j,1)=\mathrm{rand}(1)/2, \Pi(j,2)=\mathrm{rand}(1)/2, \Pi(j,3)=1-\Pi(j,1)-\Pi(j,2)$ where $\mathrm{rand}(1)$ is any random number in $(0,1)$. The matrix $P$ is a symmetric matrix with diagonal entries 0.8, others are 0.1.
Then based on the above setting, we can obtain $V_{\tau,1}$ which is demonstrate in panel (a) of Figure \ref{PlotVstar}.
\begin{figure}
	\centering
	\subfigure[$V_{\tau,1}$]{\includegraphics[width=0.4\textwidth]{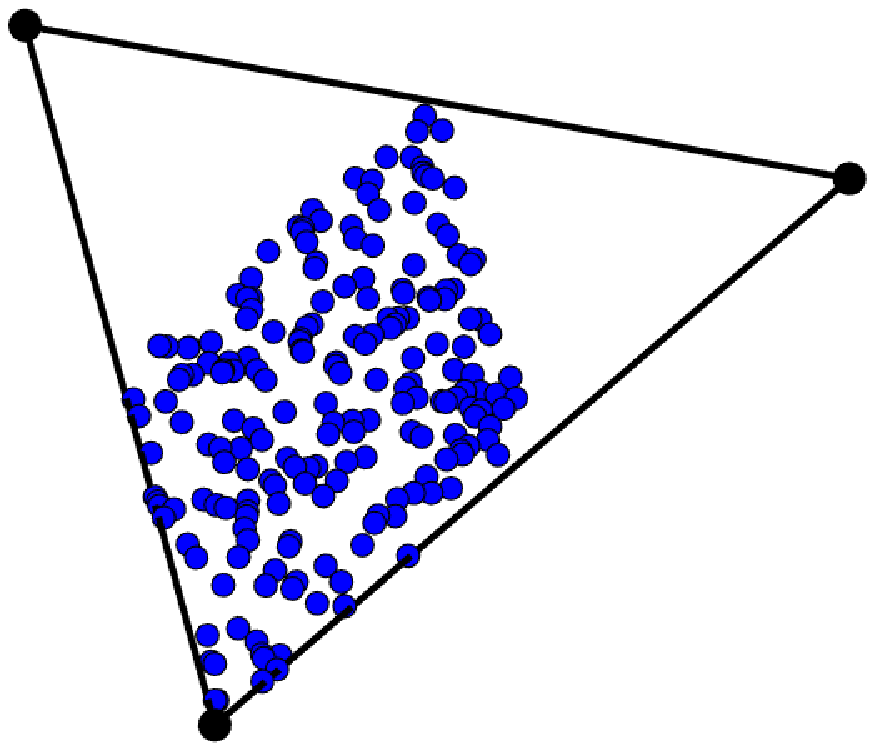}}	\subfigure[$V_{*,1}$]{\includegraphics[width=0.4\textwidth]{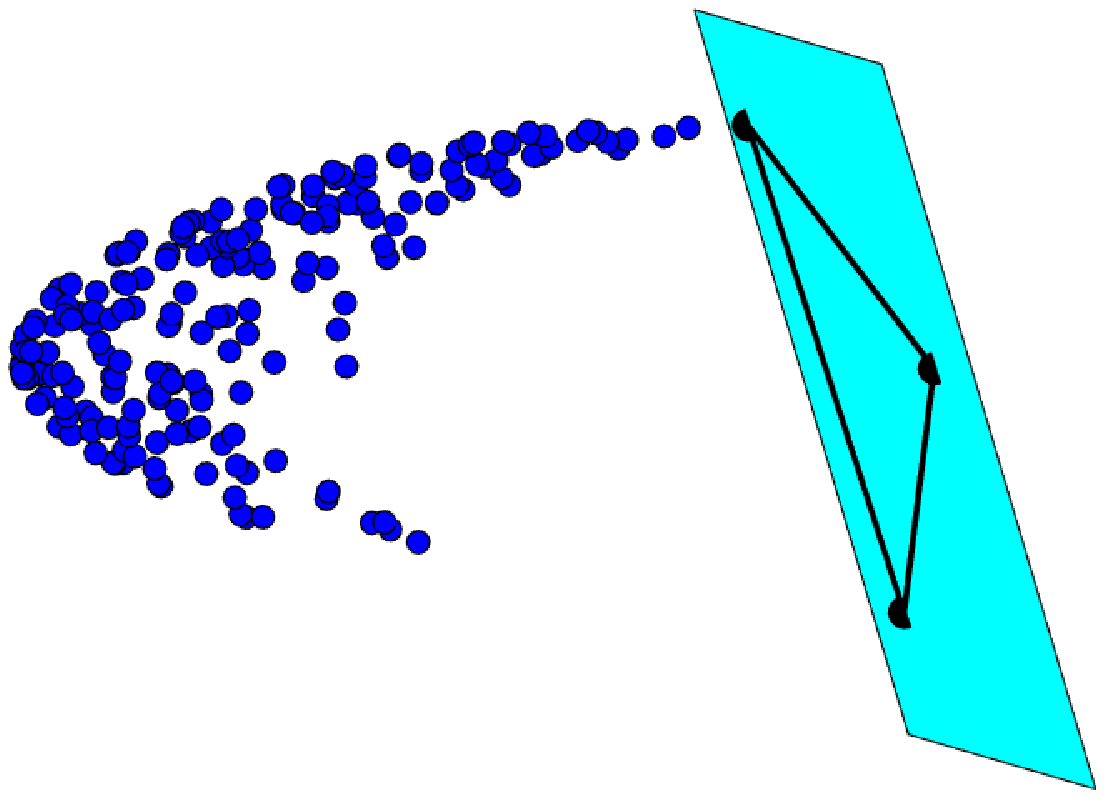}}
	\caption{Panel a: plot of $V_{\tau,1}$ and the ideal simplex formed by $V_{\tau,1}(\mathcal{I},:)$. Blue points denote mixed rows of $V_{\tau,1}$; Black points denote the $K$ rows of the corner matrix $V_{\tau,1}(\mathcal{I},:)$ returned by SP algorithm (please refer to the supplementary material) when the input is $V_{\tau,1}$; Note that by Lemma \ref{IdealSimplex}, rows for pure nodes are same if these pure nodes are from the same cluster, hence rows for pure nodes coincide if these pure nodes are from the same cluster. Panel b: plot of $V_{*,1}$ and the hyperplane formed by $V_{*,1}(\mathcal{I},:)$. Blue points denote mixed rows of $V_{*,1}$; Black points denote the $K$ rows of the corner matrix $V_{*,1}(\mathcal{I},:)$  returned by SVM-cone algorithm (please refer to the supplementary material) when the input is $V_{*,1}$; We also plot the hyperplane formed by the triangle of the 3 rows of $V_{*,1}(\mathcal{I},:)$. Note that by Lemma \ref{IdealCone}, rows respective to pure nodes are same if these pure nodes are from the same cluster, hence rows refer to pure nodes from the same cluster coincide in $\mathbb{R}^{3}$ in this figure. For visualization, we have projected and rotated these points from $\mathbb{R}^{3}$ to $\mathbb{R}^{2}$.}
	\label{PlotVstar}
\end{figure}

\subsection{The Ideal Cone (IC) and the Ideal CRSC algorithm}
In this subsection, we give another ideal algorithm.
Actually, we  normalize each rows of $V$ to have unit $l_{2}$ length, then the newly obtained matrix has a structure called Ideal Cone (to be defined later). Then the SVM-cone algorithm \citep{MaoSVM}  can be applied to hunt for the index set $\mathcal{I}$.  

Let $V_{*,1}$ be the row-normalized version of $V$ such that $V_{*,1}(i,:)=\frac{V(i,:)}{\|V(i,:)\|_{F}}$. Let $N_{V}$ be the $n\times n$ diagonal matrix such that $N_{V}(i,i)=\frac{1}{\|V(i,:)\|_{F}}$ for $1\leq i\leq n$. Then $V_{*,1}$ can be rewritten as $V_{*,1}=N_{V}V$. Next lemma shows that each row of $V_{*,1}$ can be expressed by a scaled combination of $V_{*,1}(\mathcal{I},:)$. Combining it with the fact that each row of $V_{*,1}$ has unit $l_{2}$ norm, the existence of the Ideal Cone is guaranteed.
\begin{lem}\label{IdealCone}
	Under $MMSB(n, P, \Pi)$, there exists a $Y_{1}\in\mathbb{R}^{n\times K}_{\geq 0}$ and no row of $Y_{1}$ is 0 such that
	\begin{align*}
	V_{*,1}=Y_{1}V_{*,1}(\mathcal{I},:),
	\end{align*}
	where $Y_{1}$ can be written as $Y_{1}=N_{M_{1}}\Pi\mathscr{D}^{1/2}_{\tau}(\mathcal{I},\mathcal{I})N_{V}^{-1}(\mathcal{I},\mathcal{I})$, where $N_{M_{1}}$ is an $n\times n$ diagonal matrix whose diagonal entries are positive. Meanwhile, for any two distinct nodes $i,j$, when $\Pi(i,:)=\Pi(j,:)$, we have $V_{*,1}(i,:)=V_{*,1}(j,:)$.
\end{lem}
Since $\mathrm{rank}(V_{*,1})=K$, $\mathrm{rank}(V_{*,1}(\mathcal{I},:))=K$. As $V_{*,1}(\mathcal{I},:)\in\mathbb{R}^{K\times K}$, the inverse of $V_{*,1}(\mathcal{I},:)$ exists. Therefore, Lemma \ref{IdealCone} also gives that
\begin{align}\label{Y1}
Y_{1}=V_{*,1}V^{-1}_{*,1}(\mathcal{I},:).
\end{align}
Since $V_{*,1}=N_{V}V$ and $ Y_{1}=N_{M_{1}}\Pi\mathscr{D}^{1/2}_{\tau}(\mathcal{I},\mathcal{I})N_{V}^{-1}(\mathcal{I},\mathcal{I})$, we have $$
N_{V}^{-1}N_{M_{1}}\Pi\mathscr{D}^{1/2}_{\tau}(\mathcal{I},\mathcal{I})N_{V}^{-1}(\mathcal{I},\mathcal{I})=VV^{-1}_{*,1}(\mathcal{I},:),$$ i.e.,
\begin{align}\label{ZYJ1}
N_{V}^{-1}N_{M_{1}}\Pi=VV^{-1}_{*,1}(\mathcal{I},:)N_{V}(\mathcal{I},\mathcal{I})\mathscr{D}^{-1/2}_{\tau}(\mathcal{I},\mathcal{I}).
\end{align}
For convenience, set  $J_{*,1}=N_{V}(\mathcal{I},\mathcal{I})\mathscr{D}^{-1/2}_{\tau}(\mathcal{I},\mathcal{I}),
Z_{*,1}=N_{V}^{-1}N_{M_{1}}\Pi, Y_{*,1}=VV^{-1}_{*,1}(\mathcal{I},:)$. By Eq (\ref{ZYJ1}), we have
\begin{align}\label{Z1}
Z_{*,1}=Y_{*,1}J_{*,1}\equiv VV^{-1}_{*,1}(\mathcal{I},:)N_{V}(\mathcal{I},\mathcal{I})\mathscr{D}^{-1/2}_{\tau}(\mathcal{I},\mathcal{I}).
\end{align}
Meanwhile, since $N_{V}^{-1}N_{M_{1}}$ is an $n\times n$ positive diagonal matrix, we have
\begin{align}\label{Pi}
\Pi(i,:)=\frac{Z_{*,1}(i,:)}{\|Z_{*,1}(i,:)\|_{1}}, 1\leq i\leq n.
\end{align}
The above analysis shows that once the index set $\mathcal{I}$ is known, we can exactly recover $\Pi$ by Eq. (\ref{Pi}).

Thus, the only difficulty is in finding the index set $\mathcal{I}$. From Lemma \ref{IdealCone}, we know that $V_{*,1}=Y_{1}V_{*,1}(\mathcal{I},:)$ forms the Ideal Cone.  The SVM-cone algorithm can be used to obtain the corner indices set $\mathcal{I}_{SVM-cone}$ from the Ideal Cone. And the condition for using SVM-cone is satisfied, i.e.,  $(V_{*,1}(\mathcal{I},:)V'_{*,1}(\mathcal{I},:))^{-1}\mathbf{1}>0$ holds (see Lemma \ref{LSVM}).
Though $\mathcal{I}_{SVM-cone}$ may differ from $\mathcal{I}$, $\Pi(\mathcal{I}_{SVM-cone}(k),:)=\Pi(\mathcal{I}(k),:)$ for $1\leq k\leq K$, see the supplementary material. Hence, we also use $\mathcal{I}$ to denote $\mathcal{I}_{SVM-cone}$.  
\begin{lem}\label{LSVM}
	Under $DCMM(n, P, \Theta, \Pi)$, $(V_{*,1}(\mathcal{I},:)V'_{*,1}(\mathcal{I},:))^{-1}\mathbf{1}>0$ holds.
\end{lem}
From the above analysis we construct the following algorithm called  Ideal Cone Regularized Spectral Clustering (Ideal CRSC for short). Input $\Omega, K$. Output: $\Pi$.
\begin{itemize}
	\item \texttt{RSC step}.
	\begin{itemize}
		\item Obtain $\mathscr{D}_{\tau}, \mathscr{L}_{\tau}$, $V$, $V_{*,1}$ and $N_{V}$.
	\end{itemize}
	\item \texttt{Corners Hunting (CH) step}.
	\begin{itemize}
		\item Run SVM-cone algorithm with inputs $V_{*,1}$ and $K$ to obtain the corner matrix $V_{*,1}(\mathcal{I},:)$.
	\end{itemize}
	\item \texttt{Membership Reconstruction (MR) step}.
	\begin{itemize}
		\item Recover $Y_{*,1}$ and $J_{*,1}$ by setting
		$Y_{*,1}=VV^{-1}_{*,1}(\mathcal{I},:), J_{*,1}=N_{V}(\mathcal{I},\mathcal{I})\mathscr{D}^{-1/2}_{\tau}(\mathcal{I},\mathcal{I})$.
		\item Recover $Z_{*,1}$ by setting $Z_{*,1}=Y_{*,1}J_{*,1}$.
		\item Recover $\Pi(i,:)$ by setting
		$\Pi(i,:)=\frac{Z_{*,1}(i,:)}{\|Z_{*,1}(i,:)\|_{1}}$ for $1\leq i\leq n$.
	\end{itemize}
\end{itemize}

To demonstrate that $V_{*,1}$ has the ideal cone structure, we drew panel (b) of Figure \ref{PlotVstar}, where panel (b) is obtained under the same setting as panel (a) (i.e., after computing $\mathscr{L}_{\tau}$, then obtain $V_{*,1}$. Run SVM-cone algorithm on $V_{*,1}$ with $K=3$ to obtain the index set $\mathcal{I}$, then we can plot panel (b) of Figure \ref{PlotVstar}.). Panel (b) shows that all mixed rows of $V_{*,1}$ are located  at one side of the hyperplane formed by the $K$ pure rows of $V_{*,1}$.  Meanwhile, the SVM-cone algorithm can exactly return the corner matrix $V_{*,1}(\mathcal{I},:)$ from $V_{*,1}$ with given $K$, for detailed explanation of this statement, refer to the supplementary material. 

\subsection{The algorithms: SRSC and CRSC}
We now extend the ideal case to the real case. The following two algorithms, which we call Simplex Regularized Spectral Clustering (SRSC for short) and Cone Regularized Spectral Clustering (CRSC for short) are natural extensions of the Ideal SRSC and the Ideal CRSC, respectively.
\begin{algorithm}
	\caption{\textbf{Simplex Regularized Spectral Clustering (SRSC for short)}}
	\label{alg:SRSC}
	\begin{algorithmic}[1]
		\Require The adjacency matrix $A\in \mathbb{R}^{n\times n}$, the number of communities $K$, and a ridge regularizer $\tau\geq 0$.
		\Ensure The estimated $n\times K$ membership matrix $\hat{\Pi}_{1}$.
		\State \texttt{RSC step.}
		\begin{itemize}
			\item Obtain the graph Laplacian with ridge regularization by
			\begin{align*}
			L_{\tau}=D_{\tau}^{-1/2}AD_{\tau}^{-1/2},
			\end{align*}
			where $D_{\tau}=D+\tau I$, $D$ is an $n\times n$ diagonal matrix whose  $i$-th diagonal entry is $D(i,i)=\sum_{j=1}^{n}A(i,j)$ (unless specified, for SRSC, a good default $\tau$ is $\tau=0.1\mathrm{log}(n)$).
			\item Let $\hat{V}=[\hat{\eta}_{1}, \ldots,\hat{\eta}_{K}]\in\mathbb{R}^{n\times K}$ denote the matrix containing the leading $K$ eigenvectors with unit-norm of $L_{\tau}$.
			\item Let $\hat{V}_{\tau,1}=D^{1/2}_{\tau}\hat{V}$.
		\end{itemize}
		\State \texttt{CH step.}
		\begin{itemize}
			\item Apply SP algorithm on the rows of $\hat{V}_{\tau,1}$ assuming there are $K$ clusters to obtain the near-corners matrix $\hat{V}_{\tau,1}(\mathcal{\hat{I}}_{1},:)\in\mathbb{R}^{K\times K}$, where $\mathcal{\hat{I}}_{1}$ is the index set returned by SP algorithm.
		\end{itemize}
		\State \texttt{Membership Reconstruction (MR) step.}
		\begin{itemize}
			\item Estimate $Z_{1}$ by setting $\hat{Z}_{1}=\hat{V}\hat{V}^{-1}_{\tau,1}(\mathcal{\hat{I}}_{1},:)$.
			\item Set $\hat{Z}_{1}=\mathrm{max}(0,\hat{Z}_{1})$
			\item Estimate $\Pi(i,:)$ by setting $\hat{\Pi}_{1}(i,:)=\hat{Z}_{1}(i,:)/\|\hat{Z}_{1}(i,:)\|_{1}, 1\leq i\leq n$.
		\end{itemize}
	\end{algorithmic}
\end{algorithm}

\begin{algorithm}
	\caption{\textbf{Cone Regularized Spectral Clustering (CRSC for short)}}
	\label{alg:CRSC}
	\begin{algorithmic}[1]
		\Require The adjacency matrix $A\in \mathbb{R}^{n\times n}$, the number of communities $K$, and a ridge regularizer $\tau\geq 0$.
		\Ensure The estimated $n\times K$ membership matrix $\hat{\Pi}_{*,1}$.
		\State \texttt{RSC step.}
		\begin{itemize}
			\item Obtain $L_{\tau}, D_{\tau}, \hat{V}$ as Algorithm \ref{alg:SRSC} (unless specified, for CRSC, a good default $\tau$ is $\tau=0.1\mathrm{log}(n)$). Let $\hat{V}_{*,1}$ be the $n\times K$ matrix  such that $\hat{V}_{*,1}(i,:)=\frac{\hat{V}(i,:)}{\|\hat{V}(i,:)\|_{F}}, i=1,2,\ldots, n$. Obtain the  $n\times n$ diagonal matrix $N_{\hat{V}}$, whose $i$-th diagonal entry is $1/\|\hat{V}(i,:)\|_{F}$.
		\end{itemize}
		\State \texttt{CH step.}
		\begin{itemize}
			\item Apply SVM-cone algorithm on the rows of $\hat{V}_{*,1}$ assuming there are $K$ clusters to obtain the estimated index set $\mathcal{\hat{I}}_{*,1}$.
		\end{itemize}
		\State \texttt{Membership Reconstruction (MR) step.}
		\begin{itemize}
			\item Estimate $Y_{*,1}$ by setting $\hat{Y}_{*,1}=\hat{V}\hat{V}^{-1}_{*,1}(\mathcal{\hat{I}}_{*,1},:)$.
			\item Estimate $J_{*,1}$ by setting $\hat{J}_{*,1}=N_{\hat{V}}(\mathcal{\hat{I}}_{*,1},\mathcal{\hat{I}}_{*,1})D^{-1/2}_{\tau}(\mathcal{\hat{I}}_{*,1},\mathcal{\hat{I}}_{*,1})$.
			\item Estimate $Z_{*,1}$ by setting $\hat{Z}_{*,1}=\hat{Y}_{*,1}\hat{J}_{*,1}$.
			\item Set $\hat{Z}_{*,1}=\mathrm{max}(0,\hat{Z}_{*,1})$
			\item Estimate $\Pi(i,:)$ by setting $\hat{\Pi}_{*,1}(i,:)=\hat{Z}_{*,1}(i,:)/\|\hat{Z}_{*,1}(i,:)\|_{1}, 1\leq i\leq n$.
		\end{itemize}
	\end{algorithmic}
\end{algorithm}

\begin{rem}
	Steps 1, 2, 3  in  SRSC and CRSC are straightforward extensions of the three steps in the Ideal SRSC and the Ideal CRSC except that we set $\hat{Z}_{1}=\mathrm{max}(0,\hat{Z}_{1})$ and $\hat{Z}_{*,1}=\mathrm{max}(0,\hat{Z}_{*,1})$ to transform negative entries of $\hat{Z}_{1}, \hat{Z}_{*,1}$ into positive in the MR step due to the fact that $\hat{V}\hat{V}^{-1}_{\tau}(\mathcal{\hat{I}}_{1},:)$ and $\hat{Y}_{*,1}\hat{J}_{*,1}$ may contain a few negative entries in practice and we have to remore these negative entries since weights are nonnegative.
\end{rem}

\section{Equivalence algorithms}
In this section, we design two algorithms SRSC-equivalence and CRSC-equivalence which give same estimations as Algorithms \ref{alg:SRSC} and \ref{alg:CRSC}, respectively. We start this section by defining eight $n\times n$ matrices: $V_{2}, \hat{V}_{2}, V_{\tau,2}, \hat{V}_{\tau,2}, V_{*,2}, \hat{V}_{2,*}, N_{V_{2}}$ and $N_{\hat{V}_{2}}$.
\begin{defin}
	Set $V_{2}=VV', V_{\tau,2}=\mathscr{D}^{1/2}_{\tau}V_{2}, \hat{V}_{2}=\hat{V}\hat{V}', \hat{V}_{\tau,2}=D^{1/2}_{\tau}\hat{V}_{2}$.
	Set $V_{*,2}, \hat{V}_{*,2}$  such that $V_{*,2}(i,:)=\frac{V_{2}(i,:)}{\|V_{2}(i,:)\|_{F}}, \hat{V}_{*,2}(i,:)=\frac{\hat{V}_{2}(i,:)}{\|\hat{V}_{2}(i,:)\|_{F}}$ for $1\leq i\leq n$. Let $N_{V_{2}}, N_{\hat{V}_{2}}$ be $n\times n$ diagonal matrices whose $i$-th diagonal entries are $\frac{1}{\|V_{2}(i,:)\|_{F}}$ and  $\frac{1}{\|\hat{V}_{2}(i,:)\|_{F}}$, respectively.
\end{defin}
In next subsections, we will give the two equivalences algorithms after providing the Ideal SRSC-equivalence algorithm and the Ideal CRSC-equivalence algorithm based on analyzing the properties of $V_{\tau,2}$ and $V_{*,2}$.
\subsection{The SRSC-equivalence algorithm}
To introduce the SRSC-equivalence algorithm, similar as the SRSC algorithm, we start from the ideal case. First we show that there exists the Ideal Simplex structure in $V_{\tau,2}$ by Lemma \ref{IdealSimplex2}.
\begin{lem}\label{IdealSimplex2}
	Under $MMSB(n, P, \Pi)$, we have $V_{\tau,2}=\Pi V_{\tau,2}(\mathcal{I},:)$. Meanwhile, for any two distinct nodes $i,j$, we have $V_{\tau,2}(i,:)=V_{\tau,2}(j,:)$ when $\Pi(i,:)=\Pi(j,:)$.
\end{lem}
\begin{rem}\label{IdifferentButVtauSame2}
	Similar as Remark \ref{IdifferentButVtauSame}, though the index set $\mathcal{I}$ may be different, $V_{\tau,2}(\mathcal{I},:)$ is always the same.
\end{rem}
Since $V_{\tau,2}(\mathcal{I},:)\in\mathbb{R}^{K\times n}$, $V_{\tau,2}(\mathcal{I},:)$ is a singular matrix. Based on conditions (I1) and (I2), $V_{\tau,2}(\mathcal{I},:)V'_{\tau,2}(\mathcal{I},:)$ is non-singular. By Lemma \ref{IdealSimplex2}, we have $\Pi=V_{\tau,2}V'_{\tau,2}(\mathcal{I},:)(V_{\tau,2}(\mathcal{I},:)V'_{\tau,2}(\mathcal{I},:))^{-1}$. Since $V_{\tau,2}=\mathscr{D}^{1/2}_{\tau}V_{2}$, we have
\begin{align}\label{spV2}
Z_{2}\equiv\mathscr{D}^{-1/2}_{\tau}\Pi=V_{2}V'_{\tau,2}(\mathcal{I},:)(V_{\tau,2}(\mathcal{I},:)V'_{\tau,2}(\mathcal{I},:))^{-1},
\end{align}
where we set $Z_{2}=\mathscr{D}^{-1/2}_{\tau}\Pi$ for convenience. Then we have $\Pi(i,:)=\frac{Z_{2}(i,:)}{\|Z_{2}(i,:)\|_{1}}$. Now, if we are given $\Omega$ and $K$ in advance but without known $\Pi$ and $P$, then we can compute $\mathscr{L}_{\tau}, V_{2}$ and $V_{\tau,2}$.  According to Eq (\ref{spV2}) and Remark \ref{IdifferentButVtauSame2}, as long as we know the corner matrix $V_{\tau,2}(\mathcal{I},:)$, we can obtain $\Pi$ by normalizing each rows of $Z_{2}$ to have unit $l_{1}$ norm. Thus, the only difficulty is in finding $\mathcal{I}$. Similar as the Ideal SRSC algorithm, due to the Ideal Simplex form $V_{\tau,2}=\Pi V_{\tau,2}(\mathcal{I},:)$, with given $V_{\tau,2}$ and $K$, SP algorithm can find the corner matrix $V_{\tau,2}(I,:)$.

The above analysis gives rise to the following three-stage algorithm which we call the Ideal SRSC-equivalence algorithm. Input $\Omega, K$. Output: $\Pi$.
\begin{itemize}
	\item \texttt{RSC step}.
	\begin{itemize}
		\item Obtain $\mathscr{D}_{\tau}=\mathscr{D}+\tau I$, $\mathscr{L}_{\tau}$, $V$ and $V_{\tau,2}$.
	\end{itemize}
	\item \texttt{Corners Hunting (CH) step.}
	\begin{itemize}
		\item Run SP algorithm with inputs $V_{\tau,2}$ to obtain $V_{\tau,2}(\mathcal{I},:)$.
	\end{itemize}
	\item \texttt{Membership Reconstruction (MR) step.}
	\begin{itemize}
		\item Recover $Z_{2}$ by setting $Z_{2}=V_{2}V'_{\tau,2}(\mathcal{I},:)(V_{\tau,2}(\mathcal{I},:)V'_{\tau,2}(\mathcal{I},:))^{-1}$.
		\item Recover $\Pi(i,:)$ by setting $\Pi(i,:)=\frac{Z_{2}(i,:)}{\|Z_{2}(i,:)\|_{1}}$ for $1\leq i\leq n$.
	\end{itemize}
\end{itemize}
The above analysis shows that the Ideal SRSC-equivalence exactly recovers the membership matrix $\Pi$. We now extend the ideal case to the real case as below.
\begin{algorithm}
	\caption{\textbf{SRSC-equivalence}}
	\label{alg:SRSCequivalence}
	\begin{algorithmic}[1]
		\Require The adjacency matrix $A\in \mathbb{R}^{n\times n}$, the number of communities $K$, and a ridge regularizer $\tau\geq 0$.
		\Ensure The estimated $n\times K$ membership matrix $\hat{\Pi}_{2}$.
		\State \texttt{RSC step.}
		\begin{itemize}
			\item Obtain $D_{\tau}, L_{\tau}, \hat{V}$ as Algorithm \ref{alg:SRSC}.
			\item Let $\hat{V}_{2}=\hat{V}\hat{V}'$ and $\hat{V}_{\tau,2}=D^{1/2}_{\tau}\hat{V}_{2}$.
		\end{itemize}
		\State \texttt{CH step.}
		\begin{itemize}
			\item Apply SP algorithm on the rows of $\hat{V}_{\tau,2}$ assuming there are $K$ clusters to obtain the near-corners matrix $\hat{V}_{\tau,2}(\mathcal{\hat{I}}_{2},:)\in\mathbb{R}^{K\times n}$, where $\mathcal{\hat{I}}_{2}$ is the index set returned by SP algorithm.
		\end{itemize}
		\State \texttt{Membership Reconstruction (MR) step.}
		\begin{itemize}
			\item Estimate $Z_{2}$ by setting $\hat{Z}_{2}=\hat{V}_{2}\hat{V}'_{\tau,2}(\mathcal{\hat{I}}_{2},:)(\hat{V}_{\tau,2}(\mathcal{\hat{I}}_{2},:)\hat{V}'_{\tau,2}(\mathcal{\hat{I}}_{2},:))^{-1}$.
			\item Set $\hat{Z}_{2}=\mathrm{max}(0,\hat{Z}_{2})$
			\item Estimate $\Pi(i,:)$ by setting $\hat{\Pi}_{2}(i,:)=\hat{Z}_{2}(i,:)/\|\hat{Z}_{2}(i,:)\|_{1}, 1\leq i\leq n$.
		\end{itemize}
	\end{algorithmic}
\end{algorithm}

\subsection{The CRSC-equivalence algorithm}
In this subsection, we introduce the CRSC-equivalence algorithm by starting from the ideal case. Next lemma shows $V_{*,2}$ has the Ideal Cone structure similar as $V_{*,1}$.
\begin{lem}\label{IdealCone2}
	Under $MMSB(n, P, \Pi)$, there exists a $Y_{2}\in\mathbb{R}^{n\times K}_{\geq 0}$ and no row of $Y_{2}$ is 0 such that
	\begin{align*}
	V_{*,2}=Y_{2}V_{*,2}(\mathcal{I},:),
	\end{align*}
	where $Y_{2}=N_{M_{2}}\Pi\mathscr{D}^{1/2}_{\tau}(\mathcal{I},\mathcal{I})N_{V_{2}}^{-1}(\mathcal{I},\mathcal{I})$, where $N_{M_{2}}$ is an $n\times n$ diagonal matrix whose diagonal entries are positive. Meanwhile, $V_{*,2}(i,:)=V_{*,2}(j,:)$ holds when $\Pi(i,:)=\Pi(j,:)$.
\end{lem}
Since $V_{*,2}(\mathcal{I},:)\in\mathbb{R}^{K\times n}$, $V_{*,2}(\mathcal{I},:)$ is singular but $V_{*,2}(\mathcal{I},:)V'_{*,2}(\mathcal{I},:)$ is nonsingular, by Lemma \ref{IdealCone2}, we have
\begin{align}\label{Y2}
Y_{2}=V_{*,2}V'_{*,2}(\mathcal{I},:)(V_{*,2}(\mathcal{I},:)V'_{*,2}(\mathcal{I},:))^{-1}.
\end{align}
Since $V_{*,2}=N_{V_{2}}V_{2}, Y_{2}=N_{M_{2}}\Pi\mathscr{D}^{1/2}_{\tau}(\mathcal{I},\mathcal{I})N_{V_{2}}^{-1}(\mathcal{I},\mathcal{I})$, we have $
N_{V_{2}}^{-1}N_{M_{2}}\Pi\mathscr{D}^{1/2}_{\tau}(\mathcal{I},\mathcal{I})N_{V_{2}}^{-1}(\mathcal{I},\mathcal{I})=V_{2}V'_{*,2}(\mathcal{I},:)(V_{*,2}(\mathcal{I},:)V'_{*,2}(\mathcal{I},:))^{-1}$, which gives that
\begin{align}\label{ZYJ2}
N_{V_{2}}^{-1}N_{M_{2}}\Pi=V_{2}V'_{*,2}(\mathcal{I},:)(V_{*,2}(\mathcal{I},:)V'_{*,2}(\mathcal{I},:))^{-1}N_{V_{2}}(\mathcal{I},\mathcal{I})\mathscr{D}^{-1/2}_{\tau}(\mathcal{I},\mathcal{I}).
\end{align}
For convenience, set  $J_{*,2}=N_{V_{2}}(\mathcal{I},\mathcal{I})\mathscr{D}^{-1/2}_{\tau}(\mathcal{I},\mathcal{I}),
Z_{*,2}=N_{V_{2}}^{-1}N_{M_{2}}\Pi, Y_{*,2}=V_{2}V'_{*,2}(\mathcal{I},:)(V_{*,2}(\mathcal{I},:)V'_{*,2}(\mathcal{I},:))^{-1}$. By Eq (\ref{ZYJ2}), we have
\begin{align}\label{Z2}
Z_{*,2}=Y_{*,2}J_{*,2}\equiv V_{2}V'_{*,2}(\mathcal{I},:)(V_{*,2}(\mathcal{I},:)V'_{*,2}(\mathcal{I},:))^{-1}N_{V_{2}}(\mathcal{I},\mathcal{I})\mathscr{D}^{-1/2}_{\tau}(\mathcal{I},\mathcal{I}).
\end{align}
Meanwhile, since $N_{V_{2}}^{-1}N_{M_{2}}$ is an $n\times n$ positive diagonal matrix, we have $\Pi(i,:)=\frac{Z_{*,2}(i,:)}{\|Z_{*,2}(i,:)\|_{1}}$ for $1\leq i\leq n$.
Then we have the following Ideal CRSC-equivalence algorithm. Input $\Omega, K$. Output: $\Pi$.
\begin{itemize}
	\item \texttt{RSC step}.
	\begin{itemize}
		\item Obtain $\mathscr{D}_{\tau}, \mathscr{L}_{\tau}$, $V$, $V_{2}=VV'$ and the row-normalization version of $V_{2}$, $V_{*,2}$. Compute $N_{V_{2}}$.
	\end{itemize}
	\item \texttt{Corners Hunting (CH) step}
	\begin{itemize}
		\item Run SVM-cone algorithm with inputs $V_{*,2}$ and $K$ to obtain  $\mathcal{I}$.
	\end{itemize}
	\item \texttt{Membership Reconstruction (MR) step.}
	\begin{itemize}
		\item Recover $Y_{*,2}$ and $J_{*,2}$ by setting
		$Y_{*,2}=V_{2}V'_{*,2}(\mathcal{I},:)(V_{*,2}(\mathcal{I},:)V'_{*,2}(\mathcal{I},:))^{-1}, J_{*,2}=N_{V_{2}}(\mathcal{I},\mathcal{I})\mathscr{D}^{-1/2}_{\tau}(\mathcal{I},\mathcal{I})$.
		\item Recover $Z_{*,2}$ by setting $Z_{*,2}=Y_{*,2}J_{*,2}$.
		\item Recover $\Pi(i,:)$ by setting
		$\Pi(i,:)=\frac{Z_{*,2}(i,:)}{\|Z_{*,2}(i,:)\|_{1}}$ for $1\leq i\leq n$.
	\end{itemize}
\end{itemize}
We now extend the ideal case to the real case as below.
\begin{algorithm}
	\caption{\textbf{CRSC-equivalence}}
	\label{alg:CRSCequivalence}
	\begin{algorithmic}[1]
		\Require The adjacency matrix $A\in \mathbb{R}^{n\times n}$, the number of communities $K$, and a ridge regularizer $\tau\geq 0$.
		\Ensure The estimated $n\times K$ membership matrix $\hat{\Pi}_{*,2}$.
		\State \texttt{RSC step.}
		\begin{itemize}
			\item Obtain $L_{\tau}, D_{\tau}, \hat{V}$ as Algorithm \ref{alg:SRSC}. Let $\hat{V}_{2}=\hat{V}\hat{V}'$. Let  $\hat{V}_{*,2}$ be the $n\times n$ matrix  such that $\hat{V}_{*,2}(i,:)=\frac{\hat{V}_{2}(i,:)}{\|\hat{V}_{2}(i,:)\|_{F}}, i=1,2,\ldots, n$. Obtain the  $n\times n$ diagonal matrix $N_{\hat{V}_{2}}$, whose $i$-th diagonal entry is $1/\|\hat{V}_{2}(i,:)\|_{F}$.
		\end{itemize}
		\State \texttt{CH step.}
		\begin{itemize}
			\item Apply SVM-cone algorithm on the rows of $\hat{V}_{*,2}$ assuming there are $K$ clusters to obtain the estimated index set $\mathcal{\hat{I}}_{*,2}$.
		\end{itemize}
		\State \texttt{Membership Reconstruction (MR) step.}
		\begin{itemize}
			\item Estimate $Y_{*,2}$ by setting $\hat{Y}_{*,2}=\hat{V}_{2}\hat{V}'_{*,2}(\mathcal{\hat{I}}_{*,2},:)(\hat{V}_{*,2}(\mathcal{\hat{I}}_{*,2},:)\hat{V}'_{*,2}(\mathcal{\hat{I}}_{*,2},:))^{-1}$.
			\item Estimate $J_{*,2}$ by setting $\hat{J}_{*,2}=N_{\hat{V}_{2}}(\mathcal{\hat{I}}_{*,2},\mathcal{\hat{I}}_{*,2})D^{-1/2}_{\tau}(\mathcal{\hat{I}}_{*,2},\mathcal{\hat{I}}_{*,2})$.
			\item Estimate $Z_{*,2}$ by setting $\hat{Z}_{*,2}=\hat{Y}_{*,2}\hat{J}_{*,2}$.
			\item Set $\hat{Z}_{*,2}=\mathrm{max}(0,\hat{Z}_{*,2})$
			\item Estimate $\Pi(i,:)$ by setting $\hat{\Pi}_{*,2}(i,:)=\hat{Z}_{*,2}(i,:)/\|\hat{Z}_{*,2}(i,:)\|_{1}, 1\leq i\leq n$.
		\end{itemize}
	\end{algorithmic}
\end{algorithm}

\subsection{The Equivalences}
We now emphasize the equivalence of  Algorithm \ref{alg:SRSC} and Algorithm \ref{alg:SRSCequivalence} as well as the equivalence of  Algorithm \ref{alg:CRSC} and Algorithm \ref{alg:CRSCequivalence} by lemmas \ref{SPSVMconeReturnTheSameOutputsUsingVandVV} and \ref{Equivalence}.
\begin{lem}\label{SPSVMconeReturnTheSameOutputsUsingVandVV}
	The SP algorithm will return the same node indices on both $\hat{V}_{\tau,1}$ and $\hat{V}_{\tau,2}$. Meanwhile, the SVM-cone algorithm will return the same node indices on both $\hat{V}_{*,1}$ and $\hat{V}_{*,2}$.
\end{lem}
\begin{lem}\label{Equivalence}
	For the ideal case, under $MMSB(n, P, \Pi)$, we have
	\begin{itemize}
		\item For Ideal SRSC and Ideal SRSC-equivalence, we have
		\begin{align*}
		V_{\tau,2}(\mathcal{I}:,)V'_{\tau,2}(\mathcal{I}:,)\equiv V_{\tau,1}(\mathcal{I}:,)V'_{\tau,1}(\mathcal{I}:,),Z_{2}\equiv Z_{1}.
		\end{align*}
		\item For Ideal CRSC and Ideal CRSC-equivalence, we have
		\begin{align*}
		&N_{V_{2}}\equiv N_{V},
		V_{*,2}(\mathcal{I}:,)V'_{*,2}(\mathcal{I}:,)\equiv V_{*,1}(\mathcal{I}:,)V'_{*,1}(\mathcal{I}:,),N_{M_{1}}\equiv N_{M_{2}},\\
		&Y_{2}\equiv Y_{1},Y_{*,2}\equiv Y_{*,1},J_{*,2}\equiv J_{*,1},Z_{*,2}\equiv Z_{*,1}.
		\end{align*}
	\end{itemize}
	For the empirical case, we have
	\begin{itemize}
		\item For SRSC and SRSC-equivalence, we have
		\begin{align*}
		\mathcal{\hat{I}}_{2}\equiv\mathcal{\hat{I}}_{1}, \hat{V}_{\tau,2}(\mathcal{\hat{I}}_{2},:)\hat{V}'_{\tau,2}(\mathcal{\hat{I}}_{2},:)\equiv \hat{V}_{\tau,1}(\mathcal{\hat{I}}_{1},:)\hat{V}'_{\tau,1}(\mathcal{\hat{I}}_{1},:), \hat{Z}_{2}\equiv \hat{Z}_{1}, \hat{\Pi}_{2}\equiv\hat{\Pi}_{1}.
		\end{align*}
		\item For CRSC and CRSC-equivalence, we have
		\begin{align*}
		&\mathcal{\hat{I}}_{*,2}\equiv\mathcal{\hat{I}}_{*,1},  \hat{V}_{*,2}(\mathcal{\hat{I}}_{*,2},:)\hat{V}'_{*,2}(\mathcal{\hat{I}}_{*,2},:)\equiv \hat{V}_{*,1}(\mathcal{\hat{I}}_{*,1},:)\hat{V}'_{*,1}(\mathcal{\hat{I}}_{*,1},:),\\
		& \hat{Y}_{*,2}\equiv \hat{Y}_{*,1} ,\hat{J}_{*,2}\equiv \hat{J}_{*,1}, \hat{Z}_{*,2}\equiv \hat{Z}_{*,1}, \hat{\Pi}_{*,2}\equiv\hat{\Pi}_{*,1}.
		\end{align*}
	\end{itemize}
\end{lem}
Lemma \ref{Equivalence} guarantees that Algorithm \ref{alg:SRSC} and Algorithm \ref{alg:SRSCequivalence} return same outputs. Lemma \ref{Equivalence} also guarantees that Algorithm \ref{alg:CRSC} and Algorithm \ref{alg:CRSCequivalence} return same outputs.

After showing the equivalences, from now on, for notation convenience, set $Z\equiv Z_{1}, N\equiv N_{V}, N_{M}\equiv N_{M_{1}}, Y\equiv Y_{1},Y_{*}\equiv Y_{*,1}, J_{*}\equiv J_{*,1}, Z_{*}\equiv Z_{*,1}$, and $\hat{N}\equiv N_{\hat{V}_{2}}, \mathcal{\hat{I}}\equiv\mathcal{\hat{I}}_{1},\hat{Z}\equiv \hat{Z}_{1}, \hat{\Pi}\equiv\hat{\Pi}_{1}, \mathcal{\hat{I}}_{*}\equiv\mathcal{\hat{I}}_{*,1},
\hat{Y}_{*}\equiv\hat{Y}_{*,1}, \hat{J}_{*}\equiv\hat{J}_{*,1}, \hat{Z}_{*}\equiv\hat{Z}_{*,1}, \hat{\Pi}_{*}\equiv\hat{\Pi}_{*,1}$.

\section{Main Results}\label{sec4}
In this section, we establish the performance guarantees for SRSC and CRSC. Since both methods are designed based on regularized Laplacian matrix, we first study several theoretical properties of the population and sample regularized Laplacian matrix. First, we make the following assumption
\begin{itemize}
	\item [(A1)] For two positive numbers  $\alpha$ and $\beta$, $\frac{\rho n}{\mathrm{log}(n^{\alpha}K^{-\beta})}\overset{n\rightarrow\infty}{\longrightarrow}\infty$.
\end{itemize}
Assumption (A1) means that the network can not be too sparse when $n$ is large. Meanwhile, when $K=O(1)$ or $\beta=0$, assumption (A1) is equivalent to $\rho n$ should grow faster than  $\mathrm{log}(n)$. In Lemma \ref{boundL}, we will show that $\alpha, \beta$ are directly related with the probability on the bound of $\|L_{\tau}-\mathscr{L}_{\tau}\|$ as well as the optimal choice of $\tau$ in Theorem \ref{Main}.
\begin{rem}
	In the language of \cite{mao2020estimating}, its Assumption 3.1 requires $|\lambda_{K}(\Omega)|\geq 4\sqrt{\rho n}\mathrm{log}^{\xi}(n)$ for $\xi>1$. Recall that $|\lambda_{K}(\Omega)|\leq \|\Omega\|=\rho\|\Pi'\tilde{P}\Pi\|=O(\rho n)$, we have $O(\rho n)\geq |\lambda_{K}(\Omega)|\geq 4\sqrt{\rho n}\mathrm{log}^{\xi}(n)$, which gives that $\rho n\geq O(\mathrm{log}^{2\xi}(n))$. Therefore, \cite{mao2020estimating}'s Assumption 3.1 on $\rho n$ should be $\rho n\geq O(\mathrm{log}^{2\xi}(n))$ (this is consistent with their Theorem F.1.) for $\xi>1$ instead of $\rho n=O(\mathrm{log}(n))$, and Table 1 in \cite{lei2019unified} also pointed out this. For comparison, when $\alpha=1, \beta=0$, our requirement on $\rho n$ in Condition (A1) is $\rho n\geq O(\mathrm{log}^{\xi}(n))$, which gives that based on the regularized Laplacian matrix $L_{\tau}$, our requirement on the sparsity parameter $\rho$ in Condition (A1) is weaker than the requirement of $\rho$ based on the adjacency matrix $A$ in \cite{mao2020estimating}. This guarantees that our two methods designed based on regularized Laplacian matrix $L_{\tau}$ can detect sparser networks than the SPACL algorithm \citep{mao2020estimating} designed based on $A$ for mixed membership community detection under MMSB.
\end{rem}
For convenience, set $\delta_{\mathrm{min}}=\underset{1\leq i\leq n}{\mathrm{min}}\mathscr{D}(i,i), \delta_{\mathrm{max}}=\underset{1\leq i\leq n}{\mathrm{max}}\mathscr{D}(i,i)$.
\begin{lem}\label{boundL}
	Under $MMSB(n, P, \Pi)$, suppose Condition (A1) holds, with probability at least $1-o(\frac{K^{4\beta}}{n^{4\alpha-1}})$, we have
	\begin{align*}
	\|L_{\tau}-\mathscr{L}_{\tau}\|=\begin{cases}
	O(\frac{\sqrt{\rho n\mathrm{log}(n^{\alpha}K^{-\beta})}}{\tau+\delta_{\mathrm{min}}}), & \mbox{when } C\sqrt{\rho n\mathrm{log}(n^{\alpha}K^{-\beta})}\leq \tau+\delta_{\mathrm{min}}\leq C\rho n, \\
	O(\frac{\rho n\mathrm{log}(n^{\alpha}K^{-\beta})}{(\tau+\delta_{\mathrm{min}})^{2}}), & \mbox{when~}\tau+\delta_{\mathrm{min}}<C\sqrt{\rho n\mathrm{log}(n^{\alpha}K^{-\beta})}.
	\end{cases}
	\end{align*}
\end{lem}
\begin{rem}
	In Lemma \ref{boundL}, we see that $\tau+\delta_{\mathrm{min}}$ has a upper bound $C\rho n$ (hence, $\tau\leq C\rho n$) for some $C>0$ and $\tau$ should be nonnegative since $\tau+\delta_{\mathrm{min}}$ is in the denominator position in the bound. Therefore, even we set $\tau$ as 0, the bound in Lemma \ref{boundL} is also meaningful. The upper bound $C\rho n$ of $\tau+\delta_{\mathrm{min}}$ occurs naturally in the proof of this lemma. Meanwhile, setting $\tau$ to large is meaningless. By Lemma 4 in the supplementary material, we know that $\lambda_{1}=\|\mathscr{L}_{\tau}\|\leq \frac{\delta_{\mathrm{max}}}{\tau+\delta_{\mathrm{max}}}$. Therefore, when $\tau$ is too large (say, an extreme case that $\tau\rightarrow\infty$), $\lambda_{1}$ tends to be zero, causing that $\mathscr{L}_{\tau}$ tends to be a zero matrix. And in this case, for any adjacency matrix $A$ whether it comes from MMSB or not, as long as $D(i,i)\neq 0$, $L_{\tau}$ also tends to  be a zero matrix, and $\|L_{\tau}-\mathscr{L}_{\tau}\|$ always tends to be zero when $\tau\rightarrow\infty$. However, this is meaningless for community detection. Therefore, an appropriate upper bound for $\tau$ is reasonable.
\end{rem}

In Lemma \ref{boundL}, the convergence probability is  related with  parameter $\alpha$ and $\beta$ instead of a constant, and we call such probability as \emph{parametric probability}. After giving the main result Theorem \ref{Main}, we will provide an explanation on the convenience of the theoretically optimal choice of the regularizer $\tau$ by the newly defined parametric probability. For convenience, denote $err_{n}=\|L_{\tau}-\mathscr{L}_{\tau}\|$.

\subsection{Performance guarantees for SRSC and CRSC}\label{Sec4.2}
In this subsection, we aim to show the asymptotic consistency of SRSC and CRSC, i.e., to prove that $\hat{\Pi}$ and $\hat{\Pi}_{*}$ concentrate around $\Pi$ if the sampled network is generated from the $MMSB(n, P, \Pi)$.  Meanwhile, to show the asymptotic property of these two methods, the three parameters $K, P,\Pi$ can change with $n$.  The theoretical error bounds given in Theorem \ref{Main} are directly related with the model parameters $(n, P,\Pi)$ and $K$, which allows the analyticity by changing these model parameters to see the influence of these parameters on SRSC and CRSC.

In \cite{mixedSCORE, MaoSVM, mao2020estimating}, main theoretical results for  their proposed community detection methods hinge on a row-wise deviation bound for the eigenvectors of the adjacency matrix whether under MMSB or DCMM. Similarly, for our SRSC and CRSC, the main theoretical results (i.e., Theorem \ref{Main}) also rely on the row-wise deviation bound for the eigenvector of the regularized Laplacian matrix.  Different from the theoretical techniques in Theorem 3.1 in \cite{mao2020estimating} and Lemma C.3 in \cite{mixedSCORE}, to obtain the row-wise deviation bound for the eigenvector of the regularized Laplacian matrix, we use a combination of Theorem 4.2.1 in \cite{chen2020spectral} and Lemma 5.1 in \cite{lei2015consistency}.
\begin{lem}\label{rowwiseerror}
	(Row-wise eigenvector error) Under $MMSB(n, P, \Pi)$, suppose Condition (A1) holds, assume $|\lambda_{K}|\geq C\frac{\sqrt{\rho n\mathrm{log}(n)}}{\tau+\delta_{\mathrm{min}}}$, with probability at least $1-o(\frac{K^{4\beta}}{n^{4\alpha-1}})$, we have
	\begin{align*}
	\|\hat{V}\hat{V}'-VV'\|_{2\rightarrow\infty}=O(\frac{(\tau+\delta_{\mathrm{max}})\sqrt{K\mathrm{log}(n^{\alpha}K^{-\beta})}}{(\tau+\delta_{\mathrm{min}}) |\lambda_{K}(\tilde{P})|\lambda_{K}(\Pi'\Pi)\sqrt{\rho}}).
	\end{align*}
\end{lem}
For convenience, we set $\varpi=\|\hat{V}\hat{V}'-VV'\|_{2\rightarrow\infty}$. We will  use the row-wise eigenvector error $\varpi$ to construct the error bounds for our theoretical analysis for SRSC and CRSC. We emphasize that the statement of Lemma \ref{rowwiseerror} considers both positive and negative eigenvalues of $\mathscr{L}_{\tau}$ and $L_{\tau}$.
\begin{rem}
	If one set the $\Theta$ in \cite{mixedSCORE} as $\Theta=\sqrt{\rho}I$, we can  see that the DCMM model degenerates to the $MMSB(n,P,\Pi)$ considered in this paper. In this case, the row-wise eigenvector deviation in the 4th bullet of Lemma 2.1 in \cite{mixedSCORE} is $O(\frac{1}{|\lambda_{K}(\tilde{P})|}\frac{K^{1.5}}{\sqrt{n}}\sqrt{\frac{\mathrm{log}(n)}{\rho n}})$ under their Conditions (where their conditions are our Condition (A1) and the assumption $\lambda_{K}(\Pi'\Pi)=O(\frac{n}{K})$ when $\Theta=\sqrt{\rho}I$), which is consistent with our bound in Lemma \ref{rowwiseerror}.
\end{rem}
Based on the conditions in Lemma \ref{rowwiseerror}, we can obtain the choice of $\tau$ through the following analysis. Since we assume that $|\lambda_{K}|\geq C\frac{\sqrt{\rho n\mathrm{log}(n)}}{\tau+\delta_{\mathrm{min}}}$, combine with the fact that $err_{n}=C\frac{\sqrt{\rho n\mathrm{log}(n^{\alpha}K^{-\beta})}}{\tau+\delta_{\mathrm{min}}}$ when $C\sqrt{\rho n\mathrm{log}(n^{\alpha}K^{-\beta})}\leq \tau+\delta_{\mathrm{min}}\leq C\rho n$, we have $|\lambda_{K}|\geq C err_{n}$. Due to the fact that $|\lambda_{K}|\leq \lambda_{1}\leq 1$ (by Lemma 4 in the supplementary material), we have $err_{n}\leq \frac{1}{C}$. On the one hand, when $C\sqrt{\rho n\mathrm{log}(n^{\alpha}K^{-\beta})}\leq\tau+\delta_{\mathrm{min}}\leq C\rho n$, we have $err_{n}=O(\frac{\sqrt{\rho n\mathrm{log}(n^{\alpha}K^{-\beta})}}{\tau+\delta_{\mathrm{min}}})$. By $err_{n}\leq 1/C$,  we see that $\tau+\delta_{\mathrm{min}}\geq C\sqrt{\rho n\mathrm{log}(n^{\alpha}K^{-\beta})}$ by ignoring the effect of $\delta_{\mathrm{min}}$, which is consistent with the case that $\tau+\delta_{\mathrm{min}}\geq C\sqrt{\rho n\mathrm{log}(n^{\alpha}K^{-\beta})}$. On the other hand, when $\tau+\delta_{\mathrm{min}}< C\sqrt{\rho n\mathrm{log}(n^{\alpha}K^{-\beta})}$, we have $err_{n}=O(\frac{\rho n\mathrm{log}(n^{\alpha}K^{-\beta})}{(\tau+\delta_{\mathrm{min}})^{2}})$. As $err_{n}\leq 1/C$,  we see that $\tau+\delta_{\mathrm{min}}\geq \sqrt{C\rho n\mathrm{log}(n^{\alpha}K^{-\beta})}$, which is a contradiction. Hence, to make the condition of the lower bound of $|\lambda_{K}|$ hold, we need $C\sqrt{\rho n \mathrm{log}(n^{\alpha}K^{-\beta})}\leq\tau+\delta_{\mathrm{min}}\leq C\rho n$, then $err_{n}$ can be always written as $err_{n}=O(\frac{\sqrt{\rho n\mathrm{log}(n^{\alpha}K^{-\beta})}}{\tau+\delta_{\mathrm{min}}})$ under Condition (A1). Furthermore, from the requirement $C\sqrt{\rho n\mathrm{log}(n^{\alpha}K^{-\beta})}\leq \tau+\delta_{\mathrm{min}}\leq C\rho n$, we can see that the benefit of regularization is that regularized spectral clustering (i.e., when $\tau>0$) can detect sparser networks than spectral clustering when $\tau=0$ due to the fact that if we set $\tau=0$, the lower bound requirement of $\delta_{\mathrm{min}}$ is $C\sqrt{\rho n\mathrm{log}(n^{\alpha}K^{-\beta})}$, which is larger than $C\sqrt{\rho n\mathrm{log}(n^{\alpha}K^{-\beta})}-\tau$ when $\tau>0$. And such benefit can also be found in the main theorem \ref{Main} of this paper.

Bounds provided by Lemma \ref{boundC} are the corner stones to characterize the behaviors of our SRSC and CRSC approaches. For convenience, set  $\pi_{\mathrm{min}}=\mathrm{min}_{1\leq k\leq K}\mathbf{1}'\Pi e_{k}$, where $\pi_{\mathrm{min}}$ measure the minimum summation of nodes belong to certain community and increasing $\pi_{\mathrm{min}}$ makes the network tend to be more balanced. For convenience, set $\eta=\mathrm{min}_{1\leq k\leq K}((V_{*,1}(\mathcal{I},:)V'_{*,1}(\mathcal{I},:))^{-1}\mathbf{1})(k)$.
\begin{lem}\label{boundC}
	Under $MMSB(n,P,\Pi)$, when conditions of Lemma \ref{rowwiseerror} hold, there exist two permutation matrices $\mathcal{P},\mathcal{P}_{*}\in\mathbb{R}^{K\times K}$ such that with probability at least $1-o(\frac{K^{4\beta}}{n^{4\alpha-1}})$, we have
	\begin{align*}
	&\|\hat{V}_{\tau,2}(\mathcal{\hat{I}},:)-\mathcal{P}V_{\tau,2}(\mathcal{I},:)\|_{F}=O(\frac{(\tau+\delta_{\mathrm{max}})^{1.5}\sqrt{K}\varpi\kappa(\Pi'\Pi)}{\tau+\delta_{\mathrm{min}}}),\\
	&\|\hat{V}_{*,2}(\mathcal{\hat{I}}_{*},:)-\mathcal{P}_{*}V_{*,2}(\mathcal{I},:)\|_{F}=O((\frac{\tau+\delta_{\mathrm{max}}}{\tau+\delta_{\mathrm{min}}})^{3.5}\frac{K^{2.5}\varpi\kappa^{3}(\Pi'\Pi)\sqrt{\lambda_{1}(\Pi'\Pi)}}{\eta}).
	\end{align*}
\end{lem}
Set $\varpi_{S}=\|\hat{V}_{\tau,2}(\mathcal{\hat{I}},:)-\mathcal{P}V_{\tau,2}(\mathcal{I},:)\|_{F}, \varpi_{C}=\|\hat{V}_{*,2}(\mathcal{\hat{I}},:)-\mathcal{P}_{*}V_{*,2}(\mathcal{I},:)\|_{F}$ for convenience in the proofs. Next lemma bounds the row-wise deviation between $\hat{Y}_{*}$ and $Y_{*}$ for CRSC.
\begin{lem}\label{boundY}
	Under $MMSB(n,P,\Pi)$, when conditions of Lemma \ref{rowwiseerror} hold, with probability at least $1-o(\frac{K^{4\beta}}{n^{4\alpha-1}})$, we have
	\begin{align*}
	\mathrm{max}_{1\leq i\leq n}\|e'_{i}(\hat{Y}_{*}-Y_{*}\mathcal{P}_{*})\|_{F}=O((\frac{\tau+\delta_{\mathrm{max}}}{\tau+\delta_{\mathrm{min}}})^{5}\frac{K^{3.5}\varpi\kappa^{4.5}(\Pi'\Pi)}{\eta}).
	\end{align*}
\end{lem}
Now we are ready to bound $\|e'_{i}(\hat{Z}-Z\mathcal{P})\|_{F}$ and $\|e'_{i}(\hat{Z}_{*}-Z_{*}\mathcal{P}_{*})\|_{F}$ based on Lemma \ref{boundY}.
\begin{lem}\label{BoundZ}
	Under $MMSB(n, P,\Pi)$, when conditions of Lemma \ref{rowwiseerror} hold, with probability at least $1-o(\frac{K^{4\beta}}{n^{4\alpha-1}})$, we have
	\begin{align*}
	&\mathrm{max}_{1\leq i\leq n}\|e'_{i}(\hat{Z}-Z\mathcal{P})\|_{F}=O(\frac{(\tau+\delta_{\mathrm{max}})^{1.5}K\varpi\kappa(\Pi'\Pi)\sqrt{\lambda_{1}(\Pi'\Pi)}}{(\tau+\delta_{\mathrm{min}})^{2}}),\\
	&\mathrm{max}_{1\leq i\leq n}\|e'_{i}(\hat{Z}_{*}-Z_{*}\mathcal{P}_{*})\|_{F}=O(\frac{(\tau+\delta_{\mathrm{max}})^{7.5}K^{5.5}\varpi\kappa^{4.5}(\Pi'\Pi)\lambda^{1.5}_{1}(\Pi'\Pi)}{(\tau+\delta_{\mathrm{min}})^{8}\pi_{\mathrm{min}}}).
	\end{align*}
\end{lem}
Next theory is the main result for SRSC and CRSC to infer the membership parameters under MMSB.
\begin{thm}\label{Main}
	Under $MMSB(n, P,\Pi)$, when conditions of Lemma \ref{rowwiseerror} hold, with probability at least $1-o(\frac{K^{4\beta}}{n^{4\alpha-1}})$,
	\begin{align*}
	&\|e'_{i}(\hat{\Pi}-\Pi\mathcal{P})\|_{F}=O(\frac{(\tau+\delta_{\mathrm{max}})^{3}K^{2}\kappa(\Pi'\Pi)\sqrt{\lambda_{1}(\Pi'\Pi)\mathrm{log}(n^{\alpha}K^{-\beta})}}{(\tau+\delta_{\mathrm{min}})^{3}|\lambda_{K}(\tilde{P})|\lambda_{K}(\Pi'\Pi)\sqrt{\rho}}),\\
	&\|e'_{i}(\hat{\Pi}_{*}-\Pi\mathcal{P}_{*})\|_{F}=O(\frac{(\tau+\delta_{\mathrm{max}})^{9}K^{6.6}\kappa^{4.5}(\Pi'\Pi)\lambda^{1.5}_{1}(\Pi'\Pi)\sqrt{\mathrm{log}(n^{\alpha}K^{-\beta})}}{(\tau+\delta_{\mathrm{min}})^{9}|\lambda_{K}(\tilde{P})|\lambda_{K}(\Pi'\Pi)\pi_{\mathrm{min}}\sqrt{\rho}}).
	\end{align*}
\end{thm}
Note that, when $\alpha=1, \beta=0$, the convergence probability is $1-o(n^{-3})$, which is a common probability in community detection, see \cite{SCORE, mixedSCORE}.  For a general comparison of SRSC and CRSC, from Theorem \ref{Main}, we see that CRSC is more sensitive on $K$, $\kappa(\Pi'\Pi)$ and unbalanced network (A larger $\frac{\tau+\delta_{\mathrm{max}}}{\tau+\delta_{\mathrm{min}}}$ refers to a more unbalanced network.) than SRSC. Both two methods have same sensitivity on the row-wise eigenvector deviation term $\varpi$. Generally, Theorem \ref{Main} says that the error bound for SRSC is slightly smaller than that of CRSC. Furthermore, Theorem \ref{Main} also says that a smaller $\beta$ and a larger $\alpha$ lead to a lager probability of successfully detecting mixed membership networks under MMSB. However, by Condition (A1), we see that smaller $\beta$ and larger $\alpha$ lead to stronger assumptions on the sparsity of a network under MMSB. Therefore, there is a trade-off between the sparsity of a network and the probability of successfully detecting its mixed memberships.

For both two methods, since $\delta_{\mathrm{min}}\leq \delta_{\mathrm{max}}$, when $\tau$ increases, error bounds in Theorem \ref{Main} decrease, which suggests that a larger $\tau$ gives better estimations. Recall that $C\sqrt{\rho n \mathrm{log}(n^{\alpha}K^{-\beta})}\leq\tau+\delta_{\mathrm{min}}\leq C\rho n$, therefore the theoretical optimal choice of $\tau$ is:
\begin{align}\label{tauoptimal}
\tau_{\mathrm{opt}}=O(\rho n).
\end{align}
If we further add conditions similar as Corollary 3.1 in \cite{mao2020estimating}, then we have the following corollary.
\begin{cor}\label{AddConditions}
	Under $MMSB(n, P,\Pi)$, when conditions of Lemma \ref{rowwiseerror} hold, set $\tau$ as in Eq (\ref{tauoptimal}), suppose $K=O(1), \pi_{\mathrm{min}}=O(\frac{n}{K}),  \lambda_{1}(\Pi'\Pi)=O(\frac{n}{K})$, and $\delta_{\mathrm{max}}\leq C\delta_{\mathrm{min}}$, with probability at least $1-o(\frac{K^{4\beta}}{n^{4\alpha-1}})$, we have
	\begin{align*}
	\|e'_{i}(\hat{\Pi}-\Pi\mathcal{P})\|_{F}=O(\frac{1}{|\lambda_{K}(\tilde{P})|}\sqrt{\frac{\mathrm{log}(n)}{\rho n}}),~~~\|e'_{i}(\hat{\Pi}_{*}-\Pi\mathcal{P}_{*})\|_{F}=O(\frac{1}{|\lambda_{K}(\tilde{P})|}\sqrt{\frac{\mathrm{log}(n)}{\rho n}}).
	\end{align*}
	Especially, for the sparest case when $\rho n=O(\mathrm{log}^{1+2\varsigma}(n))$ for $\varsigma\rightarrow 0^{+}$, we have
	\begin{align*}
	\|e'_{i}(\hat{\Pi}-\Pi\mathcal{P})\|_{F}=O(\frac{1}{|\lambda_{K}(\tilde{P})|\mathrm{log}^{\varsigma}(n)}),~~~\|e'_{i}(\hat{\Pi}_{*}-\Pi\mathcal{P}_{*})\|_{F}=O(\frac{1}{|\lambda_{K}(\tilde{P})|\mathrm{log}^{\varsigma}(n)}).
	\end{align*}
\end{cor}
\begin{rem}
	Under the setting of Corollary \ref{AddConditions}, the condition $|\lambda_{K}|\geq C\sqrt{\rho n\mathrm{log}(n)}/(\tau+\delta_{\mathrm{min}})$ holds naturally. By Lemma 4 in the supplementary material, we know that $|\lambda_{K}|\geq \frac{\rho |\lambda_{K}(\tilde{P})\lambda_{K}(\Pi'\Pi)|}{\tau+\delta_{\mathrm{max}}}=C\frac{\rho |\lambda_{K}(\tilde{P})|n}{K(\tau+\delta_{\mathrm{min}})}=C\frac{\rho n|\lambda_{K}(\tilde{P})|}{\tau+\delta_{\mathrm{min}}}$. The inequality $C\frac{\rho n|\lambda_{K}(\tilde{P})|}{\tau+\delta_{\mathrm{min}}}\geq C\sqrt{\rho n\mathrm{log}(n)}/(\tau+\delta_{\mathrm{min}})$ equals to $|\lambda_{K}(\tilde{P})|\geq C\sqrt{\frac{\mathrm{log}(n)}{\rho n}}$, and it just matches with the requirement of the consistency of clustering in Corollary \ref{AddConditions}.
\end{rem}
From Corollary \ref{AddConditions}, we see that when $\alpha$ is fixed, though a large $\beta$ lowers the requirement on the network sparsity in Condition (A1) (i.e., a large $\beta$ decreases the requirement on the lower bound on $\rho n$), it decreases the probability (i.e., increasing $\beta$ decreases $1-o(\frac{K^{4\beta}}{n^{4\alpha-1}})$). Similarly, when $\beta$ is fixed, though a small $\alpha$ lowers the requirement on the network sparsity in Condition (A1) (i.e., a decreasing $\alpha$ decreases the requirement on the lower bound on $\rho n$), it decreases the probability (i.e., a decreasing $\alpha$ decreases $1-o(\frac{K^{4\beta}}{n^{4\alpha-1}})$). When dealing with empirical networks, since $\alpha,\beta,\rho$ are unknown (i.e., we have no knowledge about the sparsity of the empirical networks), if $\tau$ is too large (which can be seen as setting $\beta$ too large or $\alpha$ too small in Eq (\ref{tauoptimal})), SRSC and CRSC can still work but with small convergence probability due to the fact that a very large $\beta$  or a very small $\alpha$ in Eq (\ref{tauoptimal}) decrease $1-o(\frac{K^{4\beta}}{n^{4\alpha-1}})$. This explains that when dealing with empirical networks, even if $\tau$ is very large, our methods still work but with small probability to have satisfactory performances, which suggesting that a moderate choice of $\tau$ is preferred for both SRSC and CRSC. Meanwhile, as the statement after Theorem \ref{Main}, $\tau$ can not be too small.
\begin{rem}
	(Empirical optimal choice of $\tau$) Set $\alpha=1,\beta=0$, then the convergence probabilities in the above lemmas, theorems and corollaries are $1-o(n^{-3})$. By (Eq \ref{tauoptimal}), we see that $\tau_{\mathrm{optimal}}$ depends on $\rho$ where the parameter $\rho$ controls the sparsity of a network generated under $MMSB(n, P, \Pi)$. Since most real world networks are sparse and $K\ll n$ (i.e., $K$ can be seen as $O(1)$), and network generated under the case that $\rho n=O(\mathrm{log}^{1+2\varsigma}(n))$ for any $\varsigma\rightarrow 0^{+}$ is the sparsest network satisfying Condition (A1), for such sparse network, by Eq (\ref{tauoptimal}) we should set $\tau$ as $
	\tau_{\mathrm{opt}}=O(\rho n)\equiv O(\mathrm{log}(n))$, where we set $\varsigma=0$ directly. Therefore, for both SRSC and CRSC, the optimal choices for $\tau$ for the sparsest network satisfying Condition (A1) are the same, and we should set the optimal choice of $\tau$ as
	\begin{align}\label{tauFinalOptimal}
	\tau_{\mathrm{opt}}=O(\mathrm{log}(n)).
	\end{align}
\end{rem}
Consider the balanced mixed membership network in Corollary \ref{AddConditions}, we further assume that $\tilde{P}=\gamma I_{K}+(1-\gamma)I_{K}I'_{K}$ for $0<\gamma<1$ when $K=O(1)$ and call such network as standard mixed membership network with $K$ balanced clusters. To obtain consistency estimation, $\gamma$ should grow faster than $\sqrt{\frac{\mathrm{log}(n)}{\rho n}}$ since $|\lambda_{K}(P)|=\gamma$. Let $P_{\mathrm{max}}=\max_{k,l}P(k,l), P_{\mathrm{min}}=\mathrm{min}_{k,l}P(k,l)$. Consider the sparest case when $\rho=O(\frac{\mathrm{log}^{1+2\varsigma}(n)}{n})$, since $P=\rho \tilde{P}$ , we have $P_{\mathrm{max}}-P_{\mathrm{min}}=\rho\gamma$ (the probability gap) should grow faster than $\frac{\mathrm{log}(n)}{n}$ when $\varsigma\rightarrow 0^{+}$, and $\frac{P_{\mathrm{max}}-P_{\mathrm{min}}}{\sqrt{P_{\mathrm{max}}}}=\gamma\sqrt{\rho}$ (the relative edge probability gap) should grow faster than $\sqrt{\frac{\mathrm{log}(n)}{n}}$. And such conclusion also holds when all nodes are pure. Note that for the balanced network with $K=2$ and all nodes are pure, the conclusion that $\frac{P_{\mathrm{max}}-P_{\mathrm{min}}}{\sqrt{P_{\mathrm{max}}}}$ should grow faster than $\sqrt{\frac{\mathrm{log}(n)}{n}}$ is consistent with Theorem 2.1 in \cite{li2021convex} and Corollary 1 in \cite{mcsherry2001spectral}. However, Corollary 1 \citep{mcsherry2001spectral} requires that $P_{\mathrm{max}}n$ should be at least $O(\mathrm{log}^{6}(n))$, while our requirement on $P_{\mathrm{max}}n$ (recall that $P_{\mathrm{max}}=\rho$) is it should be at least $O(\mathrm{log}(n))$.

Consider the Erdos-Renyi random graph $G(n,p)$ \citep{erdos2011on} for the sparest case when $\rho n=O(\mathrm{log}^{1+2\varsigma}(n))$ for $\varsigma\rightarrow 0^{+}$ and $K=1$. Since $P=\rho \tilde{P}=\rho \lambda_{1}(\tilde{P})=p$, we have $\lambda_{1}(\tilde{P})=\frac{p}{\rho}=p\frac{n}{\mathrm{log}^{1+2\varsigma}(n)}$.  Then by Corollary \ref{AddConditions}, the upper bound of error rate is
\begin{align*}
\|e'_{i}(\hat{\Pi}-\Pi\mathcal{P})\|_{F}=O(\frac{\mathrm{log}^{1+\varsigma}(n)}{pn}),
\end{align*}
and $\|e'_{i}(\hat{\Pi}_{*}-\Pi\mathcal{P}_{*})\|_{F}$ shares the same bound. So, we see that $p$ should grow faster than $\frac{\mathrm{log}(n)}{n}$ to make the bound less than 1, i.e.,  the probability parameter $p$ in the of Erdos-Renyi graph $G(n,p)$ should be at least the order of $\frac{\mathrm{log}(n)}{n}$ to generated a connected random graph. Hence, the disappearance of isolated vertices in $G(n,p)$ has a sharp threshold of $\frac{\mathrm{log}(n)}{n}$, and this sharp threshold is consistent with Theorem 4.6 in \cite{blum2020foundations} and the first bullet in Section 2.5 in \cite{abbe2017community}.
\begin{rem}
	(Comparison to Theorem 2.2 in \cite{mixedSCORE}) Replacing the $\Theta$ in \cite{mixedSCORE} by$\Theta=\sqrt{\rho}I$, their DCMM model  degenerates to the $MMSB(n,P,\Pi)$. Then their conditions in Theorem 2.2 are our Condition (A1) and $\lambda_{K}(\Pi'\Pi)=O(\frac{n}{K})$ actually. When $K=O(1)$, we see that the error bound in Theorem 2.2 in \cite{mixedSCORE} is also $O(\frac{1}{|\lambda_{K}(\tilde{P})|}\sqrt{\frac{\mathrm{log}(n)}{\rho n}})$. Therefore this bound can also be applied to obtain the probability gap (and the relative edge probability gap) of the standard network with $K$ balanced clusters and the sharp threshold of the Erdos-Renyi random graph $G(n,p)$.
\end{rem}
\begin{rem}
	(Comparison to Theorem 3.2 in \cite{mao2020estimating}) Error bound in Theorem 3.2 in \cite{mao2020estimating} is $O(\frac{1}{|\lambda_{K}(\tilde{P})|\sqrt{\rho n}})$. Though this bound is sharper than ours and \cite{mixedSCORE}'s with a $1/\sqrt{\mathrm{log}(n)}$ term, it can not be applied to obtain the sharp threshold of the Erdos-Renyi random graph $G(n,p)$. 
\end{rem}
\section{Evaluation on synthetic networks}\label{sec5}
In this section, a small-scale numerical study is applied to investigate the performances of our SRSC and CRSC by comparing them with Mixed-SCORE \citep{mixedSCORE}, OCCAM \citep{OCCAM},  SVM-cone-DCMMSB \citep{MaoSVM} and SPACL \citep{mao2020estimating}.
We measure the performance of these methods by the mixed-Hamming error rate:
\begin{align*}
\mathrm{min}_{O\in\{ K\times K\mathrm{permutation~matrix}\}}\frac{1}{n}\|\hat{\Pi}O-\Pi\|_{1},
\end{align*}
where $\Pi$ and $\hat{\Pi}$ are the true and estimated mixed membership matrices respectively. Here, we also consider the permutation of labels since the measurement of error should not depend on how we label each of the K communities. 

For all simulations, unless specified, our simulations have $n$ nodes and $K$ blocks, let each block own $n_{0}$ number of pure nodes. For the top $Kn_{0}$ nodes $\{1,2, \ldots, Kn_{0}\}$, we let these nodes be pure and let nodes $\{Kn_{0}+1, Kn_{0}+2,\ldots, n\}$ be mixed. Unless specified, let all the mixed nodes have four different memberships $(0.4, 0.4, 0.2), (0.4, 0.2, 0.4), (0.2, 0.4, 0.4)$ and $(1/3,1/3,1/3)$, each with $\frac{n-Kn_{0}}{4}$ number of nodes. Unless specified, $\tilde{P}$ has unit diagonals  and off-diagonals $0.5$, and let $P=\rho P$, where $\rho$ may be changed.  For each parameter setting, we report the averaged mixed-Hamming error rate 
over 50 repetitions.

\texttt{Experiment 1: Changing $K$.} Fix $(n, n_{0})=(1000, 60)$ and let $\rho$ be 0.5 or 0.8. We vary $K$ in the range $\{2, 3, \ldots, 8\}$. For the $n-Kn_{0}$ mixed nodes, let them belong to each block with equal probability $\frac{1}{K}$.  The numerical results are shown in panels (a) and (b) of Figure \ref{EX} (note that we use the label SVM-cD to denote SVM-cone-DCMMSB in Figure \ref{EX} such that the label does not cover the numerical results in the figure.), which tells us that all methods perform better when $K$ increases. This phenomenon occurs since $n$ is fixed, for a small $K$, the fraction of pure nodes is $\frac{60K}{1000}$ is small while the fraction of mixed nodes is large and all mixed nodes are heavily mixed (since mixed nodes belong to each block with equal probability). As $K$ increases in this experiment, the fraction of pure nodes increases, and this is the fundamental reason that all methods perform better as $K$ increases. Meanwhile, SRSC and SPACL perform similar and these two methods outperform other approaches while CRSC only outperform OCCAM in this experiment.
\begin{figure}
	\centering
	\subfigure[Changing $K$ when $\rho=0.5$]
	{\includegraphics[width=0.45\textwidth]{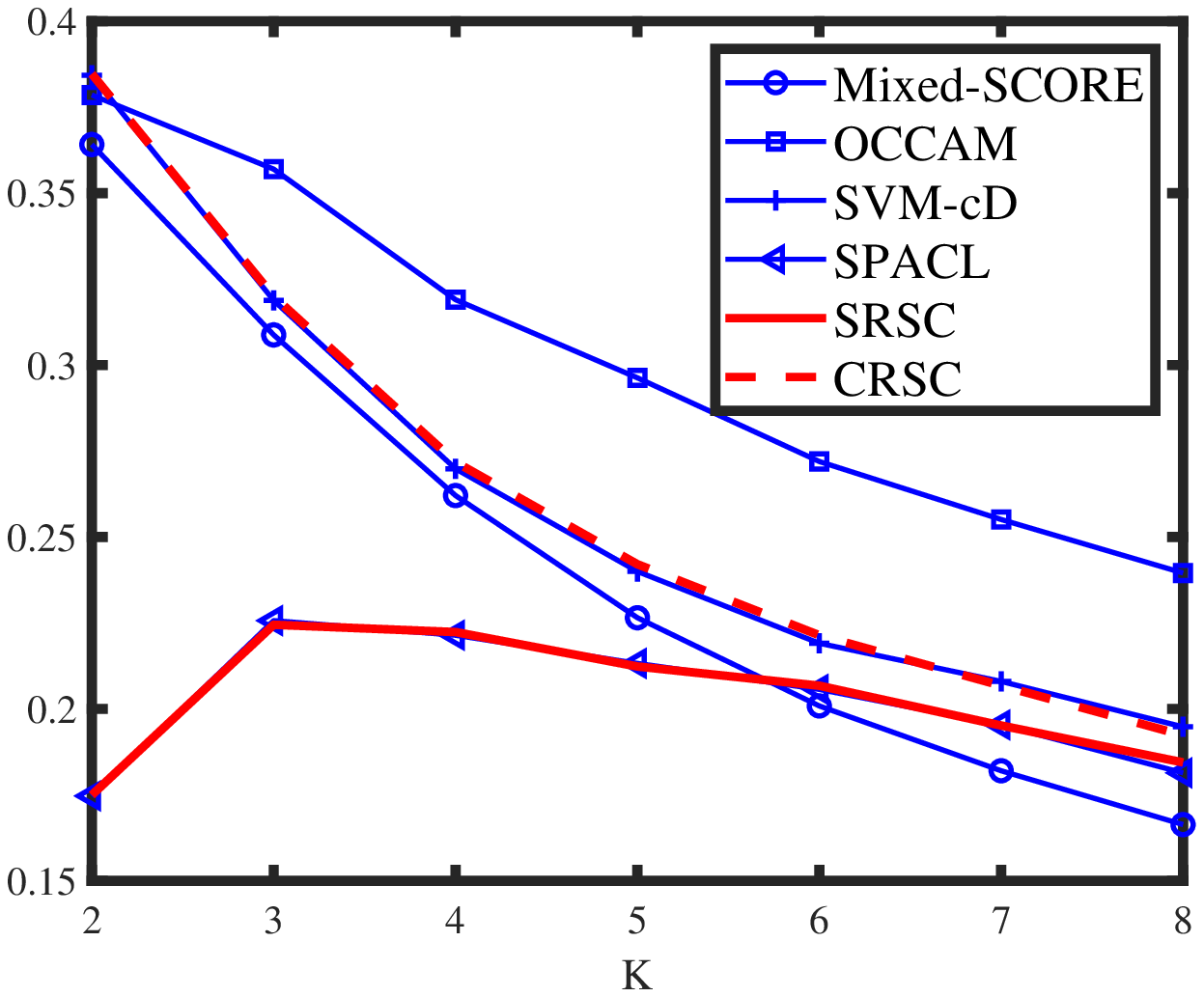}}
	\subfigure[Changing $K$ when $\rho=0.8$]
	{\includegraphics[width=0.45\textwidth]{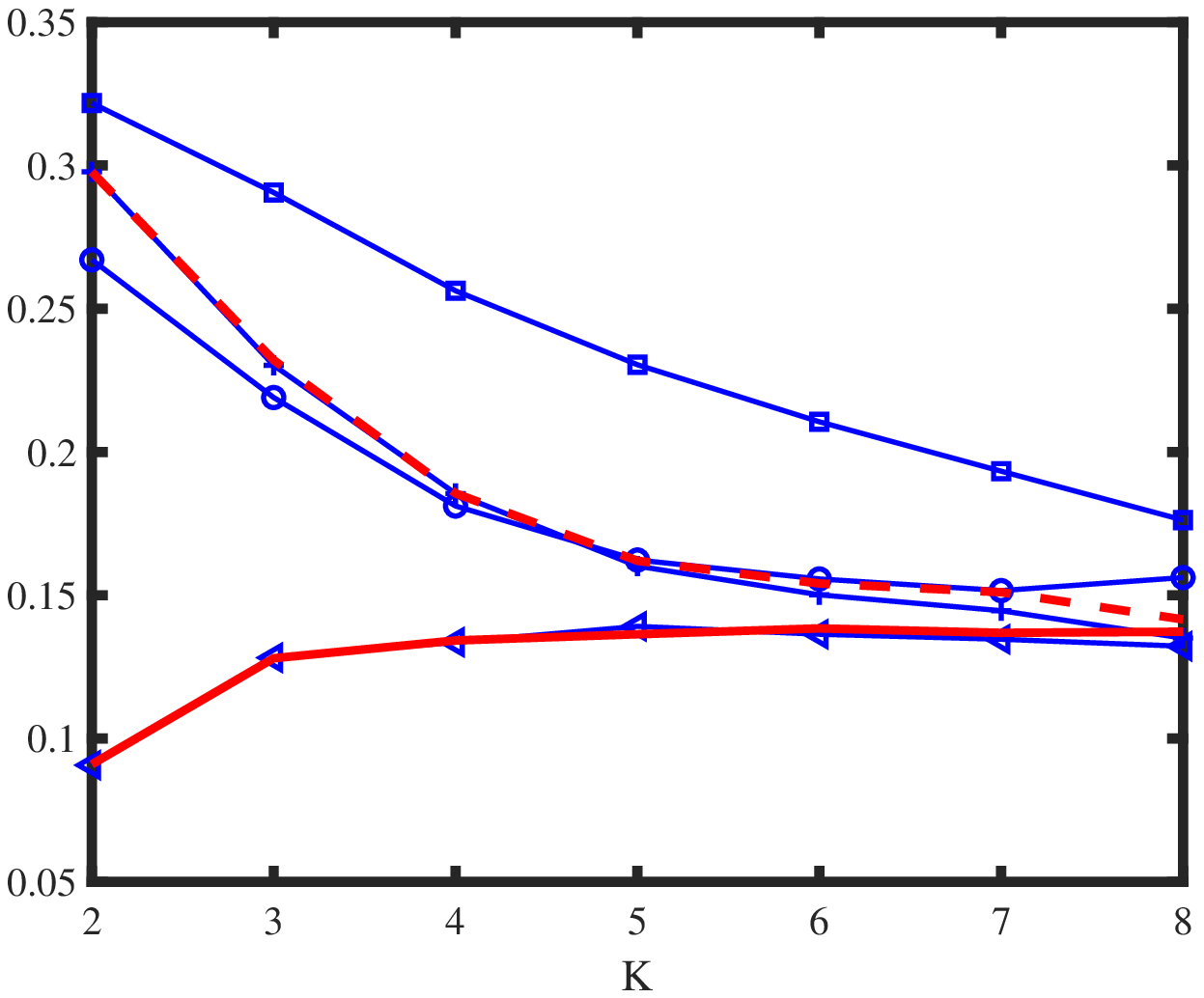}}
	\subfigure[Changing $\lambda_{K}(\tilde{P})$ when $\rho=0.5$]
	{\includegraphics[width=0.45\textwidth]{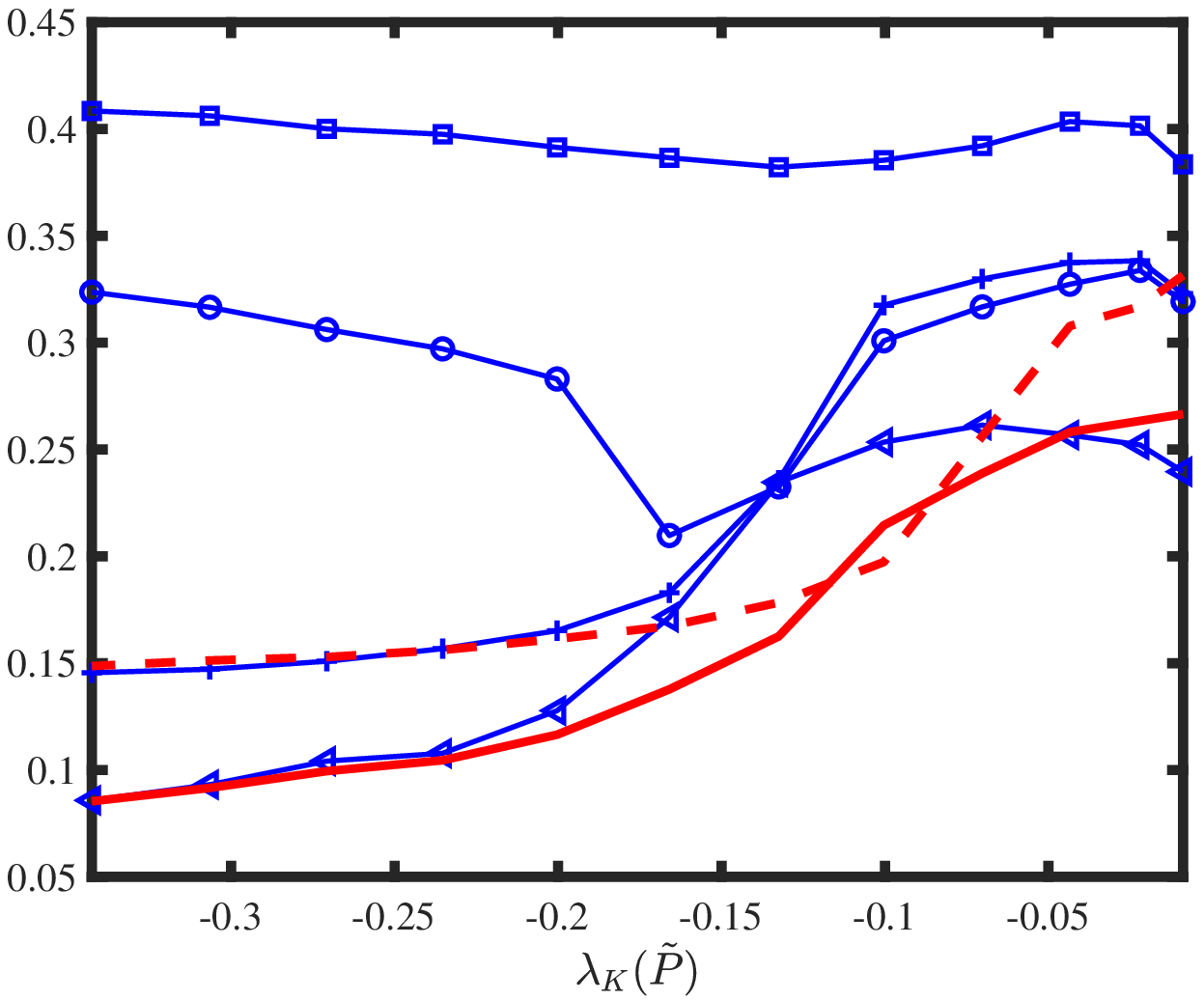}}
	\subfigure[Changing $\lambda_{K}(\tilde{P})$ when $\rho=0.8$]
	{\includegraphics[width=0.45\textwidth]{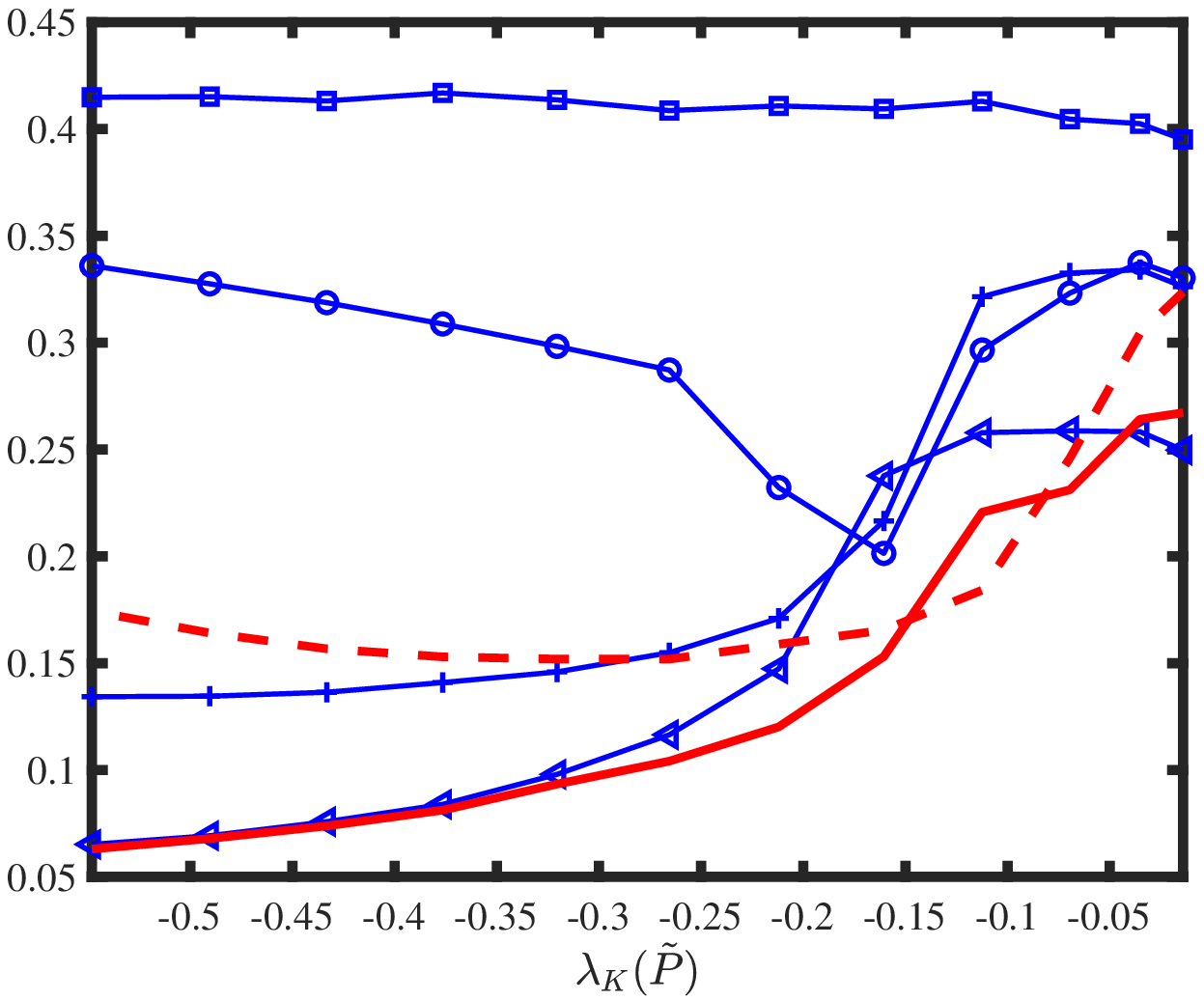}}
	\subfigure[Changing $\rho$ when $n=500$]
	{\includegraphics[width=0.45\textwidth]{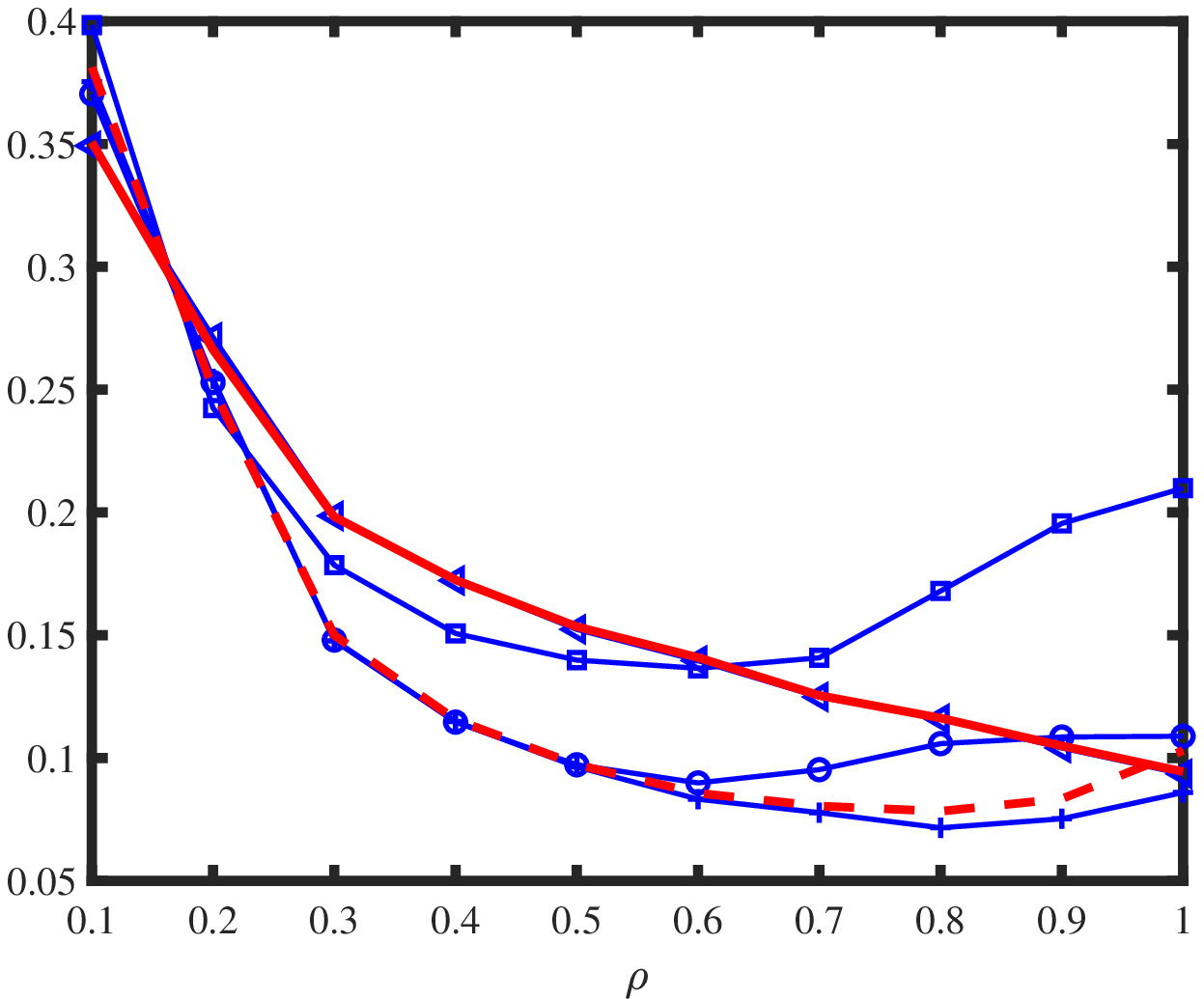}}
	\subfigure[Changing $\rho$ when $n=1000$]
	{\includegraphics[width=0.45\textwidth]{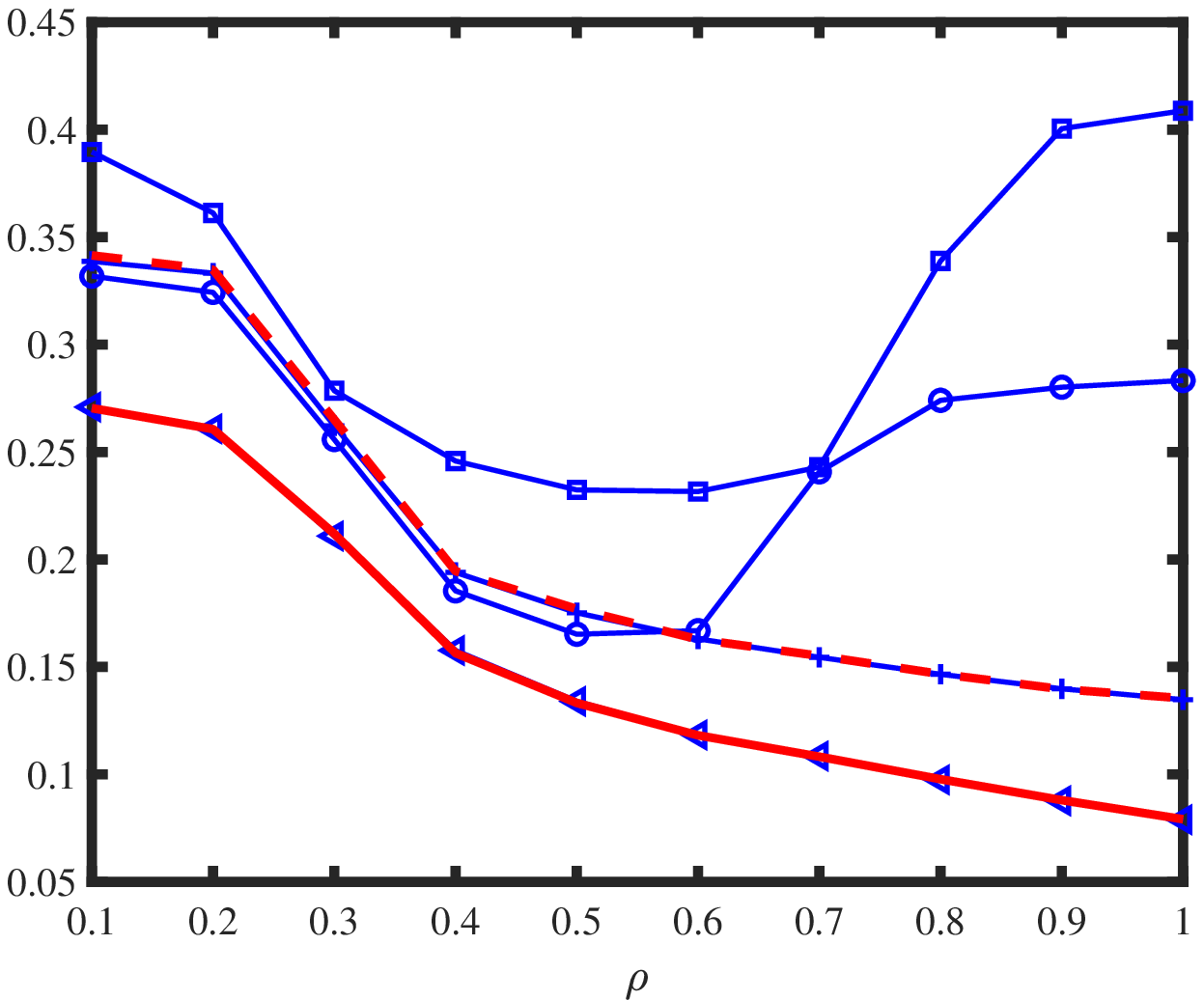}}
	\caption{Estimation errors of Experiments 1-4 (y-axis: $\sum_{i=1}^{n}\|\hat{\Pi}(i,:)-\Pi(i,:)\|_{1}/n$).}
	\label{EX}
\end{figure}

\texttt{Experiment 2: Changing $\lambda_{K}(\tilde{P})$.} Fix $(n, n_{0}, K)=(1000, 100, 3)$ and let $\rho$ be 0.5 or 0.8. We generate $\tilde{P}$ such that the smallest eigenvalue $\lambda_{K}(\tilde{P})$ of $\tilde{P}$ is negative. Set
\[\tilde{P}=\begin{bmatrix}
0.8&0.2&0.1\\
0.2&0.5&0.075*i\\
0.1&0.075*i&0.8\\
\end{bmatrix},
\]
and let $i$ in the range $\{1,2,\ldots, 12\}$. As $i$ grows, $\lambda_{K}(\tilde{P})$ becomes more negative.

The results are displayed in panels (c) and (d) of Figure \ref{EX}. We can find that OCCAM always performs poor as it can not detect networks with negative leading $K$ eigenvalues while other methods can. Meanwhile, we can also see that SRSC is much
better than others while CRSC also enjoys satisfactory performance over the entire parameter range. Especially, when $\lambda_{K}(\tilde{P})$ is close to zero,  all methods perform poorer, and this phenomenon is consistent with the fact that $\lambda_{K}(\tilde{P})$ is in the denominator position of the error bounds in Theorem 5.9.

\texttt{Experiment 3: Changing Sparsity parameter $\rho$.} Fix $(n_{0}, K)=(100, 3)$ and let $n$ be 500 or 1000. We vary $\rho$ in the range $\{0.1, 0.2, \ldots, 1\}$.  The bottom two panels of Figure \ref{EX} records the numerical results of this experiment. When $n=500$, the error of our CRSC is smaller than or similar to that of the best performing algorithm among the others; when $n=1000$, our SRSC performs similar as SPACL and they outperform other methods while our CRSC has better performance than Mixed-SCORE and OCCAM. Meanwhile, OCCAM and Mixed-SCORE have abnormal behaviors when $n=500$ and $n=1000$ such that they perform poorer when $\rho$ is larger than 0.6. This interesting phenomenon suggests that our SRSC and CRSC are more stable than Mixed-SCORE and OCCAM on the sparsity of the network since a larger $\rho$ creates a denser network.
\section{Real Data}
\subsection{Application to SNAP ego-networks}\label{SNAPego}
The SNAP ego-networks dataset contains substantial ego-networks from three platforms Facebook, GooglePlus, and Twitter. In an ego network, all nodes are friends of one central user, and the friendship groups set by the central user can be used as ground truth communities \citep{OCCAM}. Since one node may be friends of more than one central user, the node can be seen as having mixed memberships. With the known membership information, we can use SNAP ego-networks to test the performances of our methods. We obtain the SNAP ego-networks parsed by Yuan Zhang (the first author of the OCCAM method \citep{OCCAM}). For an ego-network, since the true mixed membership matrix $\Pi$ only consists entries 0 and 1 (i.e., the true mixed membership matrix of an ego-network only tells us whether a node belong to certain community or not), we set $\Pi(i,:)=\frac{\Pi(i,:)}{\|\Pi(i,:)\|_{1}}$ to make the row-summation of $\Pi$ be one for $1\leq i\leq n$.  The parsed SNAP ego-networks are slightly different from those used in \cite{OCCAM}, for readers reference, we report the following summary statistics for each network: (1) average number of nodes $n$ and average number of communities $K$. (2) average node degree $\bar{d}$ where $\bar{d}=\sum_{i=1}^{n}D(i,i)/n$. (3) density $\sum_{i,j}A(i,j)/(n(n-1))$, i.e., the overall edge probability. (4) the proportion of overlapping nodes $r_{o}$, i.e., $r_{o}=\frac{\mathrm{number~of~nodes~with~mixed~membership}}{n}$. We report the means and standard deviations of these measures for each of the social networks in Table \ref{dataSNAP}.
\begin{table}[h!]
	\centering
	\caption{Mean (SD) of summary statistics for ego-networks.}
	\label{dataSNAP}
	\begin{tabular}{cccccccccc}
		\hline\hline
		&\#Networks&$n$&$K$&$\bar{d}$&Density&$r_{o}$\\
		\hline
		Facebook&7&236.57&3&30.61&0.15&0.0901\\
		&-&(228.53)&(1.15)&(29.41)&(0.058)&(0.1118)\\
		\hline
		GooglePlus&58&433.22&2.22&66.81&0.18&0.0713\\
		&-&(327.70)&(0.46)&(65.2)&(0.11)&(0.0913)\\
		\hline
		Twitter&255&60.64&2.63&17.87&0.33&0.0865\\
		&-&(30.77)&(0.83)&(9.97)&(0.17)&(0.1185)\\
		\hline\hline
	\end{tabular}
\end{table}


To compare methods, we report the average performance over each of the social platforms and the corresponding standard deviation in Table \ref{ErrorSNAP}, where $\tau$ for SRSC and CRSC is set as $0.1\mathrm{log}(n)$  here. From the results, we can find that SRSC performs similar as CRSC on the these SNAP-ego networks. Unlike the  simulation results where Mixed-SCORE, SVM-cone-DCMMSB and SPACL sometimes may perform similar as our SRSC and CRSC, when come to the empirical datatsets, we see that  our SRSC and CRSC always outperform their competitors on the GooglePlus and Twitter  platforms networks while SPACL slightly performs better than our SRSC and CRSC on the Facebook datasets. Since there are only 7 networks in the Facebook datasets among all the SNAP-ego networks, we conclude that our SRSC and CRSC enjoy superior performances on the SNAP-ego networks than their competitors.   From Table \ref{dataSNAP}, we see that $\bar{d}$ is much smaller than the network size $n$, suggesting that most SNAP-ego networks are sparse. Our SRSC and CRSC enjoy better performances on empirical networks because the two methods are designed based on regularized Laplacian matrix which can successfully detect sparse networks.
\begin{table}[h!]
	\centering
	\caption{Mean (SD) of mixed-Hamming error rates for ego-networks.}
	\label{ErrorSNAP}
	\begin{tabular}{cccccccccc}
		\hline\hline
		&Facebook&GooglePlus&Twitter\\
		\hline
		Mixed-SCORE&0.2496(0.1322)&0.3766(0.1053)&0.3088(0.1296)\\
		OCCAM&0.2610(0.1367)&0.3564(0.1210)&0.2864(0.1406)\\
		SVM-cone-DCMMSB&0.2483(0.1496)&0.3563(0.1047)&0.2985(0.1327)\\
		SPACL&\textbf{0.2408}(0.1264)&0.3645(0.1087)&0.3056(0.1271)\\
		\hline
		SRSC&0.2513(0.1290)&0.3239(0.1286)&\textbf{0.2626}(0.1341)\\
		CRSC&0.2475(0.1358)&\textbf{0.3192}(0.1265)&0.2632(0.1388)\\
		\hline\hline
	\end{tabular}
\end{table}

\subsection{Application to Coauthorship network}
\cite{ji2016coauthorship} collected a coauthorship network data set for statisticians, based on all published papers in AOS, Biometrika, JASA, JRSS-B, from 2003 to the first half of 2012. In this network, an edge is constructed between two authors if they have coauthored at least two papers in the range of the data set. As suggested by \cite{mixedSCORE}, there are two communities called ``Carroll-Hall'' and ``North Carolina'' over 236 nodes (i.e., $n=236, K=2$ for Coauthorship network), and authors in this network have mixed memberships in these two communities, for detail introduction of the Coauthorship network, refer to \cite{ji2016coauthorship}. We find that the average degree $\bar{d}$ for the Cosuthorship network is 2.5085, which is much smaller than 236, suggesting that the Coauthorship network is sparse. Based on this observation, we argue that methods which can deal with sparser networks for mixed membership community detection may provide some new insights on the analysis of the Coauthorship network.

Since there is no ground truth of the nodes membership for the Coauthorship network \citep{ji2016coauthorship, mixedSCORE}, similar as that in \cite{mixedSCORE}, we only provide the estimated PMFs of the ``Carroll-Hall'' community \footnote{The respective estimated PMF of the ``North Carolina'' community for an author just equals 1 minus the author's weight of the ``Carroll-Hall'' community.} for 20 authors, where the 20 authors are also studied in Table 4 in \cite{mixedSCORE} and 19 of them (except Jiashun Jin) are regarded with highly mixed memberships in \cite{mixedSCORE}.  The results are in Table \ref{Coauthorship}.
\begin{table}
	\centering
	\caption{Estimated PMF of the ``Carroll-Hall'' community for the Coauthorship network.}
	\label{Coauthorship}
	\resizebox{\columnwidth}{!}{
	\begin{tabular}{lccccccccccccc}
		\hline\hline
		Methods &SRSC&CRSC&Mixed-SCORE&OCCAM&SVM-cone-DCMMSB&SPACL\\
		\hline
		Jianqing Fan&79.69\%&95.60\%&56.21\% &65.51\%&50.17\%&68.26\%\\
		Jason P Fine&94.57\%&99.61\%&56.79\% &65.15\%&49.72\%&68.51\%\\
		Michael R Kosorok&93.19\%&99.28\%&62.45\% &61.55\%&45.33\%&70.94\%\\
		J S Marron&90.29\%&98.56\%&41.00\% &74.06\%&62.11\%&61.51\%\\
		Hao Helen Zhang&89.76\%&98.42\%&48.45\% &70.05\%&56.23\%&64.86\%\\
		Yufeng Liu&88.76\%&98.16\%&46.03\% &71.39\%&58.14\%&63.78\%\\
		Xiaotong Shen&90.29\%&98.56\%&41.00\% &74.06\%&62.11\%&61.51\%\\
		Kung-Sik Chan&84.97\%&97.14\%&84.62\% &73.37\%&61.05\%&62.11\%\\
		Yichao Wu&85.26\%&97.22\%&51.42\% &68.35\%&53.90\%&66.17\%\\
		Yacine Ait-Sahalia&81.46\%&96.13\%&51.69\% &68.20\%&53.69\%&66.29\%\\
		Wenyang Zhang&81.59\%&96.17\%&51.69\% &68.20\%&53.69\%&66.29\%\\
		Howell Tong&83.13\%&96.62\%&47.34\% &70.66\%&57.10\%&64.36\%\\
		Chunming Zhang&80.76\%&95.93\%&52.03\% &68.00\%&53.43\%&66.44\%\\
		Yingying Fan&75.40\%&94.24\%&44.17\% &72.39\%&59.60\%&62.94\%\\
		Rui Song&85.68\%&97.34\%&52.65\% &67.64\%&52.94\%&66.71\%\\
		Per Aslak Mykland&82.49\%&96.44\%&47.43\% &70.62\%&57.04\%&64.40\%\\
		Bee Leng Lee&94.10\%&99.50\%&57.51\% &64.71\%&49.16\%&68.82\%\\
		Runze Li&92.66\%&99.15\%&88.82\% &41.08\%&25.32\%&81.73\%\\
		Jiancheng Jiang&70.49\%&92.54\%&29.41\% &79.72\%&71.37\%&56.14\%\\
		Jiashun Jin&100\%&100\%&100.00\% &0.00\%&0.00\%&99.72\%\\
		\hline\hline
	\end{tabular}}
\end{table}

From Table \ref{Coauthorship}, we can find that SRSC, CRSC, OCCAM and SPACL tend to classify authors in this table into the ``Carroll-Hall'' community (except Runze Li and Jiashun Jin for OCCAM method, which puts the two authors into the ``North Carolina'' community.), and such classification is quite different from that of Mixed-SCORE. There are huge differences of the estimated PMFs between SRSC (or CRSC, or SPACL) and Mixed-SCORE on the following 9 authors:  J S Marron, Hao Helen Zhang, Yufeng Liu, Xiaotong Shen, Kung-Sik Chan, Howell Tong, Yingying Fan, Per Aslak Mykland, Jiancheng Jiang. We analyze Yingying Fan and Jiancheng Jiang in detail based on papers published by them in the top 4 journals during the time period of the Coauthorship network dataset.
\begin{itemize}
	\item For Yingying Fan, she published 6 papers on the top 4 journals while she coauthored with Jianqing Fan with 4 papers. Therefore, we tend to believe that Yingying Fan is more on the ``Carroll-Hall'' community since Jianqing Fan is more on this community.
	\item For Jiancheng Jiang, he published 10 papers on the top 4 journals while he coauthored with Jianqing Fan with 9 papers. Therefore, we tend to believe that Jiancheng Jiang is more on the ``Carroll-Hall'' community since Jianqing Fan is more on this community.
\end{itemize}

\section{Discussion}\label{sec7}
In this paper, we study the impact of regularized Laplacian matrix on spectral clustering by proposing two consistent regularized spectral clustering algorithms SRSC and CRSC to mixed membership community detection under the MMSB model. The simplex structure and cone structure from the variants for the eigen-decomposition of the population regularized Laplacian matrix are new and they are the key components for the design of our two algorithms.  We show the consistencies of the estimations of SRSC and CRSC under MMSB. By introducing the parametric probability and carefully analyzing the bound of $\|L_{\tau}-\mathscr{L}_{\tau}\|$ as well as the theoretical error bounds of SRSC and CRSC,  we give a reasonable explanation on the optimal choice of the regularizer $\tau$ theoretically. Especially, based on the parametric probability, we show why choosing an intermediate regularization parameter is preferred. In contrast to prior work, our theoretical results match the classical separation condition of the standard network with two equal size clusters and the sharp threshold of the Erdos-Renyi random graph $G(n,p)$. Numerically, SRSC and CRSC enjoy competitive performances with most of the benchmark methods in both simulated and empirical data. To our knowledge, this is the first work to study the impact of regularization on spectral clustering for mixed membership community detection problems under MMSB. Meanwhile, this is also the first work to give a reasonable explanation on the optimal choice of the regularization parameter $\tau$.

Our idea of analyzing the variants of the eigen-decomposition of the population regularized Laplacian matrix can be extended in many ways. In a forthcoming manuscript, we extend this idea to study the impact of regularization on spectral clustering under the degree-corrected mixed membership (DCMM) model proposed by \cite{mixedSCORE}. In another forthcoming manuscript, we investigate the impact of regularization on spectral clustering for the topic estimation problem in text mining \citep{blei2003latent,TSCORE}.

In \cite{ali2018improved}, the authors studied the existence of an optimal value $\alpha_{opt}$ of the parameter $\alpha$ for community detection methods based on $D^{-\alpha}AD^{-\alpha}$. Recall that our SRSC and CRSC are designed based on $D^{-1/2}_{\tau}AD^{-1/2}_{\tau}$, we argue that whether there exist optimal $\alpha_{0}$ and $\beta_{0}$ as well as optimal regularizer $\tau_{\mathrm{opt}}$ such that mixed membership community detection algorithm designed based on $D^{\alpha_{0}}_{\tau_{\mathrm{opt}}}A^{\beta_{0}}D^{\alpha_{0}}_{\tau_{\mathrm{opt}}}$ outperforms methods designed based on $D^{\alpha}_{\tau}A^{\beta}D^{\alpha}_{\tau}$ for any choices of $\alpha, \beta$ and $\tau$. For this problem, the idea of parametric probability introduced in this paper may be a powerful technique to give the optimal choices.  For reasons of space, we leave studies of this problem to the future.

\section*{Acknowledgements}
The authors would like to thank Dr. Zhang Yuan  (the first author of the OCCAM method \citep{OCCAM}) for sharing the SNAP ego-networks with us.

\bibliographystyle{agsm}

\bibliography{reference}   

\pagebreak
\begin{center}
{\large\bf SUPPLEMENTARY MATERIAL}
\end{center}
\appendix
In this document, we provide the technical proofs of lemmas and theorems in the main manuscript. And we review One-Class SVM and SVM-cone algorithm in section E.
\section{Ideal Simplex, Ideal Cone and Equivalence}
\subsection{Proof of Lemma 3.1}
\begin{proof}
	Since $\mathcal{I}$ is the indices of rows corresponding to $K$ pure nodes, one from each community, W.L.O.G., reorder the nodes so that $\Pi(\mathcal{I},:)=I$. Since $\mathscr{L}_{\tau}=\mathscr{D}^{-1/2}_{\tau}\Pi P\Pi'\mathscr{D}^{-1/2}_{\tau}=VEV'$, we have $V(\mathcal{I},:)EV'=\mathscr{D}^{-1/2}_{\tau}(\mathcal{I},\mathcal{I})\Pi(\mathcal{I},:)P \Pi'\mathscr{D}^{-1/2}_{\tau}=\mathscr{D}^{-1/2}_{\tau}(\mathcal{I},\mathcal{I})P \Pi'\mathscr{D}^{-1/2}_{\tau}$. Now $VE=\mathscr{L}_{\tau}V=\mathscr{D}^{-1/2}_{\tau}\Pi P\Pi'\mathscr{D}^{-1/2}_{\tau}V=\mathscr{D}^{-1/2}_{\tau}\Pi (P\Pi'\mathscr{D}^{-1/2}_{\tau})V=\mathscr{D}^{-1/2}_{\tau}\Pi (\mathscr{D}^{1/2}_{\tau}(\mathcal{I},\mathcal{I})V(\mathcal{I},:)EV')V=\mathscr{D}^{-1/2}_{\tau}\Pi \mathscr{D}^{1/2}_{\tau}(\mathcal{I},\mathcal{I})V(\mathcal{I},:)E$, right multiplying $E^{-1}$ gives $V=\mathscr{D}^{-1/2}_{\tau}\Pi\mathscr{D}^{1/2}_{\tau}(\mathcal{I},\mathcal{I})V(\mathcal{I},:)$. Hence, we have $V_{\tau,1}=\mathscr{D}^{1/2}_{\tau}V=\Pi \mathscr{D}^{1/2}_{\tau}(\mathcal{I},\mathcal{I})V(\mathcal{I},:)=\Pi(\mathscr{D}^{1/2}_{\tau}V)(\mathcal{I},:)=\Pi V_{\tau,1}(\mathcal{I},:)$.
	
	Since $V(i,:)=e'_{i}V=e'_{i}\mathscr{D}^{-1/2}_{\tau}\Pi \mathscr{D}^{1/2}_{\tau}(\mathcal{I},:)V(\mathcal{I},:)=\mathscr{D}^{-1/2}_{\tau}(i,i)e'_{i}\Pi \mathscr{D}^{1/2}_{\tau}(\mathcal{I},:)V(\mathcal{I},:)=\mathscr{D}^{-1/2}_{\tau}(i,i)\Pi(i,:)\mathscr{D}^{1/2}_{\tau}(\mathcal{I},:)V(\mathcal{I},:)$ and $\mathscr{D}_{\tau}(i,i)=\tau+\mathscr{D}(i,i)=\tau+\sum_{m=1}^{n}\Omega(i,m)=\tau+\sum_{m=1}^{n}\Pi(i,:)P\Pi'(m,:)=\tau+\Pi(i,:)P\sum_{m=1}^{n}\Pi'(m,:)$,  we have $\mathscr{D}^{-1/2}_{\tau}(i,i)=\mathscr{D}^{-1/2}_{\tau}(j,j)$ if $\Pi(i,:)=\Pi(j,:)$ . Therefore, when $\Pi(i,:)=\Pi(j,:)$, we have $V(i,:)=V(j,:)$, which gives that $V_{\tau,1}(i,:)=V_{\tau,1}(j,:)$ when $\Pi(i,:)=\Pi(j,:)$.
\end{proof}
\subsection{Proof of Lemma 3.3}
\begin{proof}
	For convenience, set $M_{1}=\Pi\mathscr{D}^{1/2}_{\tau}(\mathcal{I},\mathcal{I})V(\mathcal{I},:)$. By the proof of Lemma 3.1, we know that
	\begin{align*}
	V=\mathscr{D}^{-1/2}_{\tau}\Pi\mathscr{D}^{1/2}_{\tau}(\mathcal{I},\mathcal{I})V(\mathcal{I},:),
	\end{align*}
	which gives $V=\mathscr{D}^{-1/2}_{\tau}M_{1}$. Hence, we have $V(i,:)=\mathscr{D}^{-1/2}_{\tau}(i,i)M_{1}(i,:)$. Therefore, $V_{*,1}(i,:)=\frac{V(i,:)}{\|V(i,:)\|_{F}}=\frac{M_{1}(i,:)}{\|M_{1}(i,:)\|_{F}}$, which gives that
	\begin{flalign*}
	V_{*,1}&=\begin{bmatrix}
	\tiny
	M_{1}(1,:)/\|M_{1}(1,:)\|_{F}\\
	M_{1}(2,:)/\|M_{1}(2,:)\|_{F}\\
	\vdots\\
	M_{1}(n,:)/\|M_{1}(n,:)\|_{F}
	\end{bmatrix}=\begin{bmatrix}
	\frac{1}{\|M_{1}(1,:)\|_{F}} &  & & \\
	& \frac{1}{\|M_{1}(2,:)\|_{F}}& &\\
	& & \ddots&\\
	&&&\frac{1}{\|M_{1}(n,:)\|_{F}}
	\end{bmatrix}M_{1}\\
	&=\begin{bmatrix}
	\frac{1}{\|M_{1}(1,:)\|_{F}} &  & & \\
	& \frac{1}{\|M_{1}(2,:)\|_{F}}& &\\
	& & \ddots&\\
	&&&\frac{1}{\|M_{1}(n,:)\|_{F}}
	\end{bmatrix}\Pi \mathscr{D}^{1/2}_{\tau}(\mathcal{I},\mathcal{I})V(\mathcal{I},:)
	=\begin{bmatrix}
	\Pi(1,:)/\|M_{1}(1,:)\|_{F}\\
	\Pi(2,:)/\|M_{1}(2,:)\|_{F}\\
	\vdots\\
	\Pi(n,:)/\|M_{1}(n,:)\|_{F}
	\end{bmatrix}\mathscr{D}^{1/2}_{\tau}(\mathcal{I},\mathcal{I})V(\mathcal{I},:)\\
	&=\begin{bmatrix}
	\Pi(1,:)/\|M_{1}(1,:)\|_{F}\\
	\Pi(2,:)/\|M_{1}(2,:)\|_{F}\\
	\vdots\\
	\Pi(n,:)/\|M_{1}(n,:)\|_{F}
	\end{bmatrix}\mathscr{D}^{1/2}_{\tau}(\mathcal{I},\mathcal{I})N_{V}^{-1}(\mathcal{I},\mathcal{I})N_{V}(\mathcal{I},\mathcal{I})V(\mathcal{I},:)\\
	&=\begin{bmatrix}
	\Pi(1,:)/\|M_{1}(1,:)\|_{F}\\
	\Pi(2,:)/\|M_{1}(2,:)\|_{F}\\
	\vdots\\
	\Pi(n,:)/\|M_{1}(n,:)\|_{F}
	\end{bmatrix}\mathscr{D}^{1/2}_{\tau}(\mathcal{I},\mathcal{I})N_{V}^{-1}(\mathcal{I},\mathcal{I})V_{*}(\mathcal{I},:).
	\end{flalign*}
	Therefore, we have
	\begin{align*}
	Y_{1}=\begin{bmatrix}
	\Pi(1,:)/\|M_{1}(1,:)\|_{F}\\
	\Pi(2,:)/\|M_{1}(2,:)\|_{F}\\
	\vdots\\
	\Pi(n,:)/\|M_{1}(n,:)\|_{F}
	\end{bmatrix}\mathscr{D}^{1/2}_{\tau}(\mathcal{I},\mathcal{I})N_{V}^{-1}(\mathcal{I},\mathcal{I})=N_{M_{1}}\Pi\mathscr{D}^{1/2}_{\tau}(\mathcal{I},\mathcal{I})N_{V}^{-1}(\mathcal{I},\mathcal{I}),
	\end{align*}
	where $N_{M_{1}}=\begin{bmatrix}
	\frac{1}{\|M_{1}(1,:)\|_{F}} &  & & \\
	& \frac{1}{\|M_{1}(2,:)\|_{F}}& &\\
	& & \ddots&\\
	&&&\frac{1}{\|M_{1}(n,:)\|_{F}}
	\end{bmatrix}$.
	Sure, all entries of $Y_{1}$ are nonnegative. And since we assume that each community has at least one pure node, no row of $Y_{1}$ is 0.
	
	Then we prove that $V_{*,1}(i,:)=V_{*,1}(j,:)$ when $\Pi(i,:)=\Pi(j,:)$. For $1\leq i\leq n$, we have
	\begin{flalign*}
	V_{*,1}(i,:)&=e'_{i}V_{*,1}=e'_{i}\begin{bmatrix}
	\frac{1}{\|M_{1}(1,:)\|_{F}} &  & & \\
	& \frac{1}{\|M_{1}(2,:)\|_{F}}& &\\
	& & \ddots&\\
	&&&\frac{1}{\|M_{1}(n,:)\|_{F}}
	\end{bmatrix}M_{1}=\frac{1}{\|M_{1}(i,:)\|_{F}}e'_{i}M_{1}\\
	&=\frac{1}{\|e'_{i}M_{1}\|_{F}}e'_{i}M_{1}=\frac{1}{\|e'_{i}\Pi\mathscr{D}^{1/2}_{\tau}(\mathcal{I},\mathcal{I})V(\mathcal{I},:)\|_{F}}e'_{i}\Pi\mathscr{D}^{1/2}_{\tau}(\mathcal{I},\mathcal{I})V(\mathcal{I},:)\\
	&=\frac{1}{\|\Pi(i,:)\mathscr{D}^{1/2}_{\tau}(\mathcal{I},\mathcal{I})V(\mathcal{I},:)\|_{F}}\Pi(i,:)\mathscr{D}^{1/2}_{\tau}(\mathcal{I},\mathcal{I})V(\mathcal{I},:),
	\end{flalign*}
	which gives that if $\Pi(j,:)=\Pi(i,:)$, we have $V_{*,1}(i,:)=V_{*,1}(j,:)$.
\end{proof}
\subsection{Proof of Lemma 3.4}
\begin{proof}
	Since $I=V'V=V'(\mathcal{I},:)\mathscr{D}^{1/2}_{\tau}(\mathcal{I},\mathcal{I})\Pi'\mathscr{D}^{-1}_{\tau}\Pi\mathscr{D}^{1/2}_{\tau}(\mathcal{I},\mathcal{I})V(\mathcal{I},:)$ and the inverse of $V(\mathcal{I},:)$ exists, we have $(V(\mathcal{I},:)V'(\mathcal{I},:))^{-1}=\mathscr{D}^{1/2}_{\tau}(\mathcal{I},\mathcal{I})\Pi'\mathscr{D}^{-1}_{\tau}\Pi\mathscr{D}^{1/2}_{\tau}(\mathcal{I},\mathcal{I})$.
	
	Since $V_{*,1}(\mathcal{I},:)=N_{V}(\mathcal{I},\mathcal{I})V(\mathcal{I},:)$, we have
	\begin{align*}
	(V_{*,1}(\mathcal{I},:)V'_{*,1}(\mathcal{I},:))^{-1}=N_{V}^{-1}(\mathcal{I},\mathcal{I})\mathscr{D}^{1/2}_{\tau}(\mathcal{I},\mathcal{I})\Pi'\mathscr{D}^{-1}_{\tau}\Pi\mathscr{D}^{1/2}_{\tau}(\mathcal{I},\mathcal{I})N_{V}^{-1}(\mathcal{I},\mathcal{I}).
	\end{align*}
	Since all entries of $N_{V}^{-1}(\mathcal{I},\mathcal{I}), \Pi, \mathscr{D}_{\tau}$ and nonnegative and $N_{V},\mathscr{D}_{\tau}$ are diagonal matrices, we see that all entries of $(V_{*,1}(\mathcal{I},:)V'_{*,1}(\mathcal{I},:))^{-1}$ are nonnegative and its diagonal entries are strictly positive, hence we have $(V_{*,1}(\mathcal{I},:)V'_{*,1}(\mathcal{I},:))^{-1}\mathbf{1}>0$.
\end{proof}
\subsection{Proof of Lemma 4.2}
\begin{proof}
	Since $V_{\tau,2}=\mathscr{D}^{1/2}_{\tau}VV'=V_{\tau,1}V'$, we have $V_{\tau,2}(\mathcal{I},:)=V_{\tau,1}(\mathcal{I},:)V'$, combine it with $V_{\tau,1}=\Pi V_{\tau,1}(\mathcal{I},:)$ by Lemma 3.1, we have $V_{\tau,2}= \Pi V_{\tau,1}(\mathcal{I},:)V'=\Pi V_{\tau,2}(\mathcal{I},:)$. Therefore, $V_{\tau,2}=\Pi V_{\tau,2}(\mathcal{I},:)$. Meanwhile, $V_{\tau,2}(i,:)=e'_{i}V_{\tau,2}=e'_{i}V_{\tau,1}V'=V_{\tau,1}(i,:)V'$, since $V_{\tau,1}(i,:)=V_{\tau,1}(j,:)$ when $\Pi(i,:)=\Pi(j,:)$, the conclusion holds.
\end{proof}
\subsection{Proof of Lemma 4.4}
\begin{proof}
	Set $M_{2}=\Pi\mathscr{D}^{1/2}_{\tau}(\mathcal{I},\mathcal{I})V_{2}(\mathcal{I},:)$. Since $V_{2}=\mathscr{D}^{-1/2}_{\tau}\Pi\mathscr{D}^{1/2}_{\tau}(\mathcal{I},\mathcal{I})V_{2}(\mathcal{I},:)$,
	we have $V_{2}=\mathscr{D}^{-1/2}_{\tau}M_{2}$.
	Follow a similar proof of Lemma 3.3, we have
	$Y_{2}=N_{M_{2}}\Pi\mathscr{D}^{-1/2}_{\tau}(\mathcal{I},\mathcal{I})N_{V_{2}}^{-1}(\mathcal{I},\mathcal{I})$,
	where $N_{M_{2}}$ is an $n\times n$ diagonal matrix whose $i$-th diagonal entry is $\frac{1}{\|M_{2}(i,:)\|_{F}}$. Meanwhile, all entries of $Y_{2}$ are nonnegative and no row of $Y_{2}$ is 0. The last statement can be proved easily by following similar proof as the one in Lemma 3.3 and we omit it here.
\end{proof}
\subsection{Proof of Lemma 4.5}
\begin{proof}
	Since in the proof of Lemma G.1 \citep{MaoSVM}, we find that $\hat{V}$ and $\hat{V}\hat{V}'$ are model-independent as long as $\hat{V}$ contains the leading $K$ eigenvectors with unit-norm of a symmetric matrix, hence the outputs of the SVM-cone algorithm using $\hat{V}_{*,1}$ and $\hat{V}_{*,2}$ as inputs are same as proved by Lemma G.1 \citep{MaoSVM}.
	
	Now, we prove the part for SP algorithm. First, we write down the SP algorithm as below.
	\begin{algorithm}
		\caption{\textbf{Successive Projection (SP)} \citep{gillis2015semidefinite}}
		\label{alg:SP}
		\begin{algorithmic}[1]
			\Require Near-separable matrix $M=HW+(M-HW)\in\mathbb{R}^{n\times m}_{+}$ , where $W, H$ should satisfy Assumption 1 \cite{gillis2015semidefinite}, the number $r$ of columns to be extracted.
			\Ensure Set of indices $\mathcal{K}$ such that $M(\mathcal{K},:)\approx W$ (up to permutation)
			\State Let $R=M, \mathcal{K}=\{\}, k=1$.
			\State \textbf{While} $R\neq 0$ and $k\leq r$ \textbf{do}
			\State ~~~~~~~$k_{*}=\mathrm{argmax}_{k}\|R(k,:)\|_{F}$.
			\State ~~~~~~$u_{k}=R(k_{*},:)$.
			\State ~~~~~~$R\leftarrow (I-\frac{u_{k}u'_{k}}{\|u_{k}\|^{2}_{F}})R$.
			\State ~~~~~~$\mathcal{K}=\mathcal{K}\cup \{k_{*}\}$.
			\State ~~~~~~k=k+1.
			\State \textbf{end while}
		\end{algorithmic}
	\end{algorithm}
	For convenience, call $H\equiv I-\frac{u_{k}u'_{k}}{\|u_{k}\|^{2}_{F}}$ as the operator matrix.
	
	Set $M_{1}=\hat{V}_{\tau,1}, H_{1}=\Pi, W_{1}=V_{\tau,1}(\mathcal{I},:)$. By conditions (I) and (II), since $\hat{V}_{\tau,1}=V_{\tau,1}+(\hat{V}_{\tau,1}-V_{\tau,1})=\Pi V_{\tau,1}(\mathcal{I},:)+(\hat{V}_{\tau,1}-V_{\tau,1})$, sure we have $H_{1}, W_{1}$ satisfy Assumption 1 in \cite{gillis2015semidefinite}. Hence we can apply SP algorithm on $\hat{V}_{\tau,1}$.
	Set $M_{2}=\hat{V}_{\tau,2}, H_{2}=\Pi, W_{2}=V_{\tau,2}(\mathcal{I},:)$. By conditions (I) and (II), since $\hat{V}_{\tau,2}=V_{\tau,2}+(\hat{V}_{\tau,2}-V_{\tau,2})=\Pi V_{\tau,2}(\mathcal{I},:)+(\hat{V}_{\tau,2}-V_{\tau,2})$, sure we have $H_{2}, W_{2}$ satisfy Assumption 1 \cite{gillis2015semidefinite}. Hence we can apply SP algorithm on $\hat{V}_{\tau,2}$.
	
	To prove Lemma 4.5, we follow a similar proof as Lemma 3.4 in \cite{mao2020estimating}, i.e., we use the induction method to prove this lemma. For step $k=1$: when the input in SP is $\hat{V}_{\tau}$, set $R_{1}=\hat{V}_{\tau}$, we have
	\begin{align*}
	k_{*,1}=\mathrm{max}_{k}\|R_{1}(k,:)\|_{F}=\mathrm{max}_{k}\|\hat{V}_{\tau}(k,:)\|_{F}.
	\end{align*}
	When the input in SP is $\hat{V}_{2,\tau}$, set $R_{2}=\hat{V}_{2,\tau}$. Since $\hat{V}_{\tau,2}=\hat{V}_{\tau}\hat{V}'$, we have
	\begin{align*}
	k_{*,2}=\mathrm{max}_{k}\|R_{2}(k,:)\|_{F}=\mathrm{max}_{k}\|\hat{V}_{2,\tau}(k,:)\|_{F}=\mathrm{max}_{k}\|e'_{k}\hat{V}_{\tau}\hat{V}'\|_{F}=k_{*,1}.
	\end{align*}
	Hence, SP algorithm will give the same index at step 1, and we denote it as $k_{1}\equiv k_{*,1}\equiv k_{*,2}$.
	
	Meanwhile, set $H_{1,1}=(I-\frac{e'_{k_{1}}R_{1}R'_{1}e_{k_{1}}}{\|e'_{k_{1}}R_{1}\|^{2}_{F}})$ and $H_{1,2}=(I-\frac{e'_{k_{1}}R_{2}R'_{2}e_{k_{1}}}{\|e'_{k_{1}}R_{2}\|^{2}_{F}})$, we have $H_{1,2}=(I-\frac{e'_{k_{1}}R_{1}\hat{V}'\hat{V}R'_{1}e_{k_{1}}}{\|e'_{k_{1}}R_{1}\hat{V}\|^{2}_{F}})=H_{1,1}$. Hence, SP algorithm will give the same operator matrix when updating $R$, and we denote the operator matrix when $k=1$ as $H_{1}\equiv H_{1,1}\equiv H_{1,2}$. Then, when $k=1$, $R_{1}$ and $R_{2}$ (note that when $k=1$, $R_{2}=R_{1}\hat{V}'$) are updated as below
	\begin{align*}
	&R_{1}\leftarrow H_{1}R_{1},\\
	&R_{2}\leftarrow H_{1}R_{2},
	\end{align*}
	which gives that the updated $R_{1}$ and $R_{2}$ stills have the relationship that $R_{2}=R_{1}\hat{V}'$.
	
	Now, when $k=2$, since $R_{2}=R_{1}\hat{V}'$, SP algorithm will return the same index and operator matrix following a similar proof as the case $k=1$. Inductively, SP algorithm will return the same index and operator at every step for $R_{1}$ and $R_{2}$. Hence, the proof is finished.
\end{proof}
\subsection{Proof of Lemma 4.6}
\begin{proof}
	For SRSC and SRSC-equivalence: since $V_{\tau,2}=\mathscr{D}^{1/2}_{\tau}V_{2}=\mathscr{D}^{1/2}_{\tau}VV'=V_{\tau,1}V'$, we have $V_{\tau,2}(\mathcal{I},:)=V_{\tau,1}(\mathcal{I},:)V'$, which gives $V_{\tau,2}(\mathcal{I},:)V'_{\tau,2}(\mathcal{I},:)\equiv V_{\tau,1}(\mathcal{I},:)V'_{\tau,1}(\mathcal{I},:)$. By definition, we have $Z_{1}=Z_{2}$ surely. Since $\hat{V}_{\tau,2}=D^{1/2}_{\tau}\hat{V}_{2}=\hat{V}_{\tau,1}\hat{V}'$,  combine it with  the fact that $\mathcal{\hat{I}}_{1}\equiv \mathcal{\hat{I}}_{2}$ giving by  Lemma 4.5, we have $\hat{V}_{\tau,2}(\mathcal{\hat{I}}_{2},:)\hat{V}'_{\tau,2}(\mathcal{\hat{I}}_{2},:)\equiv \hat{V}_{\tau,1}(\mathcal{\hat{I}}_{1},:)\hat{V}'_{\tau,1}(\mathcal{\hat{I}}_{1},:)$. Therefore, we have $\hat{Z}_{2}\equiv\hat{Z}_{1}$ and $\hat{\Pi}_{2}\equiv\hat{\Pi}_{1}$.
	
	For CRSC and CRSC-equivalence: since $V_{2}(i,:)=e'_{i}VV'=V(i,:)V'$ for any $1\leq i\leq n$, we have $N_{V_{2}}(i,i)=\frac{1}{\|V_{2}(i,:)\|_{F}}=\frac{1}{\|V(i,:)V'\|_{F}}=\frac{1}{\|V(i,:)\|_{F}}$ where the last equality holds by lemma A.1 in \cite{yu2015a}. Hence, we have $N_{V}\equiv N_{V_{2}}$. Then we have $V_{*,2}(\mathcal{I},:)V=N_{V_{2}}(\mathcal{I},\mathcal{I})V_{2}(\mathcal{I},:)V=N_{V}(\mathcal{I},\mathcal{I})V(\mathcal{I},:)V'V=N_{V}(\mathcal{I},\mathcal{I})V(\mathcal{I},:)\equiv V_{*,1}(\mathcal{I},:)$. Meanwhile, $V_{2}(\mathcal{I},:)=(VV')(\mathcal{I},:)=V(\mathcal{I},:)V'$, which gives  $V_{*,2}(\mathcal{I},:)=N_{V_{2}}(\mathcal{I},\mathcal{I})V_{2}(\mathcal{I},:)=N_{V}(\mathcal{I},\mathcal{I})V(\mathcal{I},:)V'\equiv V_{*,1}(\mathcal{I},:)V'$. Then we have  $V_{*,2}(\mathcal{I},:)V'_{*,2}(\mathcal{I},:)=V_{*,1}(\mathcal{I},:)V'VV'_{*,1}(\mathcal{I},:)\equiv V_{*,1}(\mathcal{I},:)V'_{*,1}(\mathcal{I},:)$. Meanwhile, we also have $V_{*,2}=N_{V_{2}}V_{2}=N_{V}VV'=V_{*,1}V'$. Combine the above equalities, we have
	\begin{align*}
	&Y_{2}=V_{*,2}V'_{*,2}(\mathcal{I},:)(V_{*,2}(\mathcal{I},:)V'_{*,2}(\mathcal{I},:))^{-1}=V_{*,1}V'VV'_{*,1}(\mathcal{I},:)(V_{*,1}(\mathcal{I},:)V'_{*,1}(\mathcal{I},:))^{-1}\equiv Y_{1},\\
	&Y_{*,2}=V_{2}V'_{*,2}(\mathcal{I},:)(V_{*,2}(\mathcal{I},:)V'_{*,2}(\mathcal{I},:))^{-1}=VV'VV'_{*,1}(\mathcal{I},:)(V_{*,1}(\mathcal{I},:)V'_{*,1}()\mathcal{I},:)^{-1}\equiv Y_{*,1},\\
	&J_{*,2}=N_{V_{2}}(\mathcal{I},\mathcal{I})\mathscr{D}^{-1/2}_{\tau}(\mathcal{I},\mathcal{I})\equiv =N_{V}(\mathcal{I},\mathcal{I})\mathscr{D}^{-1/2}_{\tau}(\mathcal{I},\mathcal{I})J_{*,1}.
	\end{align*}
	Since $Z_{*,1}=Y_{*,1}J_{*,1}, Z_{*,2}=Y_{*,2}J_{*,2}$, we have $Z_{*,1}\equiv Z_{*,2}$. Meanwhile, note that $M_{1}=\Pi\mathscr{D}^{1/2}_{\tau}(\mathcal{I},\mathcal{I})V(\mathcal{I},:)\in\mathbb{R}^{n\times K}, M_{2}=\Pi\mathscr{D}^{1/2}_{\tau}(\mathcal{I},\mathcal{I})V_{2}(\mathcal{I},:)\in\mathbb{R}^{n\times n}$ gives $M_{1}\neq M_{2}$, but we still have $N_{M_{1}}\equiv N_{M_{2}}$ based on the fact that $N_{M_{2}}(i,i)=\frac{1}{\|M_{2}(i,:)\|_{F}}=\frac{1}{\|e'_{i}M_{2}\|_{F}}=\frac{1}{\|e'_{i}\Pi\mathscr{D}^{1/2}_{\tau}(\mathcal{I},\mathcal{I})V_{2}(\mathcal{I},:)\|_{F}}=\frac{1}{\|e'_{i}\Pi\mathscr{D}^{1/2}_{\tau}(\mathcal{I},\mathcal{I})V(\mathcal{I},:)V'\|_{F}}\equiv N_{M_{1}}(i,i)$ for $1\leq i\leq n$.

	Similarly, we have $N_{\hat{V}}\equiv N_{\hat{V}_{2}}$, where $N_{\hat{V}}$ is the diagonal matrix such that $\hat{V}_{*,1}=N_{\hat{V}}\hat{V}$. Lemma 4.5 guarantees $\mathcal{\hat{I}}_{*,1}\equiv \mathcal{\hat{I}}_{*,2}$. Then, follow a similar analysis as that of the ideal case, for the empirical case, we have $ \hat{V}_{*,1}(\hat{\mathcal{I}}_{*,1},:)\hat{V}'_{*,1}(\hat{\mathcal{I}}_{*,1},:)\equiv \hat{V}_{*,2}(\hat{\mathcal{I}}_{*,2},:)\hat{V}'_{*,2}(\hat{\mathcal{I}}_{*,2},:), \hat{Y}_{*,2}\equiv\hat{Y}_{*,1},\hat{J}_{*,2}\equiv\hat{J}_{*,1}, \hat{Z}_{*,2}\equiv\hat{Z}_{*,1},\hat{\Pi}_{*,2}\equiv\hat{\Pi}_{*,1}$.
\end{proof}
\section{Theoretical properties for SRSC and CRSC}\label{AppendixCommon}
Lemma \ref{P1} provides a further study on the Ideal Cone given in Lemma 3.3, it shows that $V_{*,1}(i,:)$ and $V_{*,2}(i,:)$ for CRSC can be written as a scaled convex combination of the $K$ rows of $V_{*,1}(\mathcal{I},:)$ and $V_{*,2}(\mathcal{I},:)$, respectively. Lemma \ref{P1} is consistent with Lemma A.1. in \cite{MaoSVM}. Meanwhile, Lemma \ref{P1} is one of the reasons that the SVM-cone algorithm (i.e, Algorithm \ref{alg:SVMcone}) can return the index set $\mathcal{I}$, for detail, refer to section \ref{OneClassSVMandSVMcone}.
\begin{lem}\label{P1}
	Under $MMSB(n,P,\Pi)$, for $1\leq i\leq n$, $V_{*,1}(i,:)$ can be written as $V_{*,1}(i,:)=r_{1}(i)\Phi_{1}(i,:)V_{*,1}(\mathcal{I},:)$, where $r_{1}(i)\geq 1$. Meanwhile, $r_{1}(i)=1$ and $\Phi_{1}(i,:)=e'_{k}$ if $i$ is a pure node such that $\Pi(i,k)=1$; $r_{1}(i)>1$ and $\Phi_{1}(i,:)\neq e'_{k}$ if $\Pi(i,k)<1$ for $1\leq k\leq K$. Similarly, $V_{*,2}(i,:)$ can be written as $V_{*,2}(i,:)=r_{2}(i)\Phi_{2}(i,:)V_{*,2}(\mathcal{I},:)$, where $r_{2}(i)\geq 1$. Meanwhile, $r_{2}(i)=1$ and $\Phi_{2}(i,:)=e'_{k}$ if $\Pi(i,k)=1$; $r_{2}(i)>1$ and $\Phi_{2}(i,:)\neq e'_{k}$ if $\Pi(i,k)<1$ for $1\leq k\leq K$.
\end{lem}
Lemma \ref{P2} is powerful to bound the behaviors of $\|V\|_{2\rightarrow\infty}$ and $\|V_{2}\|_{2\rightarrow\infty}$, and the result in Lemma \ref{P2} is called as the delocalization of population eigenvectors in Lemma 3.2 \cite{mao2020estimating}.
\begin{lem}\label{P2}
	Under $MMSB(n,P,\Pi)$, we have
	\begin{align*}
	\sqrt{\frac{\tau+\delta_{\mathrm{min}}}{\tau+\delta_{\mathrm{max}}}}\frac{1}{\sqrt{K\lambda_{1}(\Pi'\Pi)}}\leq \|V(i,:)\|_{F}\leq\sqrt{\frac{\tau+\delta_{\mathrm{max}}}{\tau+\delta_{\mathrm{min}}}}\frac{1}{\sqrt{\lambda_{K}(\Pi'\Pi)}},\qquad 1\leq i\leq n.
	\end{align*}
\end{lem}
Note that since $V_{2}(i,:)=e'_{i}VV'=V(i,:)V'$, by Lemma A.1 \cite{yu2015a}, we have $\|V_{2}(i,:)\|_{F}=\|V(i,:)V'\|_{F}=\|V(i,:)\|_{F}$, therefore results in Lemma \ref{P2} also holds for $V_{2}(i,:)$.
\begin{lem}\label{P3}
	Under $MMSB(n,P,\Pi)$, we have
	\begin{align*}
	&\lambda_{1}(V_{\tau,1}(\mathcal{I},:)V'_{\tau,1}(\mathcal{I},:))\leq\frac{\tau+\delta_{\mathrm{max}}}{\lambda_{K}(\Pi'\Pi)},~~~\lambda_{K}(V_{\tau,1}(\mathcal{I},:)V'_{\tau,1}(\mathcal{I},:))\geq\frac{\tau+\delta_{\mathrm{min}}}{\lambda_{1}(\Pi'\Pi)},\\
	&\lambda_{1}(V_{*,1}(\mathcal{I},:)V'_{*,1}(\mathcal{I},:))\leq K\frac{\tau+\delta_{\mathrm{max}}}{\tau+\delta_{\mathrm{min}}}\kappa(\Pi'\Pi),~~~\lambda_{K}(V_{*,1}(\mathcal{I},:)V'_{*,1}(\mathcal{I},:))\geq (\frac{\tau+\delta_{\mathrm{max}}}{\tau+\delta_{\mathrm{min}}})^{-1}\kappa^{-1}(\Pi'\Pi).
	\end{align*}
\end{lem}
Lemma \ref{P3} will be frequently used in our proofs since we always need to obtain the bound of $\lambda_{K}(V_{*,1}(\mathcal{I},:)V'_{*,1}(\mathcal{I},:))$ for further study.
\begin{lem}\label{P4}
	Under $MMSB(n, P,\Pi)$, we have
	\begin{align*}
	|\lambda_{K}|\geq\frac{\rho|\lambda_{K}(\tilde{P})|\lambda_{K}(\Pi'\Pi)}{\tau+\delta_{\mathrm{max}}}\mathrm{~and~} \lambda_{1}\leq\frac{\delta_{\mathrm{max}}}{\tau+\delta_{\mathrm{max}}}.
	\end{align*}
\end{lem}
\subsection{Proof of Lemma \ref{P1}}
\begin{proof}
	Since $V_{*,1}=YV_{*,1}(\mathcal{I},:)$, for $1\leq i\leq n$, we have
	\begin{align*}
	V_{*,1}(i,:)=Y(i,:)V_{*,1}(\mathcal{I},:)=Y(i,:)\mathbf{1}\frac{Y(i,:)}{Y(i,:)\mathbf{1}}V_{*,1}(\mathcal{I},:)=r_{1}(i)\Phi_{1}(i,:)V_{*,1}(\mathcal{I},:),
	\end{align*}
	where we set $r_{1}(i)=Y(i,:)\textbf{1}$, $\Phi_{1}(i,:)=\frac{Y(i,:)}{Y(i,:)\mathbf{1}}$, and $\mathbf{1}$ is a $K\times 1$  vector with all entries being ones.
	
	By the proof of Lemma 3.3, we know that $Y(i,:)=\frac{\Pi(i,:)}{\|M_{1}(i,:)\|_{F}}\mathscr{D}^{1/2}_{\tau}(\mathcal{I},\mathcal{I})N^{-1}(\mathcal{I},\mathcal{I})$, where $M_{1}=\Pi\mathscr{D}^{1/2}_{\tau}(\mathcal{I},\mathcal{I})V(\mathcal{I},:)$.  For convenience, set $T=\mathscr{D}^{1/2}_{\tau}(\mathcal{I},\mathcal{I}), Q=N^{-1}(\mathcal{I},\mathcal{I})$, and $R=V(\mathcal{I},:)$ (note that such setting of $T,Q, R$ is only for notation convenience in the proof of Lemma \ref{P1}).

	On the one hand, if node $i$ is pure such that $\Pi(i,k)=1$ for certain $k$ among $\{1,2,\ldots,K\}$ (i.e., $\Pi(i,:)=e_{k}$ if $\Pi(i,k)=1$), we have $M_{1}(i,:)=\Pi(i,:)\mathscr{D}^{1/2}_{\tau}(\mathcal{I},\mathcal{I})V(\mathcal{I},:)=T(k,k)R(k,:)$, and $\Pi(i,:)TQ=T(k,k)Q(k,:)$, which give that $Y(i,:)=\frac{T(k,k)Q(k,:)}{\|T(k,k)R(k,:)\|_{F}}=\frac{Q(k,:)}{\|R(k,:)\|_{F}}$. Recall that the $k$-th diagonal entry of $N^{-1}(\mathcal{I},\mathcal{I})$ is $\|[V(\mathcal{I},:)](k,:)\|_{F}$, i.e., $Q(k,:)\mathbf{1}=\|R(k,:)\|_{F}$, which gives that $r_{1}(i)=Y(i,:)\mathbf{1}=1$ and $\Phi_{1}(i,:)=e'_{k}$ if $\Pi(i,k)=1$.

	On the other hand, if $i$ it not a pure node, since $\|M_{1}(i,:)\|_{F}=\|\Pi(i,:)\mathscr{D}^{1/2}_{\tau}(\mathcal{I},\mathcal{I})V(\mathcal{I},:)\|_{F}=\|\sum_{k=1}^{K}\Pi(i,k)T(k,k)R(k,:)\|_{F}< \sum_{k=1}^{K}\Pi(i,k)T(k,k)\|R(k,:)\|_{F}=\sum_{k=1}^{K}\Pi(i,k)T(k,k)Q(k,k)$, combine it with $\Pi(i,:)TQ\mathbf{1}=\sum_{k=1}^{K}\Pi(i,k)T(k,k)Q(k,k)$, so $r_{1}(i)=Y(i,:)\mathbf{1}=\frac{\Pi(i,:)TQ\mathbf{1}}{\|M_{1}(i,:)\|_{F}}> 1$. Follow the above proof, we can obtain the results for $V_{*,2}$, here, we omit the detail.
\end{proof}
\subsection{Proof of Lemma \ref{P2}}
\begin{proof}
	Since $I=V'V=V'(\mathcal{I},:)\mathscr{D}^{1/2}_{\tau}(\mathcal{I},\mathcal{I})\Pi'\mathscr{D}^{-1}_{\tau}\Pi\mathscr{D}^{1/2}_{\tau}(\mathcal{I},\mathcal{I})V(\mathcal{I},:)$, we have
	\begin{align}\label{PiDinvPi}
	((\mathscr{D}^{1/2}_{\tau}(\mathcal{I},\mathcal{I})V(\mathcal{I},:))((\mathscr{D}^{1/2}_{\tau}(\mathcal{I},\mathcal{I})V(\mathcal{I},:))')^{-1}=\Pi'\mathscr{D}^{-1}_{\tau}\Pi,
	\end{align}
	which gives
	\begin{align*}
	\mathrm{max}_{k}\|e'_{k}(\mathscr{D}^{1/2}_{\tau}(\mathcal{I},\mathcal{I})V(\mathcal{I},:))\|_{F}^{2}&=\mathrm{max}_{k}e'_{k}(\mathscr{D}^{1/2}_{\tau}(\mathcal{I},\mathcal{I})V(\mathcal{I},:))(\mathscr{D}^{1/2}_{\tau}(\mathcal{I},\mathcal{I})V(\mathcal{I},:))'e_{k}\\
	&\leq \mathrm{max}_{\|x\|=1}x'(\mathscr{D}^{1/2}_{\tau}(\mathcal{I},\mathcal{I})V(\mathcal{I},:))(\mathscr{D}^{1/2}_{\tau}(\mathcal{I},\mathcal{I})V(\mathcal{I},:))'x\\
	&=\lambda_{1}((\mathscr{D}^{1/2}_{\tau}(\mathcal{I},\mathcal{I})V(\mathcal{I},:))(\mathscr{D}^{1/2}_{\tau}(\mathcal{I},\mathcal{I})V(\mathcal{I},:))')\\
	&=\frac{1}{\lambda_{K}(\Pi'\mathscr{D}^{-1}_{\tau}\Pi)}=
	\frac{1}{\lambda_{K}(\mathscr{D}^{-1}_{\tau}\Pi'\Pi)}\leq\frac{1}{\lambda_{K}(\mathscr{D}^{-1}_{\tau})\lambda_{K}(\Pi'\Pi)}\\
	&=\frac{\tau+\delta_{\mathrm{max}}}{\lambda_{K}(\Pi'\Pi)},
	\end{align*}
	where $x$ is a $K\times 1$ vector whose $l_{2}$ norm is 1. Meanwhile, we also have
	\begin{align*}
	\mathrm{min}_{k}\|e'_{k}(\mathscr{D}^{1/2}_{\tau}(\mathcal{I},\mathcal{I})V(\mathcal{I},:))\|_{F}^{2}&=\mathrm{min}_{k}e'_{k}(\mathscr{D}^{1/2}_{\tau}(\mathcal{I},\mathcal{I})V(\mathcal{I},:))(\mathscr{D}^{1/2}_{\tau}(\mathcal{I},\mathcal{I})V(\mathcal{I},:))'e_{k}\\
	&\geq \mathrm{min}_{\|x\|=1}x'(\mathscr{D}^{1/2}_{\tau}(\mathcal{I},\mathcal{I})V(\mathcal{I},:))(\mathscr{D}^{1/2}_{\tau}(\mathcal{I},\mathcal{I})V(\mathcal{I},:))'x\\
	&=\lambda_{K}((\mathscr{D}^{1/2}_{\tau}(\mathcal{I},\mathcal{I})V(\mathcal{I},:))(\mathscr{D}^{1/2}_{\tau}(\mathcal{I},\mathcal{I})V(\mathcal{I},:))')\\
	&=\frac{1}{\lambda_{1}(\Pi'\mathscr{D}^{-1}_{\tau}\Pi)}=\frac{1}{\lambda_{1}(\mathscr{D}^{-1}_{\tau}\Pi'\Pi)}\geq\frac{1}{\lambda_{1}(\mathscr{D}^{-1}_{\tau})\lambda_{1}(\Pi'\Pi)}\\
	&=\frac{\tau+\delta_{\mathrm{min}}}{\lambda_{1}(\Pi'\Pi)}.
	\end{align*}
	By the proof of Lemma 3.1, we have $V(i,:)=\mathscr{D}^{-1/2}_{\tau}(i,i)\Pi(i,:)\mathscr{D}^{1/2}_{\tau}(\mathcal{I},\mathcal{I})V(\mathcal{I},:)=\theta(i)\Pi(i,:)\mathscr{D}^{1/2}_{\tau}(\mathcal{I},\mathcal{I})V(\mathcal{I},:)$ for $1\leq i\leq n$, which gives that
	\begin{align*}
	\|V(i,:)\|_{F}&=\|\theta(i)\Pi(i,:)\mathscr{D}^{1/2}_{\tau}(\mathcal{I},\mathcal{I})V(\mathcal{I},:)\|_{F}\\
	&=\theta(i)\|\Pi(i,:)\mathscr{D}^{1/2}_{\tau}(\mathcal{I},\mathcal{I})V(\mathcal{I},:)\|_{F}\\
	&\leq \theta(i) \mathrm{max}_{k}\Pi(i,k)\mathrm{max}_{i}\|e'_{i}(\mathscr{D}^{1/2}_{\tau}(\mathcal{I},\mathcal{I})V(\mathcal{I},:))\|_{F}\\
	&\leq \theta(i)\mathrm{max}_{i}\|e'_{i}(\mathscr{D}^{1/2}_{\tau}(\mathcal{I},\mathcal{I})V(\mathcal{I},:))\|_{F}\\
	&\leq\frac{\theta_{\mathrm{max}}\sqrt{\tau+\delta_{\mathrm{max}}}}{\sqrt{\lambda_{K}(\Pi'\Pi)}}=\sqrt{\frac{\tau+\delta_{\mathrm{max}}}{\tau+\delta_{\mathrm{min}}}}\frac{1}{\sqrt{\lambda_{K}(\Pi'\Pi)}}.
	\end{align*}
	Similarly, we have
	\begin{align*}
	\|V(i,:)\|_{F}&=\|\theta(i)\Pi(i,:)\mathscr{D}^{1/2}_{\tau}(\mathcal{I},\mathcal{I})V(\mathcal{I},:)\|_{F}\\
	&=\theta(i)\|\Pi(i,:)\mathscr{D}^{1/2}_{\tau}(\mathcal{I},\mathcal{I})V(\mathcal{I},:)\|_{F}\\
	&\geq \theta(i) \mathrm{min}_{i}\|\Pi(i,:)\|_{F}\mathrm{min}_{i}\|e'_{i}(\mathscr{D}^{1/2}_{\tau}(\mathcal{I},\mathcal{I})V(\mathcal{I},:))\|_{F}\\
	&\geq \theta(i)\mathrm{min}_{i}\|e'_{i}(\mathscr{D}^{1/2}_{\tau}(\mathcal{I},\mathcal{I})V(\mathcal{I},:))\|_{F}/\sqrt{K}\\
	&\geq\frac{\theta_{\mathrm{min}}\sqrt{\tau+\delta_{\mathrm{min}}}}{\sqrt{K\lambda_{1}(\Pi'\Pi)}}=\sqrt{\frac{\tau+\delta_{\mathrm{min}}}{\tau+\delta_{\mathrm{max}}}}\frac{1}{\sqrt{K\lambda_{1}(\Pi'\Pi)}},
	\end{align*}
	where we use the fact that $\mathrm{min}_{i}\|\Pi(i,:)\|_{F}\geq \frac{1}{\sqrt{K}}$ since $\sum_{k=1}^{K}\Pi(i,k)=1$ and all entries of $\Pi$ are nonnegative. Meanwhile, we also have, for $1\leq i\leq n$,
	\begin{align*}
	\sqrt{\frac{\lambda_{K}(\Pi'\Pi)}{\tau+\delta_{\mathrm{max}}}}\leq\frac{\theta(i)}{\|V(i,:)\|_{F}}\leq \sqrt{\frac{K\lambda_{1}(\Pi'\Pi)}{\tau+\delta_{\mathrm{min}}}}.
	\end{align*}
\end{proof}
\subsection{Proof of Lemma \ref{P3}}
\begin{proof}
	In this proof, we will frequently use the fact that for any two matrices $X_{1}$ and $X_{2}$, the nonzero eigenvalues of $X_{1}X_{2}$ are the same
	as the nonzero eigenvalues of $X_{2}X_{1}$.
	
	Eq (\ref{PiDinvPi}) gives that
	\begin{align*}
	\lambda_{1}(V_{\tau,1}(\mathcal{I},:)V'_{\tau,1}(\mathcal{I},:))&=\lambda_{1}(\mathscr{D}^{1/2}_{\tau}(\mathcal{I},\mathcal{I})V(\mathcal{I},:)V'(\mathcal{I},:)\mathscr{D}^{1/2}_{\tau}(\mathcal{I},\mathcal{I}))\\
	&=\frac{1}{\lambda_{K}(\Pi'\mathscr{D}^{-1}_{\tau}\Pi)}\leq\frac{\tau+\delta_{\mathrm{max}}}{\lambda_{K}(\Pi'\Pi)},
	\end{align*}
	and
	\begin{align*}
	\lambda_{K}(V_{\tau,1}(\mathcal{I},:)V'_{\tau,1}(\mathcal{I},:))&=\lambda_{1}(\mathscr{D}^{1/2}_{\tau}(\mathcal{I},\mathcal{I})V(\mathcal{I},:)V'(\mathcal{I},:)\mathscr{D}^{1/2}_{\tau}(\mathcal{I},\mathcal{I}))\\
	&=\frac{1}{\lambda_{1}(\Pi'\mathscr{D}^{-1}_{\tau}\Pi)}\geq \frac{\tau+\delta_{\mathrm{min}}}{\lambda_{1}(\Pi'\Pi)}.
	\end{align*}
	By the proof of Lemma 3.4, we know that $V(\mathcal{I},:)V'(\mathcal{I},:)=\mathscr{D}^{-1/2}_{\tau}(\mathcal{I},\mathcal{I})(\Pi'\mathscr{D}^{-1}_{\tau}\Pi)^{-1}\mathscr{D}^{-1/2}_{\tau}(\mathcal{I},\mathcal{I})$,which gives
	\begin{align*}
	\lambda_{1}(V_{*,1}(\mathcal{I},:)V'_{*,1}(\mathcal{I},:))&=\lambda_{1}(N(\mathcal{I},\mathcal{I})V(\mathcal{I},:)V'(\mathcal{I},:)N(\mathcal{I},\mathcal{I}))\\
	&=\lambda_{1}(N(\mathcal{I},\mathcal{I})\mathscr{D}^{-1/2}_{\tau}(\mathcal{I},\mathcal{I})(\Pi'\mathscr{D}^{-1}_{\tau}\Pi)^{-1}\mathscr{D}^{-1/2}_{\tau}(\mathcal{I},\mathcal{I})N(\mathcal{I},\mathcal{I}))\\
	&=\lambda_{1}(N^{2}(\mathcal{I},\mathcal{I})\mathscr{D}^{-1}_{\tau}(\mathcal{I},\mathcal{I})(\Pi'\mathscr{D}^{-1}_{\tau}\Pi)^{-1})\\
	&\leq\lambda^{2}_{1}(N(\mathcal{I},\mathcal{I})\mathscr{D}^{-1/2}_{\tau}(\mathcal{I},\mathcal{I}))\lambda_{1}((\Pi'\mathscr{D}^{-1}_{\tau}\Pi)^{-1})\\
	&=\lambda^{2}_{1}(N(\mathcal{I},\mathcal{I})\mathscr{D}^{-1/2}_{\tau}(\mathcal{I},\mathcal{I}))/\lambda_{K}(\Pi'\mathscr{D}^{-1}_{\tau}\Pi)\\
	&\leq(\mathrm{max}_{i\in \mathcal{I}}\theta(i)/\|V(i,:)\|_{F})^{2}/\lambda_{K}(\Pi'\mathscr{D}^{-1}_{\tau}\Pi)\\
	&\leq\frac{K\lambda_{1}(\Pi'\Pi)}{\tau+\delta_{\mathrm{min}}}\frac{\tau+\delta_{\mathrm{max}}}{\lambda_{K}(\Pi'\Pi)}\\
	&=K\frac{\tau+\delta_{\mathrm{max}}}{\tau+\delta_{\mathrm{min}}}\kappa(\Pi'\Pi),
	\end{align*}
	where we use the fact that $N(i,i)=\frac{1}{\|V(i,:)\|_{F}}$. Similarly, we have
	\begin{align*}
	\lambda_{K}(V_{*,1}(\mathcal{I},:)V'_{*,1}(\mathcal{I},:))&=\lambda_{K}(N(\mathcal{I},\mathcal{I})V(\mathcal{I},:)V'(\mathcal{I},:)N(\mathcal{I},\mathcal{I}))\\
	&=\lambda_{K}(N(\mathcal{I},\mathcal{I})\mathscr{D}^{-1/2}_{\tau}(\mathcal{I},\mathcal{I})(\Pi'\mathscr{D}^{-1}_{\tau}\Pi)^{-1}\mathscr{D}^{-1/2}_{\tau}(\mathcal{I},\mathcal{I})N(\mathcal{I},\mathcal{I}))\\
	&=\lambda_{K}(N^{2}(\mathcal{I},\mathcal{I})\mathscr{D}^{-1}_{\tau}(\mathcal{I},\mathcal{I})(\Pi'\mathscr{D}^{-1}_{\tau}\Pi)^{-1})\\
	&\geq\lambda^{2}_{K}(N(\mathcal{I},\mathcal{I})\mathscr{D}^{-1}_{\tau}(\mathcal{I},\mathcal{I}))\lambda_{K}((\Pi'\mathscr{D}^{-1}_{\tau}\Pi)^{-1})\\
	&=\lambda^{2}_{K}(N(\mathcal{I},\mathcal{I})\mathscr{D}^{-1}_{\tau}(\mathcal{I},\mathcal{I}))/\lambda_{1}(\Pi'\mathscr{D}^{-1}_{\tau}\Pi)\\
	&\geq(\mathrm{min}_{i\in \mathcal{I}}\theta(i)/\|V(i,:)\|_{F})^{2}/\lambda_{1}(\Pi'\mathscr{D}^{-1}_{\tau}\Pi)\\
	&\geq(\frac{\tau+\delta_{\mathrm{max}}}{\tau+\delta_{\mathrm{min}}})^{-1}\kappa^{-1}(\Pi'\Pi).
	\end{align*}
\end{proof}
\subsection{Proof of Lemma \ref{P4}}
\begin{proof}
	Set $H=P\Pi'\mathscr{D}^{-1}_{\tau}\Pi P\in\mathbb{R}^{K\times K}$, by basic algebra, we have that $H$ is full rank and positive definite, which gives that
	\begin{align*}
	|\lambda_{K}|&=|\lambda_{K}(\mathscr{L}_{\tau})|=|\lambda_{K}(\mathscr{D}^{-1/2}_{\tau}\Pi P\Pi'\mathscr{D}^{-1/2}_{\tau})|=\sqrt{\lambda_{K}(\mathscr{D}^{-1/2}_{\tau}\Pi P\Pi'\mathscr{D}^{-1}_{\tau}\Pi P\Pi'\mathscr{D}^{-1/2}_{\tau})}\\
	&=\sqrt{\lambda_{K}(\mathscr{D}^{-1/2}_{\tau}\Pi H\Pi'\mathscr{D}^{-1/2}_{\tau})}=\sqrt{\lambda_{K}(\mathscr{D}^{-1/2}_{\tau}\Pi H^{1/2}H^{1/2}\Pi'\mathscr{D}^{-1/2}_{\tau})}=\sqrt{\lambda_{K}(H^{1/2}\Pi'\mathscr{D}^{-1}_{\tau}\Pi H^{1/2})}\\
	&=\sqrt{\lambda_{K}(H\Pi'\mathscr{D}^{-1}_{\tau}\Pi)}\geq \sqrt{\lambda_{K}(H)\lambda_{K}(\Pi'\mathscr{D}^{-1}_{\tau}\Pi)}=\sqrt{\lambda_{K}(P\Pi'\mathscr{D}^{-1}_{\tau}\Pi P)\lambda_{K}(\Pi'\mathscr{D}^{-1}_{\tau}\Pi)}\\
	&=\sqrt{\lambda_{K}(P^{2}\Pi'\mathscr{D}^{-1}_{\tau}\Pi)\lambda_{K}(\Pi'\mathscr{D}^{-1}_{\tau}\Pi)}\geq \sqrt{\lambda_{K}(P^{2})\lambda^{2}_{K}(\Pi'\mathscr{D}^{-1}_{\tau}\Pi)}\geq \sqrt{\lambda^{2}_{K}(P)\lambda^{2}_{K}(\Pi'\mathscr{D}^{-1}_{\tau}\Pi)}\\
	&=|\lambda_{K}(P)|\lambda_{K}(\Pi'\mathscr{D}^{-1}_{\tau}\Pi)\geq\frac{\rho|\lambda_{K}(\tilde{P})|\lambda_{K}(\Pi'\Pi)}{\tau+\delta_{\mathrm{max}}},
	\end{align*}
	where we have used the fact that for any matrix $T\in\mathbb{R}^{n\times K}$ with rank $K<n$, $TT'$ and $T'T$ have the same leading $K$ eigenvalues. For $\lambda_{1}$, we have
	\begin{align*}
	\lambda_{1}&=\|\mathscr{L}_{\tau}\|=\|\mathscr{D}_{\tau}^{-1/2}\Omega\mathscr{D}_{\tau}^{-1/2}\|=\|\mathscr{D}^{-1/2}_{\tau}\mathscr{D}^{1/2}\mathscr{D}^{-1/2}\Omega\mathscr{D}^{-1/2}\mathscr{D}^{1/2}\mathscr{D}^{-1/2}_{\tau}\|\\
	&\leq \|\mathscr{D}^{-1/2}_{\tau}\mathscr{D}^{1/2}\|^{2}\|\mathscr{D}^{-1/2}\Omega\mathscr{D}^{-1/2}\|=\|\mathscr{D}^{-1}_{\tau}\mathscr{D}\|=\mathrm{max}_{1\leq i\leq n}\frac{\mathscr{D}(i,i)}{\tau+\mathscr{D}(i,i)}\leq \frac{\delta_{\mathrm{max}}}{\tau+\delta_{\mathrm{max}}}\leq 1.
	\end{align*}
\end{proof}
\section{Basic properties of $\mathscr{L}_{\tau}$}
\subsection{Proof of Lemma 5.2}
\begin{proof}
	We apply Theorem 1.4 (Bernstein inequality) in \cite{tropp2012user} to bound $\|L_{\tau}-\mathscr{L}_{\tau}\|$, and this theorem is written as below
	\begin{thm}\label{Bern}
		Consider a finite sequence $\{X_{k}\}$ of independent, random, self-adjoint matrices with dimension $d$. Assume that each random matrix satisfies
		\begin{align*}
		\mathbb{E}[X_{k}]=0, \mathrm{and~}\lambda_{\mathrm{max}}(X_{k})\leq R~\mathrm{almost~surely}.
		\end{align*}
		Then, for all $t\geq 0$,
		\begin{align*}
		\mathbb{P}(\lambda_{\mathrm{max}}(\sum_{k}X_{k})\geq t)\leq d\cdot \mathrm{exp}(\frac{-t^{2}/2}{\sigma^{2}+Rt/3}),
		\end{align*}
		where $\sigma^{2}:=\|\sum_{k}\mathbb{E}[X^{2}_{k}]\|$.
	\end{thm}
	
	Now, we start the proof. Since
	\begin{align*}
	\|L_{\tau}-\mathscr{L}_{\tau}\|&\leq\|\mathscr{D}^{-1/2}_{\tau}A\mathscr{D}^{-1/2}_{\tau}-\mathscr{D}^{-1/2}_{\tau}\Omega\mathscr{D}^{-1/2}_{\tau}\|+\|D_{\tau}^{-1/2}AD^{-1/2}_{\tau}-\mathscr{D}^{-1/2}_{\tau}A\mathscr{D}^{-1/2}_{\tau}\|,
	\end{align*}
	we bound the two terms of the right hand side separately.
	
	For the first term, we apply Theorem \ref{Bern}.
	Let $e_{i}$ be an $n\times 1$ vector, where $e_{i}(i)=1$ and 0 elsewhere, for nodes $1\leq i\leq n$. For convenience, set $W=\mathscr{D}^{-1/2}_{\tau}A\mathscr{D}^{-1/2}_{\tau}-\mathscr{D}^{-1/2}_{\tau}\Omega\mathscr{D}^{-1/2}_{\tau}$, and we have $W(i,j)=\frac{A(i,j)-\Omega(i,j)}{\sqrt{\mathscr{D}_{\tau}(i,i)\mathscr{D}_{\tau}(j,j)}}$ for $1\leq i,j\leq n$. Then we can write $W$ as $W=\sum_{i=1}^{n}\sum_{j=1}^{n}W(i,j)e_{i}e'_{j}$. Set $W^{(i,j)}$ as the $n\times n$ matrix such that $W^{(i,j)}=W(i,j)(e_{i}e'_{j}+e_{j}e_{i}')$, which gives that $W=\sum_{1\leq i< j\leq n}W^{(i,j)}$. Then we have $\mathbb{E}[W^{(i,j)}]=0$ and
	\begin{align*}
	\|W^{(i,j)}\|&=\|W(i,j)(e_{i}e'_{j}+e_{j}e_{i})\|=|W(i,j)|\|(e_{i}e'_{j}+e_{j}e_{i}')\|=|W(i,j)|=|\frac{A(i,j)-\Omega(i,j)}{\sqrt{\mathscr{D}_{\tau}(i,i)\mathscr{D}_{\tau}(j,j)}}|\\
	&\leq\frac{1}{\sqrt{\mathscr{D}_{\tau}(i,i)\mathscr{D}_{\tau}(j,j)}}\leq \frac{1}{\delta_{\mathrm{min}}+\tau}.
	\end{align*}
	Next  we consider the variance parameter
	\begin{align*}
	\sigma^{2}:=\|\sum_{1\leq i<j\leq n}\mathbb{E}[(W^{(i,j)})^{2}]\|.
	\end{align*}
	We obtain the bound of $\mathbb{E}(W^{2}(i,j))$ as below
	\begin{align*}
	\mathbb{E}(W^{2}(i,j))&=\frac{\mathbb{E}((A(i,j)-\Omega(i,j))^{2})}{\mathscr{D}_{\tau}(i,i)\mathscr{D}_{\tau}(j,j)}=\frac{\mathbb{E}((A(i,j)-\mathbb{E}(A(i,j)))^{2})}{\mathscr{D}_{\tau}(i,i)\mathscr{D}_{\tau}(j,j)}=\frac{\mathrm{Var}(A(i,j))}{\mathscr{D}_{\tau}(i,i)\mathscr{D}_{\tau}(j,j)}= \frac{\Omega(i,j)(1-\Omega(i,j))}{\mathscr{D}_{\tau}(i,i)\mathscr{D}_{\tau}(j,j)}\\
	&\leq \frac{\Omega(i,j)}{(\delta_{\mathrm{min}}+\tau)^{2}}=\frac{\Pi(i,:)P\Pi'(j,:)}{(\delta_{\mathrm{min}}+\tau)^{2}}\leq \frac{\mathrm{max}_{1\leq i\leq j\leq n}\{\Pi(i,:)P\Pi'(j,:)\}}{(\delta_{\mathrm{min}}+\tau)^{2}}\leq \frac{C\rho}{(\delta_{\mathrm{min}}+\tau)^{2}},
	\end{align*}
	where we have used the fact that $\Pi(i,:)P\Pi'(j,:)=\rho \Pi(i,:)\tilde{P}\Pi'(j,:)\leq \rho \|\tilde{P}\|_{\mathrm{max}}\|\Pi(i,:)\|_{1}\|\Pi(j,:)\|_{1}=C\rho$.
	Next we bound $\sigma^{2}$ as below
	\begin{align*}
	\sigma^{2}&=\|\sum_{1\leq i<j\leq n}\mathbb{E}(W^{2}(i,j))(e_{i}e_{j}'+e_{j}e_{i}')(e_{i}e_{j}'+e_{j}e_{i}')\|=\|\sum_{1\leq i<j\leq n}\mathbb{E}[W^{2}(i,j)(e_{i}e'_{i}+e_{j}e_{j}')]\|\\
	&\leq\underset{1\leq i\leq n}{\mathrm{max}}|\sum_{j=1}^{n}\mathbb{E}(W^{2}(i,j))|\leq \underset{1\leq i\leq n}{\mathrm{max}}\sum_{j=1}^{n}C\rho/(\delta_{\mathrm{min}}+\tau)^{2}=\frac{C\rho n}{(\tau+\delta_{\mathrm{min}})^{2}}.
	\end{align*}
	Thus, we have
	\begin{align*}
	\sigma^{2}\leq C\rho n/(\tau+\delta_{\mathrm{min}})^{2}.
	\end{align*}
	Set $t=\frac{\sqrt{\frac{32}{3}C\rho n\mathrm{log}(n^{\alpha}K^{-\beta})}}{\tau+\delta_{\mathrm{min}}}$, combine Theorem \ref{Bern} with $\sigma^{2}\leq C\rho n/(\tau+\delta_{\mathrm{min}})^{2}, R=1/(\tau+\delta_{\mathrm{min}}), d=n$, we have
	\begin{align*}
	\mathbb{P}(\|W\|\geq t)&=\mathbb{P}(\|\sum_{1\leq i<j\leq n}W^{(i,j)}\|\geq t)\leq n\cdot \mathrm{exp}(\frac{-t^{2}/2}{\sigma^{2}+Rt/3})\leq n\cdot\mathrm{exp}(\frac{-\frac{16}{3}\mathrm{log}(n^{\alpha}K^{-\beta})}{1+\frac{1}{3}\sqrt{\frac{32\mathrm{log}(n^{\alpha}K^{-\beta})}{3C\rho n}}})\leq \frac{K^{4\beta}}{n^{4\alpha-1}},
	\end{align*}
	where we have  use Condition (A1) such that $1+\frac{1}{3}\sqrt{32\mathrm{log}(n^{\alpha}K^{-\beta})/(3C\rho n)}\leq \frac{4}{3}$ for sufficiently large $n$ in the  last inequality. $C$ is always a positive constant, we can set $t=\frac{\sqrt{C\rho n\mathrm{log}(n^{\alpha}K^{-\beta})}}{\tau+\delta_{\mathrm{min}}}$ for convenience.
	
	For the second term $\|D_{\tau}^{-1/2}AD^{-1/2}_{\tau}-\mathscr{D}^{-1/2}_{\tau}A\mathscr{D}^{-1/2}_{\tau}\|$. Since
	\begin{align*}
	\|L_{\tau}\|&=\|D^{-1/2}_{\tau}AD^{-1/2}_{\tau}\|=\|D^{-1/2}_{\tau}D^{1/2}D^{-1/2}AD^{-1/2}D^{1/2}D^{-1/2}_{\tau}\|\\
	&\leq\|D^{-1/2}_{\tau}D^{1/2}\|\|D^{-1/2}AD^{-1/2}\|\|D^{1/2}D^{-1/2}_{\tau}\|=\|D^{-1/2}_{\tau}D^{1/2}\|\|D^{1/2}D^{-1/2}_{\tau}\|\leq 1,
	\end{align*}
	we have
	\begin{align*}
	&\|D_{\tau}^{-1/2}AD^{-1/2}_{\tau}-\mathscr{D}^{-1/2}_{\tau}A\mathscr{D}^{-1/2}_{\tau}\|\\
	&=\|D_{\tau}^{-1/2}AD^{-1/2}_{\tau}-\mathscr{D}_{\tau}^{-1/2}D_{\tau}^{1/2}L_{\tau}D_{\tau}^{1/2}\mathscr{D}_{\tau}^{-1/2}\|\\
	&=\|(I-\mathscr{D}^{-1/2}_{\tau}D^{1/2}_{\tau})L_{\tau}D_{\tau}^{1/2}\mathscr{D}_{\tau}^{-1/2}+L_{\tau}(I-D_{\tau}^{1/2}\mathscr{D}_{\tau}^{-1/2})\|\\
	&\leq\|I-\mathscr{D}^{-1/2}_{\tau}D^{1/2}_{\tau}\|\|L_{\tau}\|\|D_{\tau}^{1/2}\mathscr{D}_{\tau}^{-1/2}\|+\|L_{\tau}\|\|I-D_{\tau}^{1/2}\mathscr{D}_{\tau}^{-1/2}\|\\
	&\leq\|I-\mathscr{D}^{-1/2}_{\tau}D^{1/2}_{\tau}\|\|D_{\tau}^{1/2}\mathscr{D}_{\tau}^{-1/2}\|+\|I-D_{\tau}^{1/2}\mathscr{D}_{\tau}^{-1/2}\|\\
	&\leq\|I-\mathscr{D}^{-1/2}_{\tau}D^{1/2}_{\tau}\|\|D_{\tau}^{1/2}\mathscr{D}_{\tau}^{-1/2}-I+I\|+\|I-D_{\tau}^{1/2}\mathscr{D}_{\tau}^{-1/2}\|\\
	&\leq\|I-\mathscr{D}^{-1/2}_{\tau}D^{1/2}_{\tau}\|(\|D_{\tau}^{1/2}\mathscr{D}_{\tau}^{-1/2}-I\|+\|I\|)+\|I-D_{\tau}^{1/2}\mathscr{D}_{\tau}^{-1/2}\|\\
	&=\|I-\mathscr{D}^{-1/2}_{\tau}D^{1/2}_{\tau}\|(\|D_{\tau}^{1/2}\mathscr{D}_{\tau}^{-1/2}-I\|+1)+\|I-D_{\tau}^{1/2}\mathscr{D}_{\tau}^{-1/2}\|\\
	&=2\|I-D_{\tau}^{1/2}\mathscr{D}_{\tau}^{-1/2}\|+\|I-D_{\tau}^{1/2}\mathscr{D}_{\tau}^{-1/2}\|^{2}.
	\end{align*}
	Next we bound $\|I-D_{\tau}^{1/2}\mathscr{D}_{\tau}^{-1/2}\|$.
	Apply the two sided concentration inequality (see for example \cite{chung2006complex}, chap. 2) for each $1\leq i\leq n$,
	\begin{align*}
	\mathbb{P}(|D(i,i)-\mathscr{D}(i,i)|\geq \varrho)&\leq \mathrm{exp}(-\frac{\varrho^{2}}{2\mathscr{D}(i,i)})+\mathrm{exp}(-\frac{\varrho^{2}}{2\mathscr{D}(i,i)+\frac{2}{3}\varrho}).
	\end{align*}
	Let $\varrho=t(\mathscr{D}(i,i)+\tau)$, we have
	\begin{align*}
	&\mathbb{P}(|D(i,i)-\mathscr{D}(i,i)|\geq t(\mathscr{D}(i,i)+\tau))\leq \mathrm{exp}(\frac{-t^{2}(\mathscr{D}(i,i)+\tau)^{2}}{2\mathscr{D}(i,i)})+\mathrm{exp}(\frac{-t^{2}(\mathscr{D}(i,i)+\tau)^{2}}{2\mathscr{D}(i,i)+\frac{2}{3}t(\mathscr{D}(i,i)+\tau)})\\
	&\leq 2\mathrm{exp}(-\frac{t^{2}(\mathscr{D}(i,i)+\tau)^{2}}{(2+\frac{2}{3}t)(\mathscr{D}(i,i)+\tau)})=2\mathrm{exp}(-\frac{t^{2}(\mathscr{D}(i,i)+\tau)}{2+\frac{2}{3}t})\leq 2\mathrm{exp}(-\frac{t^{2}(\delta_{\mathrm{min}}+\tau)}{2+\frac{2}{3}t})\\
	&=2\mathrm{exp}(-4\mathrm{log}(n^{\alpha}K^{-\beta})\frac{1}{\frac{8(\tau+\delta_{\mathrm{min}})}{C\rho n}+\frac{8}{3}\sqrt{\frac{\mathrm{log}(n^{\alpha}K^{-\beta})}{C\rho n}}})\leq 2\frac{K^{4\beta}}{n^{4\alpha-1}}
	\end{align*}
	where we add a constraint on $\tau+\delta_{\mathrm{min}}$ such  that $\tau+\delta_{\mathrm{min}}\leq C\rho n$ in the  last inequality (for sufficiently large $n$,  we have $\frac{8(\tau+\delta_{\mathrm{min}})}{C\rho n}+\frac{8}{3}\sqrt{\frac{\mathrm{log}(n^{\alpha}K^{-\beta})}{C\rho n}}\leq 1$). Then, we have
	\begin{align*}
	\mathbb{P}(\|I-D_{\tau}^{1/2}\mathscr{D}_{\tau}^{-1/2}\|\geq t)&\leq \mathbb{P}(\mathrm{max}_{1\leq i\leq n}|\frac{D(i,i)+\tau}{\mathscr{D}(i,i)+\tau}-1|\geq t)\\
	&\leq\mathbb{P}(\cup_{1\leq i\leq n}\{|(D(i,i)+\tau)-(\mathscr{D}(i,i)+\tau)|\geq t(\mathscr{D}(i,i)+\tau)\})\\
	&=\mathbb{P}(\cup_{1\leq i\leq n}\{|D(i,i)-\mathscr{D}(i,i)|\geq t(\mathscr{D}_{\tau}(i,i)+\tau)\})\\
	&\leq 2\frac{K^{4\beta}}{n^{4\alpha-1}}.
	\end{align*}
	Therefore, we have
	\begin{align*}
	&\|D_{\tau}^{-1/2}AD^{-1/2}_{\tau}-\mathscr{D}^{-1/2}_{\tau}A\mathscr{D}^{-1/2}_{\tau}\|\leq 2\|I-D_{\tau}^{1/2}\mathscr{D}_{\tau}^{-1/2}\|+\|I-D_{\tau}^{1/2}\mathscr{D}_{\tau}^{-1/2}\|^{2}\leq 2t+t^{2},
	\end{align*}
	with probability at least $1-o(\frac{K^{4\beta}}{n^{4\alpha-1}})$.
	
	Combining the two parts yields
	\begin{align*}
	&\|L_{\tau}-\mathscr{L}_{\tau}\|\leq t^{2}+3t=O(\frac{\rho n\mathrm{log}(n^{\alpha}K^{-\beta})}{(\tau+\delta_{\mathrm{min}})^{2}})+O(\frac{\sqrt{\rho n\mathrm{log}(n^{\alpha}K^{-\beta})}}{\tau+\delta_{\mathrm{min}}})\\
	&=\begin{cases}
	O(\frac{\sqrt{\rho n\mathrm{log}(n^{\alpha}K^{-\beta})}}{\tau+\delta_{\mathrm{min}}}), & \mbox{when } C\sqrt{\rho n\mathrm{log}(n^{\alpha}K^{-\beta})}\leq \tau+\delta_{\mathrm{min}}\leq C\rho n, \\
	O(\frac{\rho n\mathrm{log}(n^{\alpha}K^{-\beta})}{(\tau+\delta_{\mathrm{min}})^{2}}), & \mbox{when~}\tau+\delta_{\mathrm{min}}<C\sqrt{\rho n\mathrm{log}(n^{\alpha}K^{-\beta})},
	\end{cases}
	\end{align*}
	with probability at least $1-o(\frac{K^{4\beta}}{n^{4\alpha-1}})$.
\end{proof}
\begin{rem}
	Actually, since $\mathbb{E}[(L_{\tau}(i,j)-\mathscr{L}_{\tau}(i,j))^{2}]\leq \frac{\rho}{\tau}$ (see the proof of Lemma 5.4 for detail) and $|L_{\tau}(i,j)-\mathscr{L}_{\tau}(i,j)|\leq \frac{1}{\tau}$, by Lemma 1 \cite{chen2021asymmetry}, with high probability, we have
	\begin{align*}
	\|L_{\tau}-\mathscr{L}_{\tau}\|\leq C\frac{\sqrt{\rho n\mathrm{log}(n)}}{\tau}+C\frac{\mathrm{log}(n)}{\tau}.
	\end{align*}
	Under Condition (A1), this bound can be written as $\|L_{\tau}-\mathscr{L}_{\tau}\|=O(\frac{\sqrt{\rho n\mathrm{log}(n)}}{\tau})$.
	Though this bound is slightly larger, it is consistent with the bound in Lemma 5.2.
\end{rem}

\section{Proof of consistency for SRSC and CRSC}
\subsection{Proof of Lemma 5.4}
\begin{proof}
	To prove this lemma, we apply Theorem 4.2.1 \citep{chen2020spectral} and Lemma 5.1 \citep{lei2015consistency} where Lemma 5.1 \citep{lei2015consistency} is obtained based on the Davis-Kahan theorem \citep{yu2015a}. First, we use Theorem 4.2.1 \citep{chen2020spectral} to bound $\|\hat{V}\mathrm{sgn}(H)-V\|_{2\rightarrow\infty}$ where $\mathrm{sgn}(H)$ is defined below. Let $H=\hat{V}'V$, and $H=U_{H}\Sigma_{H}V'_{H}$ be the SVD decomposition of $H$ with $U_{H},V_{H}\in \mathbb{R}^{n\times K}$, where $U_{H}$ and $V_{H}$ represent respectively the left and right singular matrices of $H$. Define $\mathrm{sgn}(H)=U_{H}V'_{H}$. Since $\mathbb{E}(A(i,j)-\Omega(i,j))=0$, $\mathbb{E}[(L_{\tau}(i,j)-\mathscr{L}_{\tau}(i,j))^{2}]=\mathbb{E}[(\frac{A(i,j)}{\sqrt{(\tau+D(i,i))(\tau+D(j,j))}}-\frac{\Omega(i,j)}{\sqrt{(\tau+\mathscr{D}(i,i))(\tau+\mathscr{D}(j,j))}})^{2}]\leq \frac{\mathbb{E}[(A(i,j)-\Omega(i,j))^{2}]}{\mathrm{min}((\tau+1)^{2},(\tau+\delta_{\mathrm{min}})^{2})}=\frac{\mathrm{Var}(A(i,j))}{\tilde{\tau}^{2}}=\Omega(i,j)(1-\Omega(i,j))/\tilde{\tau}^{2}\leq\Omega(i,j)/\tilde{\tau}^{2}=\rho P(g_{i},g_{j})/\tilde{\tau}^{2}\leq \frac{\rho}{\tilde{\tau}^{2}}, |L_{\tau}(i,j)-\mathscr{L}_{\tau}(i,j)|\leq \mathrm{max}(\frac{1}{\tau+1},\frac{1}{\tau+\delta_{\mathrm{min}}})=\frac{1}{\tilde{\tau}}$ where we set $\tilde{\tau}=\mathrm{min}(\tau+1,\tau+\delta_{\mathrm{min}})$, and by Condition (A1) and Lemma \ref{P2} we have $c_{b}=\frac{1}{\tilde{\tau} \frac{\sqrt{\rho}}{\tilde{\tau}}\sqrt{n/(\mu \mathrm{log}(n))}}=O(1)$ where $\mu=\frac{n\|V\|^{2}_{2\rightarrow\infty}}{K}$, meanwhile, by Condition (A1) and Lemma \ref{P4}, we have $|\lambda_{K}|\geq C\frac{\sqrt{\rho n\mathrm{log}(n)}}{\tilde{\tau}}$ when $\frac{|\lambda_{K}(\tilde{P})|\tilde{\tau}}{K(\tau+\delta_{\mathrm{max}})}\geq\sqrt{\frac{\mathrm{log}(n)}{\rho n}}$, Theorem 4.2.1. \cite{chen2020spectral} gives that with high probability,
	\begin{align*}
	\|\hat{V}\mathrm{sgn}(H)-V\|_{2\rightarrow\infty}\leq \frac{\kappa(\mathscr{L}_{\tau})\sqrt{K\mu\rho}+\sqrt{K\rho\mathrm{log}(n)}}{|\lambda_{K}|\tilde{\tau}}.
	\end{align*}
	Note that for the special case when $K=1$, $MMSB(n, P,\Pi)$ degenerates to the Erdos-Renyi random graph with $\mathrm{rank}(\mathrm{L}_{\tau})=1$, the bound of $\|\hat{V}\mathrm{sgn}(H)-V\|_{2\rightarrow\infty}$ is consistent with Corollary 3 in \cite{chen2021asymmetry}. Generally, by Lemmas \ref{P2} and \ref{P4}, we can set $\mu=O(1), \kappa(\mathscr{L}_{\tau})=O(1)$, then by Lemma \ref{P4}, we have
	\begin{align*}
	\|\hat{V}\mathrm{sgn}(H)-V\|_{2\rightarrow\infty}=O(\frac{(\tau+\delta_{\mathrm{max}})\sqrt{K\mathrm{log}(n)}}{\tilde{\tau}|\lambda_{K}(\tilde{P})|\lambda_{K}(\Pi'\Pi)\sqrt{\rho}}).
	\end{align*}
	
	Second, we apply the principal subspace perturbation introduced in Lemma 5.1 \citep{lei2015consistency} to bound $\|V-\hat{V}\mathrm{sgn}(H)\|_{F}$. We write this lemma as below
	\begin{lem}\label{PSP}
		(Principal subspace perturbation \citep{lei2015consistency}). Assume that $X\in\mathbb{R}^{n\times n}$ is a rank $K$ symmetric matrix with smallest nonzero singular value $\sigma_{K}(X)$. Let $\hat{X}$ be any symmetric matrix and $\hat{U},U\in\mathbb{R}^{n\times K}$ be the $K$ leading eigenvectors of $\hat{X}$ and $X$, respectively. Then there exists a $K\times K$ orthogonal matrix $\hat{O}$ such that
		\begin{align*}
		\|U-\hat{U}\hat{O}\|_{F}\leq \frac{2\sqrt{2K}\|\hat{X}-X||}{\sigma_{K}(X)}.
		\end{align*}
	\end{lem}
	Let $\hat{X}=L_{\tau}, X=\mathscr{L}_{\tau}, U=V, \hat{U}=\hat{V}, \sigma_{K}(X)=|\lambda_{K}|$, by Lemma \ref{PSP}, there exists a $K\times K$ orthogonal matrix $\hat{O}$ such that
	\begin{align*}
	\|V-\hat{V}\hat{O}\|_{F}\leq \frac{2\sqrt{2K}\|L_{\tau}-\mathscr{L}_{\tau}||}{|\lambda_{K}|}.
	\end{align*}
	By the proof of Theorem 2 \citep{yu2015a}, we know that $\hat{O}=\mathrm{sgn}(H)$, combine it with Lemmas \ref{P4} and 5.2, we see that with probability at least $1-o(\frac{K^{4\beta}}{n^{4\alpha-1}})$,
	\begin{align*}
	\|V-\hat{V}\hat{O}\|_{F}=O(\frac{(\tau+\delta_{\mathrm{max}})\sqrt{Kn\mathrm{log}(n^{\alpha}K^{-\beta})}}{(\tau+\delta_{\mathrm{min}})|\lambda_{K}(\tilde{P})|\lambda_{K}(\Pi'\Pi)\sqrt{\rho}}).
	\end{align*}
	Now we are ready to bound $\|\hat{V}\hat{V}'-VV'\|_{2\rightarrow\infty}$. Since
	\begin{align*}
	&\|\hat{V}\hat{V}'-VV'\|_{2\rightarrow\infty}=\mathrm{max}_{1\leq i\leq n}\|e'_{i}(VV'-\hat{V}\hat{V}')\|_{F}\\
	&=\mathrm{max}_{1\leq i\leq n}\|e'_{i}(VV'-\hat{V}\mathrm{sgn}(H)V'+\hat{V}\mathrm{sgn}(H)V'-\hat{V}\hat{V}')\|_{F}\\
	&\overset{\mathrm{By~Lemma~}A1 in Yu et. al. (2015)}{\leq}\mathrm{max}_{1\leq i\leq n}\|e'_{i}(V-\hat{V}\mathrm{sgn}(H))\|_{F}+\mathrm{max}_{1\leq i\leq n}\|e'_{i}\hat{V}(\mathrm{sgn}(H)V'-\hat{V}')\|_{F}\\
	&=\|V-\hat{V}\mathrm{sgn}(H)\|_{2\rightarrow\infty}+\mathrm{max}_{1\leq i\leq n}\|e'_{i}\hat{V}(\mathrm{sgn}(H)V'-\hat{V}')\|_{F}\\
	&\leq\|V-\hat{V}\mathrm{sgn}(H)\|_{2\rightarrow\infty}+\mathrm{max}_{1\leq i\leq n}\|e'_{i}\hat{V}\|_{F}\|\mathrm{sgn}(H)V'-\hat{V}'\|_{F}\\
	&=\|V-\hat{V}\mathrm{sgn}(H)\|_{2\rightarrow\infty}+\mathrm{max}_{1\leq i\leq n}\|e'_{i}\hat{V}\|_{F}\|V-\hat{V}\mathrm{sgn}(H)\|_{F}\\
	&=\|V-\hat{V}\mathrm{sgn}(H)\|_{2\rightarrow\infty}+\mathrm{max}_{1\leq i\leq n}\|e'_{i}(\hat{V}\mathrm{sgn}(H)-V+V)\|_{F}\|V-\hat{V}\mathrm{sgn}(H)\|_{F}\\
	&\leq\|V-\hat{V}\mathrm{sgn}(H)\|_{2\rightarrow\infty}+(\|\hat{V}\mathrm{sgn}(H)-V\|_{2\rightarrow\infty}+\|V\|_{2\rightarrow\infty})\|V-\hat{V}\mathrm{sgn}(H)\|_{F}\\
	&\overset{\mathrm{By~Lemma~}\ref{P2}}{\leq}\|V-\hat{V}\mathrm{sgn}(H)\|_{2\rightarrow\infty}+(\|\hat{V}\mathrm{sgn}(H)-V\|_{2\rightarrow\infty}+\sqrt{\frac{\tau+\delta_{\mathrm{max}}}{\tau+\delta_{\mathrm{min}}}}\frac{1}{\sqrt{\lambda_{K}(\Pi'\Pi)}})\|V-\hat{V}\mathrm{sgn}(H)\|_{F}\\
	&=O(\frac{(\tau+\delta_{\mathrm{max}})\sqrt{K\mathrm{log}(n)}}{\tilde{\tau} |\lambda_{K}(\tilde{P})|\lambda_{K}(\Pi'\Pi)\sqrt{\rho}})+(O(\frac{(\tau+\delta_{\mathrm{max}})\sqrt{K\mathrm{log}(n)}}{\tilde{\tau} |\lambda_{K}(\tilde{P})|\lambda_{K}(\Pi'\Pi)\sqrt{\rho}})+\sqrt{\frac{\tau+\delta_{\mathrm{max}}}{\tau+\delta_{\mathrm{min}}}}\frac{1}{\sqrt{\lambda_{K}(\Pi'\Pi)}})\\
	&\times O(\frac{(\tau+\delta_{\mathrm{max}})\sqrt{Kn\mathrm{log}(n^{\alpha}K^{-\beta})}}{(\tau+\delta_{\mathrm{min}})|\lambda_{K}(\tilde{P})|\lambda_{K}(\Pi'\Pi)\sqrt{\rho}})=O(\frac{(\tau+\delta_{\mathrm{max}})\sqrt{K\mathrm{log}(n^{\alpha}K^{-\beta})}}{\tilde{\tau} |\lambda_{K}(\tilde{P})|\lambda_{K}(\Pi'\Pi)\sqrt{\rho}})\\
	&=O(\frac{(\tau+\delta_{\mathrm{max}})\sqrt{K\mathrm{log}(n^{\alpha}K^{-\beta})}}{(\tau+\delta_{\mathrm{min}}) |\lambda_{K}(\tilde{P})|\lambda_{K}(\Pi'\Pi)\sqrt{\rho}}),
	\end{align*}
	where we use $\tau+\delta_{\mathrm{min}}$ to replace $\tilde{\tau}$ for convenience since $\tilde{\tau}=O(\tau+\delta_{\mathrm{min}})$.
	\begin{rem}
		Actually, we can also obtain the row-wise eigenvector deviation $\|\hat{V}\hat{V}'-VV'\|_{2\rightarrow\infty}$ based on an application of Theorem 4.2 \citep{cape2019the}, where this theorem gives that $\|V-\hat{V}\mathrm{sgn}(H)\|_{2\rightarrow\infty}=O(\frac{\mathrm{max}_{i}\sum_{j}|L_{\tau}(i,j)-\mathscr{L}_{\tau}(i,j)|}{|\lambda_{K}|}\|V\|_{2\rightarrow\infty})$ under the condition that $|\lambda_{K}|\geq 4\mathrm{max}_{i}\sum_{j}|L_{\tau}(i,j)-\mathscr{L}_{\tau}(i,j)|$. As long as $O(\mathrm{max}_{i}\sum_{j}|L_{\tau}(i,j)-\mathscr{L}_{\tau}(i,j)|)=O(err_{n})$, combine it with Lemmas \ref{P2} and \ref{P4}, then we obtain the theoretical bound in Lemma 5.4.
	\end{rem}
\end{proof}
\subsection{Proof of Lemma 5.6}
\begin{proof}
	\begin{itemize}
		\item For SRSC algorithm, we apply the following theorem which is Theorem 1.1 in \cite{gillis2015semidefinite}.
		\begin{thm}\label{gillis2015siamSP}
			(Theorem 1.1 in \cite{gillis2015semidefinite}) Fix $m\geq r$ and $n\geq r$. Consider a matrix $Q=RS+T$, where $S\in\mathbb{R}^{r\times m}$ has a full row rank, $R\in \mathbb{R}^{n\times r}$ is a nonnegative matrix such that the sum of each row is at most 1, and $T\in \mathbb{R}^{n\times m}$. Suppose $R$ has a submatrix equal to $I_{r}$. Write $\epsilon=\mathrm{max}_{1\leq i\leq n}\|T(i,:)\|$. Suppose $\epsilon=O(\frac{\sigma_{\mathrm{min}}(S)}{\sqrt{r}\kappa^{2}(S)})$, where $\sigma_{\mathrm{min}}(S)$ and $\kappa(S)$ are the minimum singular value and condition number of $S$, respectively. If we apply the SP algorithm to rows of $Q$, then it outputs an index set $\mathcal{K}\subset \{1,2,\ldots, n\}$ such that $|\mathcal{K}|=r$ and $\mathrm{max}_{1\leq k\leq r}\mathrm{min}_{j\in\mathcal{K}}\|S(k,:)-Q(j,:)\|=O(\epsilon \kappa^{2}(S))$.
		\end{thm}
		Set $m=n, r=K, Q=\hat{V}_{\tau,2}, T=\hat{V}_{\tau,2}-V_{\tau,2}, S=V_{\tau,2}(\mathcal{I},:), R=\Pi$ and $\epsilon=\|\hat{V}_{\tau,2}-V_{\tau,2}\|_{2\rightarrow\infty}$. By condition (I2), $R$ has an identity submatrix $I_{K}$ and all entries of $R$ are nonnegative. Now, use Theorem \ref{gillis2015siamSP}, there exists a permutation matrix $\mathcal{P}$ such that
		\begin{align*}
		\|\hat{V}_{\tau,2}(\mathcal{\hat{I}},:)-\mathcal{P}V_{\tau,2}(\mathcal{I},:)\|_{F}=O(\epsilon\kappa^{2}(V_{\tau,2}(\mathcal{I},:))\sqrt{K}).
		\end{align*}
		Next, we bound $\epsilon$ as below:
		\begin{align*}
		&\epsilon=\|\hat{V}_{\tau,2}-V_{\tau,2}\|_{2\rightarrow\infty}=\mathrm{max}_{1\leq i\leq n}\|\hat{V}_{\tau,2}(i,:)-V_{\tau,2}(i,:)\|_{F}=\mathrm{max}_{1\leq i\leq n}\|D^{1/2}_{\tau}(i,i)\hat{V}_{2}(i,:)-\mathscr{D}^{1/2}_{\tau}V_{2}(i,:)\|_{F}\\
		&=\mathrm{max}_{1\leq i\leq n}\|(D^{1/2}_{\tau}(i,i)-\mathscr{D}^{1/2}_{\tau}(i,i))\hat{V}_{2}(i,:)+\mathscr{D}^{1/2}_{\tau}(i,i)(\hat{V}_{2}(i,:)-V_{2}(i,:))\|_{F}\\
		&\leq\mathrm{max}_{1\leq i\leq n}(|D^{1/2}_{\tau}(i,i)-\mathscr{D}^{1/2}_{\tau}(i,i)|\|\hat{V}_{2}(i,:)\|_{F}+\mathscr{D}^{1/2}_{\tau}(i,i)\|\hat{V}_{2}(i,:)-V_{2}(i,:)\|_{F})\\
		&\leq\mathrm{max}_{1\leq i\leq n}|D^{1/2}_{\tau}(i,i)-\mathscr{D}^{1/2}_{\tau}(i,i)|\|\hat{V}_{2}(i,:)-V_{2}(i,:)+V_{2}(i,:)\|_{F}+\mathrm{max}_{1\leq i\leq n}\mathscr{D}^{1/2}_{\tau}(i,i)\varpi\\
		&\leq\mathrm{max}_{1\leq i\leq n}|D^{1/2}_{\tau}(i,i)-\mathscr{D}^{1/2}_{\tau}(i,i)|(\|\hat{V}_{2}(i,:)-V_{2}(i,:)\|_{F}+\|V_{2}(i,:)\|_{F})+\varpi\sqrt{\tau+\delta_{\mathrm{max}}}\\
		&\leq\mathrm{max}_{1\leq i\leq n}|D^{1/2}_{\tau}(i,i)-\mathscr{D}^{1/2}_{\tau}(i,i)|(\varpi+\|V_{2}(i,:)\|_{F})+\varpi\sqrt{\tau+\delta_{\mathrm{max}}}\\
		&=\mathrm{max}_{1\leq i\leq n}|D^{1/2}_{\tau}(i,i)-\mathscr{D}^{1/2}_{\tau}(i,i)|(\varpi+\|V(i,:)V'\|_{F})+\varpi\sqrt{\tau+\delta_{\mathrm{max}}}\\
		&=\mathrm{max}_{1\leq i\leq n}|D^{1/2}_{\tau}(i,i)-\mathscr{D}^{1/2}_{\tau}(i,i)|(\varpi+\|V(i,:)\|_{F})+\varpi\sqrt{\tau+\delta_{\mathrm{max}}}\\
		&\overset{\mathrm{By~Lemmas~}\ref{P2}}{\leq}\mathrm{max}_{1\leq i\leq n}|D^{1/2}_{\tau}(i,i)-\mathscr{D}^{1/2}_{\tau}(i,i)|(\varpi+\sqrt{\frac{\tau+\delta_{\mathrm{max}}}{(\tau+\delta_{\mathrm{min}})\lambda_{K}(\Pi'\Pi)}})+\varpi\sqrt{\tau+\delta_{\mathrm{max}}}\\
		&=\mathrm{max}_{1\leq i\leq n}|(1-D^{1/2}_{\tau}(i,i)\mathscr{D}^{-1/2}_{\tau}(i,i))\mathscr{D}^{1/2}_{\tau}(i,i)|(\varpi+\sqrt{\frac{\tau+\delta_{\mathrm{max}}}{(\tau+\delta_{\mathrm{min}})\lambda_{K}(\Pi'\Pi)}})+\varpi\sqrt{\tau+\delta_{\mathrm{max}}}\\
		&=\mathrm{max}_{1\leq i\leq n}\mathscr{D}^{1/2}_{\tau}(i,i)|1-D^{1/2}_{\tau}(i,i)\mathscr{D}^{-1/2}_{\tau}(i,i)|(\varpi+\sqrt{\frac{\tau+\delta_{\mathrm{max}}}{(\tau+\delta_{\mathrm{min}})\lambda_{K}(\Pi'\Pi)}})+\varpi\sqrt{\tau+\delta_{\mathrm{max}}}\\
		&\leq\|I-D^{1/2}_{\tau}\mathscr{D}^{-1/2}_{\tau}\|(\varpi+\sqrt{\frac{\tau+\delta_{\mathrm{max}}}{(\tau+\delta_{\mathrm{min}})\lambda_{K}(\Pi'\Pi)}})\sqrt{\tau+\delta_{\mathrm{max}}}+\varpi\sqrt{\tau+\delta_{\mathrm{max}}}\\
		&\mathrm{By~the~proof~of~Lemma~}5.2,~\mathrm{with~probability~at~least~}1-o(\frac{K^{4\beta}}{n^{4\alpha-1}})\\
		&\leq err_{n}(\varpi+\sqrt{\frac{\tau+\delta_{\mathrm{max}}}{(\tau+\delta_{\mathrm{min}})\lambda_{K}(\Pi'\Pi)}})\sqrt{\tau+\delta_{\mathrm{max}}}+\varpi\sqrt{\tau+\delta_{\mathrm{max}}}=O(\varpi\sqrt{\tau+\delta_{\mathrm{max}}}),
		\end{align*}
		where we have used the fact that $err_{n}=O(\frac{\sqrt{\rho n\mathrm{log}(n^{\alpha}K^{-\beta})}}{\tau+\delta_{\mathrm{min}}})\leq 1$ when $\tau+\delta_{\mathrm{min}}\geq C\sqrt{\rho n\mathrm{log}(n^{\alpha}K^{-\beta})}$, and this fact gives $\mathrm{max}(err_{n}\varpi,\varpi)\sqrt{\tau+\delta_{\mathrm{max}}}=\varpi\sqrt{\tau+\delta_{\mathrm{max}}}$.
		
		For $\kappa^{2}(V_{\tau,2}(\mathcal{I},:))$, since $\kappa^{2}(V_{\tau,2}(\mathcal{I},:))=\kappa(V_{\tau,2}(\mathcal{I},:)V'_{\tau,2}(\mathcal{I},:))=\kappa(\mathscr{D}_{\tau}(\mathcal{I},\mathcal{I})V_{2}(\mathcal{I},:)V'_{2}(\mathcal{I},:)\mathscr{D}^{1/2}_{\tau}(\mathcal{I},\mathcal{I}))=\kappa(\mathscr{D}_{\tau}(\mathcal{I},\mathcal{I})V(\mathcal{I},:)V'(\mathcal{I},:)\mathscr{D}^{1/2}_{\tau}(\mathcal{I},\mathcal{I}))=\kappa(V_{\tau,1}(\mathcal{I},:)V'_{\tau,1}(\mathcal{I},:))$, by Lemma \ref{P3}, we have $1\leq\kappa^{2}(V_{\tau,2}(\mathcal{I},:))\leq\frac{\tau+\delta_{\mathrm{max}}}{\tau+\delta_{\mathrm{min}}}\kappa(\Pi'\Pi)$.
		By Theorem \ref{gillis2015siamSP}, we have
		\begin{align*}
		&\|\hat{V}_{\tau,2}(\mathcal{\hat{I}},:)-\mathcal{P}V_{\tau,2}(\mathcal{I},:)\|_{F}=O(\epsilon\kappa^{2}(V_{\tau,2}(\mathcal{I},:))\sqrt{K})\leq O(\frac{(\tau+\delta_{\mathrm{max}})^{1.5}\sqrt{K}\varpi\kappa(\Pi'\Pi)}{\tau+\delta_{\mathrm{min}}}).
		\end{align*}
		\begin{rem}\label{SPreturnsTrueIdeally}
			For Ideal SRSC, we have $m=n, r=K, Q=V_{\tau,1}, T=0, S=V_{\tau,1}(\mathcal{I},:), R=\Pi$. Since $\epsilon=\|T\|_{2\rightarrow\infty}=0$, we see that the index set returned by SP algorithm is actually $\mathcal{I}$ up to a permutation by Theorem \ref{gillis2015siamSP}, and this is the reason that we state our Ideal SRSC exactly returns $\Pi$ based on the fact that the SP algorithm exactly returns $\mathcal{I}$ when $V_{\tau,1}$ has the ideal simplex structure $V_{\tau,1}=\Pi V_{\tau,1}(\mathcal{I},:)$. Similar arguments hold for the Ideal SRSC-equivalence.
		\end{rem}
		\item For CRSC algorithm, by Lemma 3.4, we see that $V_{*,1}(\mathcal{I},:)$ satisfies condition 1 in \cite{MaoSVM}.
		Meanwhile, since $(V_{*,1}(\mathcal{I},:)V'_{*,1}(\mathcal{I},:))^{-1}\mathbf{1}>0$, we have $(V_{*,1}(\mathcal{I},:)V'_{*,1}(\mathcal{I},:))^{-1}\mathbf{1}\geq \eta\mathbf{1}$, hence $V_{*,1}(\mathcal{I},:)$ satisfies condition 2 in \cite{MaoSVM}. Now, we give a lower bound for $\eta$ to show that $\eta$ is strictly positive. By the proof of Lemma \ref{P2}, we have $(V(\mathcal{I},:)V'(\mathcal{I},:))^{-1}=\mathscr{D}^{1/2}_{\tau}(\mathcal{I},\mathcal{I})\Pi'\mathscr{D}^{-1}_{\tau}\Pi\mathscr{D}^{1/2}_{\tau}(\mathcal{I},\mathcal{I})$, which gives that
		\begin{align*}
		&(V_{*,1}(\mathcal{I},:)V'_{*,1}(\mathcal{I},:))^{-1}=(N(\mathcal{I},\mathcal{I})V(\mathcal{I},:)V'(\mathcal{I},:)N(\mathcal{I},\mathcal{I}))^{-1}\\
		&=N^{-1}(\mathcal{I},\mathcal{I})\mathscr{D}^{1/2}_{\tau}(\mathcal{I},\mathcal{I})\Pi'\mathscr{D}^{-1}_{\tau}\Pi\mathscr{D}^{1/2}_{\tau}(\mathcal{I},\mathcal{I})N^{-1}(\mathcal{I},\mathcal{I})\\
		&\geq \frac{\theta^{2}_{\mathrm{min}}}{\theta^{2}_{\mathrm{max}}N^{2}_{\mathrm{max}}}\Pi'\Pi,
		\end{align*}
		where we set $N_{\mathrm{max}}=\mathrm{max}_{1\leq i\leq n}N(i,i)$. By the proof of Lemma 5.8, we have $N_{\mathrm{max}}\leq\sqrt{\frac{(\tau+\delta_{\mathrm{max}})K\lambda_{1}(\Pi'\Pi)}{\tau+\delta_{\mathrm{min}}}}$, which gives that
		\begin{align*}
		(V_{*,1}(\mathcal{I}_{*},:)V'_{*,1}(\mathcal{I},:))^{-1}\geq (\frac{\tau+\delta_{\mathrm{min}}}{\tau+\delta_{\mathrm{max}}})^{2}\frac{\Pi'\Pi}{K\lambda_{1}(\Pi'\Pi)}.
		\end{align*}
		Then we have
		\begin{align*}
		&\eta=\mathrm{min}_{1\leq k\leq K}((V_{*,1}(\mathcal{I},:)V'_{*,1}(\mathcal{I},:))^{-1}\mathbf{1})(k)\geq \mathrm{min}_{1\leq k\leq K}(\frac{\tau+\delta_{\mathrm{min}}}{\tau+\delta_{\mathrm{max}}})^{2}\frac{e'_{k}\Pi'\Pi\mathbf{1}}{K\lambda_{1}(\Pi'\Pi)}\\
		&=\mathrm{min}_{1\leq k\leq K}(\frac{\tau+\delta_{\mathrm{min}}}{\tau+\delta_{\mathrm{max}}})^{2}\frac{e'_{k}\Pi'\mathbf{1}}{K\lambda_{1}(\Pi'\Pi)}=\frac{(\tau+\delta_{\mathrm{min}})^{2}\pi_{\mathrm{min}}}{(\tau+\delta_{\mathrm{max}})^{2}K\lambda_{1}(\Pi'\Pi)},
		\end{align*}
		i.e., $\eta$ is strictly positive. By Lemma 4.6, we have $V_{*,2}(\mathcal{I},:)V'_{*,2}(\mathcal{I},:)\equiv V_{*,1}(\mathcal{I},:)V'_{*,1}(\mathcal{I},:)$, hence $V_{*,2}(\mathcal{I},:)$ also satisfies conditions 1 and 2 in \cite{MaoSVM}.  The above analysis shows that we can directly apply Lemma F.1 of \cite{MaoSVM} since the Ideal CRSC algorithm satisfies conditions 1 and 2 in \cite{MaoSVM}, therefore there exists a permutation matrix $\mathcal{P}_{*}\in\mathbb{R}^{K\times K}$ such that
		\begin{align*}
		\|\hat{V}_{*,2}(\mathcal{\hat{I}}_{*},:)-\mathcal{P}_{*}V_{*,2}(\mathcal{I},:)\|_{F}= O(\frac{K\zeta\epsilon_{*}}{\lambda^{1.5}_{K}(V_{*,2}(\mathcal{I},:))V'_{*,2}(\mathcal{I},:)}),
		\end{align*}
		where $\zeta\leq \frac{4K}{\eta\lambda^{1.5}_{K}(V_{*,2}(\mathcal{I},:)V'_{*,2}(\mathcal{I},:))}=O(\frac{K}{\eta\lambda^{1.5}_{K}(V_{*,1}(\mathcal{I},:)V'_{*,1}(\mathcal{I},:))})$, and $\epsilon_{*}=\mathrm{max}_{1\leq i\leq n}\|\hat{V}_{*,2}(i,:)-V_{*,2}(i,:)\|$. Next we bound $\epsilon_{*}$.
		
		Since
		\begin{align*}
		&\|\hat{V}_{*,2}(i,:)-V_{*,2}(i,:)\|_{F}=\|\frac{\hat{V}_{2}(i,:)\|V_{2}(i,:)\|_{F}-V_{2}(i,:)\|\hat{V}_{2}(i,:)\|_{F}}{\|\hat{V}_{2}(i,:)\|_{F}\|V_{2}(i,:)\|_{F}}\|_{F}\leq\frac{2\|\hat{V}_{2}(i,:)-V_{2}(i,:)\|_{F}}{\|V_{2}(i,:)\|_{F}}\\
		&\leq \frac{2\|\hat{V}_{2}-V_{2}\|_{2\rightarrow\infty}}{\|V_{2}(i,:)\|_{F}}\leq\frac{2\varpi}{\|V_{2}(i,:)\|_{F}}=\frac{2\varpi}{\|(VV')(i,:)\|_{F}}=\frac{2\varpi}{\|V(i,:)V'\|_{F}}=\frac{2\varpi}{\|V(i,:)\|_{F}}\\
		&\leq 2\varpi\sqrt{\frac{(\tau+\delta_{\mathrm{max}})K\lambda_{1}(\Pi'\Pi)}{\tau+\delta_{\mathrm{min}}}},
		\end{align*}
		where the first inequality holds by lemma F.2 \cite{MaoSVM},  the last inequality holds by Lemma \ref{P2}. Now we have $\epsilon_{*}=O(\varpi\sqrt{\frac{(\tau+\delta_{\mathrm{max}})K\lambda_{1}(\Pi'\Pi)}{\tau+\delta_{\mathrm{min}}}})$. Finally, by Lemma \ref{P3}, we have
		\begin{align*}
		\|\hat{V}_{*,2}(\mathcal{\hat{I}}_{*},:)-\mathcal{P}_{*}V_{*,2}(\mathcal{I},:)\|_{F}=O((\frac{\tau+\delta_{\mathrm{max}}}{\tau+\delta_{\mathrm{min}}})^{3.5}\frac{K^{2.5}\varpi\kappa^{3}(\Pi'\Pi)\sqrt{\lambda_{1}(\Pi'\Pi)}}{\eta}).
		\end{align*}
	\end{itemize}
\end{proof}

\subsection{Proof of Lemma 5.7}
\begin{proof}
	For convenience, set $V_{*,1}(\mathcal{I},:)=V_{C}, \hat{V}_{*,1}(\mathcal{\hat{I}}_{*},:)=\hat{V}_{C}, V_{*,2}(\mathcal{I},:)=V_{2C}, \hat{V}_{*,2}(\mathcal{\hat{I}}_{*},:)=\hat{V}_{2C}$. We bound $\|e'_{i}(\hat{Y}_{*}-Y_{*}\mathcal{P})\|_{F}$ when the input is $\hat{V}_{*,1}$ in SVM-cone using below technique which follows the proof idea of Theorem 3.5 in \cite{MaoSVM}.
	\begin{align*}
	&\|e'_{i}(\hat{Y}_{*}-Y_{*}\mathcal{P}_{*})\|_{F}=\|e'_{i}(\hat{V}\hat{V}'_{C}(\hat{V}_{C}\hat{V}'_{C})^{-1}-VV'_{C}(V_{C}V'_{C})^{-1}\mathcal{P}_{*})\|_{F}\\
	&=\|e'_{i}(\hat{V}-V(V'\hat{V}))\hat{V}'_{C}(\hat{V}_{C}\hat{V}'_{C})^{-1}+e'_{i}(V(V'\hat{V})\hat{V}'_{C}(\hat{V}_{C}\hat{V}'_{C})^{-1}-V(V'\hat{V})(\mathcal{P}'_{*}(V_{C}V'_{C})(V'_{C})^{-1}(V'\hat{V}))^{-1})\|_{F}\\
	&\leq\|e'_{i}(\hat{V}-V(V'\hat{V}))\hat{V}'_{C}(\hat{V}_{C}\hat{V}'_{C})^{-1}\|_{F}+\|e'_{i}V(V'\hat{V})(\hat{V}'_{C}(\hat{V}_{C}\hat{V}'_{C})^{-1}-(\mathcal{P}'_{*}(V_{C}V'_{C})(V'_{C})^{-1}(V'\hat{V}))^{-1})\|_{F}\\
	&\leq \|e'_{i}(\hat{V}-V(V'\hat{V}))\|_{F}\|\hat{V}'_{C}\|_{F}\|(\hat{V}_{C}\hat{V}'_{C})^{-1}\|_{F}+\|e'_{i}V(V'\hat{V})(\hat{V}'_{C}(\hat{V}_{C}\hat{V}'_{C})^{-1}-(\mathcal{P}'_{*}(V_{C}V'_{C})(V'_{C})^{-1}(V'\hat{V}))^{-1})\|_{F}\\
	&= \sqrt{K}\|e'_{i}(\hat{V}-V(V'\hat{V}))\|_{F}\|(\hat{V}_{C}\hat{V}'_{C})^{-1}\|_{F}+\|e'_{i}V(V'\hat{V})(\hat{V}^{-1}_{C}-(\mathcal{P}_{*}'V_{C}(V'\hat{V}))^{-1})\|_{F}\\
	&\leq K\|e'_{i}(\hat{V}-V(V'\hat{V}))\|_{F}/\lambda_{K}(\hat{V}_{C}\hat{V}'_{C})+\|e'_{i}V(V'\hat{V})(\hat{V}^{-1}_{C}-(\mathcal{P}_{*}'V_{C}(V'\hat{V}))^{-1})\|_{F}\\
	&=
	K\|e'_{i}(\hat{V}\hat{V}'-VV')\hat{V}\|_{F}O(\frac{\tau+\delta_{\mathrm{max}}}{\tau+\delta_{\mathrm{min}}}\kappa(\Pi'\Pi))+\|e'_{i}V(V'\hat{V})(\hat{V}^{-1}_{C}-(\mathcal{P}_{*}'V_{C}(V'\hat{V}))^{-1})\|_{F}\\
	&\mathrm{By~Lemma~A.1~in~ Yu~ et~ al.~ (2015)} \mathrm{~or~Remark~3.2~in~ Mao~ et~ al.~ (2020)}\\
	&\leq K\|e'_{i}(\hat{V}\hat{V}'-VV')\|_{F}O(\frac{\tau+\delta_{\mathrm{max}}}{\tau+\delta_{\mathrm{min}}}\kappa(\Pi'\Pi))+\|e'_{i}V(V'\hat{V})(\hat{V}^{-1}_{C}-(\mathcal{P}'_{*}V_{C}(V'\hat{V}))^{-1})\|_{F}\\
	&=K\varpi O(\frac{\tau+\delta_{\mathrm{max}}}{\tau+\delta_{\mathrm{min}}}\kappa(\Pi'\Pi))+\|e'_{i}V(V'\hat{V})(\hat{V}^{-1}_{C}-(\mathcal{P}'_{*}V_{C}(V'\hat{V}))^{-1})\|_{F}\\
	&=O(\frac{\tau+\delta_{\mathrm{max}}}{\tau+\delta_{\mathrm{min}}}K\varpi\kappa(\Pi'\Pi))+\|e'_{i}V(V'\hat{V})(\hat{V}^{-1}_{C}-(\mathcal{P}'_{*}V_{C}(V'\hat{V}))^{-1})\|_{F},
	\end{align*}
	where we have used similar idea in the proof of Lemma G.3 in \cite{mao2020estimating} such that apply $O(\frac{1}{\lambda_{K}(V_{C}V'_{C})})$ to estimate $\frac{1}{\lambda_{K}(\hat{V}_{C}\hat{V}'_{C})}$, then by Lemma \ref{P3}, we have $\frac{1}{\lambda_{K}(\hat{V}_{C}\hat{V}'_{C})}\leq O(\frac{\tau+\delta_{\mathrm{max}}}{\tau+\delta_{\mathrm{min}}}\kappa(\Pi'\Pi))$.
	
	Now we aim to bound $\|e'_{i}V(V'\hat{V})(\hat{V}^{-1}_{C}-(\mathcal{P}_{*}'V_{C}(V'\hat{V}))^{-1})\|_{F}$. For convenience, set $T=V'\hat{V}, S=\mathcal{P}_{*}'V_{C}T$. We have
	\begin{align*}
	&\|e'_{i}V(V'\hat{V})(\hat{V}^{-1}_{C}-(\mathcal{P}'_{*}V_{C}(V'\hat{V}))^{-1})\|_{F}=\|e'_{i}VTS^{-1}(S-\hat{V}_{C})\hat{V}^{-1}_{C}\|_{F}\\
	&\leq\|e'_{i}VTS^{-1}(S-\hat{V}_{C})\|_{F}\|\hat{V}^{-1}_{C}\|_{F}\leq\|e'_{i}VTS^{-1}(S-\hat{V}_{C})\|_{F}\frac{\sqrt{K}}{|\lambda_{K}(\hat{V}_{C})|}\\
	&=\|e'_{i}VTS^{-1}(S-\hat{V}_{C})\|_{F}\frac{\sqrt{K}}{\sqrt{\lambda_{K}(\hat{V}_{C}\hat{V}'_{C})}}\leq\|e'_{i}VTS^{-1}(S-\hat{V}_{C})\|_{F}O(\sqrt{\frac{\tau+\delta_{\mathrm{max}}}{\tau+\delta_{\mathrm{min}}}K\kappa(\Pi'\Pi)})\\
	&=\|e'_{i}VTT^{-1}V'_{C}(V_{C}V'_{C})^{-1}\mathcal{P}_{*}(S-\hat{V}_{C})\|_{F}O(\sqrt{\frac{\tau+\delta_{\mathrm{max}}}{\tau+\delta_{\mathrm{min}}}K\kappa(\Pi'\Pi)})\\
	&=\|e'_{i}VV'_{C}(V_{C}V'_{C})^{-1}\mathcal{P}_{*}(S-\hat{V}_{C})\|_{F}O(\sqrt{\frac{\tau+\delta_{\mathrm{max}}}{\tau+\delta_{\mathrm{min}}}K\kappa(\Pi'\Pi)})\\
	&=\|e'_{i}Y_{*}\mathcal{P}_{*}(S-\hat{V}_{C})\|_{F}O(\sqrt{\frac{\tau+\delta_{\mathrm{max}}}{\tau+\delta_{\mathrm{min}}}K\kappa(\Pi'\Pi)})\\
	&\leq\|e'_{i}Y_{*}\|_{F}\|S-\hat{V}_{C}\|_{F}O(\sqrt{\frac{\tau+\delta_{\mathrm{max}}}{\tau+\delta_{\mathrm{min}}}K\kappa(\Pi'\Pi)})\\
	&\overset{\mathrm{By~the~proof~of~Lemma~}5.8}{\leq}\frac{\tau+\delta_{\mathrm{max}}}{\tau+\delta_{\mathrm{min}}}\sqrt{\frac{K\kappa(\Pi'\Pi)}{\lambda_{K}(\Pi'\Pi)}}\|S-\hat{V}_{C}\|_{F}O(\sqrt{\frac{\tau+\delta_{\mathrm{max}}}{\tau+\delta_{\mathrm{min}}}K\kappa(\Pi'\Pi)})\\
	&=\|\hat{V}_{C}-\mathcal{P}_{*}'V_{C}V'\hat{V}\|_{F}O((\frac{\tau+\delta_{\mathrm{max}}}{\tau+\delta_{\mathrm{min}}})^{1.5}\frac{K\kappa(\Pi'\Pi)}{\sqrt{\lambda_{K}(\Pi'\Pi)}})=\|(\hat{V}_{C}\hat{V}'-\mathcal{P}'_{*}V_{C}V')\hat{V}\|_{F}O((\frac{\tau+\delta_{\mathrm{max}}}{\tau+\delta_{\mathrm{min}}})^{1.5}\frac{K\kappa(\Pi'\Pi)}{\sqrt{\lambda_{K}(\Pi'\Pi)}})\\
	&\leq\|\hat{V}_{C}\hat{V}'-\mathcal{P}_{*}'V_{C}V'\|_{F}O((\frac{\tau+\delta_{\mathrm{max}}}{\tau+\delta_{\mathrm{min}}})^{1.5}\frac{K\kappa(\Pi'\Pi)}{\sqrt{\lambda_{K}(\Pi'\Pi)}})\\
	&=\|\hat{V}_{C}\hat{V}'-\mathcal{P}_{*}'V_{C}V'\|_{F}O((\frac{\tau+\delta_{\mathrm{max}}}{\tau+\delta_{\mathrm{min}}})^{1.5}\frac{K\kappa(\Pi'\Pi)}{\sqrt{\lambda_{K}(\Pi'\Pi)}})\\
	&\overset{\mathrm{By~Lemma~}4.6}{=}\|\hat{V}_{2C}-\mathcal{P}_{*}'V_{2C}\|_{F}O((\frac{\tau+\delta_{\mathrm{max}}}{\tau+\delta_{\mathrm{min}}})^{1.5}\frac{K\kappa(\Pi'\Pi)}{\sqrt{\lambda_{K}(\Pi'\Pi)}})\\
	&\leq(\|\hat{V}_{2C}-\mathcal{P}_{*}V_{2C}\|_{F}+\|(\mathcal{P}_{*}-\mathcal{P}_{*}')V_{2C}\|_{F})O((\frac{\tau+\delta_{\mathrm{max}}}{\tau+\delta_{\mathrm{min}}})^{1.5}\frac{K\kappa(\Pi'\Pi)}{\sqrt{\lambda_{K}(\Pi'\Pi)}})\\
	&\overset{\mathrm{By~Lemma~}5.6}{=}(O((\frac{\tau+\delta_{\mathrm{max}}}{\tau+\delta_{\mathrm{min}}})^{3.5}\frac{K^{2.5}\varpi\kappa^{3}(\Pi'\Pi)\sqrt{\lambda_{1}(\Pi'\Pi)}}{\eta})+\|(\mathcal{P}_{*}-\mathcal{P}_{*}')V_{2C}\|_{F})O((\frac{\tau+\delta_{\mathrm{max}}}{\tau+\delta_{\mathrm{min}}})^{1.5}\frac{K\kappa(\Pi'\Pi)}{\sqrt{\lambda_{K}(\Pi'\Pi)}})\\
	&\leq(O((\frac{\tau+\delta_{\mathrm{max}}}{\tau+\delta_{\mathrm{min}}})^{3.5}\frac{K^{2.5}\varpi\kappa^{3}(\Pi'\Pi)\sqrt{\lambda_{1}(\Pi'\Pi)}}{\eta})+K\sqrt{2})O((\frac{\tau+\delta_{\mathrm{max}}}{\tau+\delta_{\mathrm{min}}})^{1.5}\frac{K\kappa(\Pi'\Pi)}{\sqrt{\lambda_{K}(\Pi'\Pi)}})\\
	&=O((\frac{\tau+\delta_{\mathrm{max}}}{\tau+\delta_{\mathrm{min}}})^{5}\frac{K^{3.5}\varpi\kappa^{4.5}(\Pi'\Pi)}{\eta}).
	\end{align*}
	Then, we have
	\begin{align*}
	&\|e'_{i}(\hat{Y}_{*}-Y_{*}\mathcal{P}_{*})\|_{F}\leq O(\frac{\tau+\delta_{\mathrm{max}}}{\tau+\delta_{\mathrm{min}}}K\varpi\kappa(\Pi'\Pi))+\|e'_{i}V(V'\hat{V})(\hat{V}^{-1}_{C}-(\mathcal{P}_{*}'V_{C}(V'\hat{V}))^{-1})\|_{F}\\
	&=O(\frac{\tau+\delta_{\mathrm{max}}}{\tau+\delta_{\mathrm{min}}}K\varpi\kappa(\Pi'\Pi))+O((\frac{\tau+\delta_{\mathrm{max}}}{\tau+\delta_{\mathrm{min}}})^{5}\frac{K^{3.5}\varpi\kappa^{4.5}(\Pi'\Pi)}{\eta})\\
	&=O((\frac{\tau+\delta_{\mathrm{max}}}{\tau+\delta_{\mathrm{min}}})^{5}\frac{K^{3.5}\varpi\kappa^{4.5}(\Pi'\Pi)}{\eta}).
	\end{align*}
\end{proof}

\subsection{Proof of Lemma 5.8}
\begin{proof}
	First, we consider the bound $\|e'_{i}(\hat{Z}-Z\mathcal{P})\|_{F}$ for SRSC algorithm.  Recall that $\hat{Z}=\hat{V}\hat{V}'_{\tau,1}(\mathcal{\hat{I}},:)(\hat{V}_{\tau,1}(\mathcal{\hat{I}},:)\hat{V}'_{\tau,1}(\mathcal{\hat{I}},:))^{-1}$ has similar form as $\hat{Y}_{*}=\hat{V}\hat{V}'_{*,1}(\mathcal{\hat{I}}_{*},:)(\hat{V}_{*,1}(\mathcal{\hat{I}}_{*},:)\hat{V}'_{*,1}(\mathcal{\hat{I}}_{*},:))^{-1}$, the proof for SRSC to bound $\|e'_{i}(\hat{Z}-Z\mathcal{P})\|_{F}$ is similar as the proof of Lemma 5.7, hence we omit most details during the proof. For convenience, set $V_{S}=V_{\tau,1}(\mathcal{I},:), \hat{V}_{S}=\hat{V}_{\tau,1}(\mathcal{\hat{I}},:), V_{2S}=V_{\tau,2}(\mathcal{I},:), \hat{V}_{2S}=\hat{V}_{\tau,2}(\mathcal{\hat{I}},:)$.
	We have
	\begin{align*}
	&\|e'_{i}(\hat{Z}-Z\mathcal{P})\|_{F}=\|e'_{i}(\mathrm{max}(\hat{V}\hat{V}'_{S}(\hat{V}_{S}\hat{V}'_{S})^{-1},0)-VV'_{S}(V_{S}V'_{S})^{-1}\mathcal{P})\|_{F}\\
	&\leq K\|e'_{i}(\hat{V}-V(V'\hat{V}))\|_{F}/\lambda_{K}(\hat{V}_{S}\hat{V}'_{S})+\|e'_{i}V(V'\hat{V})(\hat{V}^{-1}_{S}-(\mathcal{P}'V_{S}(V'\hat{V}))^{-1})\|_{F}\\
	&\leq
	K\|e'_{i}(\hat{V}\hat{V}'-VV')\hat{V}\|_{F}O(\frac{\lambda_{1}(\Pi'\Pi)}{\tau+\delta_{\mathrm{min}}})+\|e'_{i}V(V'\hat{V})(\hat{V}^{-1}_{S}-(\mathcal{P}'V_{S}(V'\hat{V}))^{-1})\|_{F}\\
	&\leq K\|e'_{i}(\hat{V}\hat{V}'-VV')\|_{F}O(\frac{\lambda_{1}(\Pi'\Pi)}{\tau+\delta_{\mathrm{min}}})+\|e'_{i}V(V'\hat{V})(\hat{V}^{-1}_{S}-(\mathcal{P}'V_{S}(V'\hat{V}))^{-1})\|_{F}\\
	&=O(\frac{K\varpi\lambda_{1}(\Pi'\Pi)}{\tau+\delta_{\mathrm{min}}})+\|e'_{i}V(V'\hat{V})(\hat{V}^{-1}_{S}-(\mathcal{P}'V_{S}(V'\hat{V}))^{-1})\|_{F}.
	\end{align*}
	Now we aim to bound $\|e'_{i}V(V'\hat{V})(\hat{V}^{-1}_{S}-(\mathcal{P}'V_{S}(V'\hat{V}))^{-1})\|_{F}$. For convenience, set $T=V'\hat{V}, S=\mathcal{P}'V_{S}T$. We have
	\begin{align*}
	&\|e'_{i}V(V'\hat{V})(\hat{V}^{-1}_{S}-(\mathcal{P}'V_{S}(V'\hat{V}))^{-1})\|_{F}=\|e'_{i}VTS^{-1}(S-\hat{V}_{S})\hat{V}^{-1}_{S}\|_{F}\\
	&\leq\|e'_{i}VTS^{-1}(S-\hat{V}_{S})\|_{F}\|\hat{V}^{-1}_{S}\|_{F}\leq\|e'_{i}VTS^{-1}(S-\hat{V}_{S})\|_{F}\frac{\sqrt{K}}{|\lambda_{K}(\hat{V}_{S})|}\\
	&=\|e'_{i}VTS^{-1}(S-\hat{V}_{S})\|_{F}\frac{\sqrt{K}}{\sqrt{\lambda_{K}(\hat{V}_{S}\hat{V}'_{S})}}\leq\|e'_{i}VTS^{-1}(S-\hat{V}_{S})\|_{F}O(\sqrt{\frac{K\lambda_{1}(\Pi'\Pi)}{\tau+\delta_{\mathrm{min}}}})\\
	&=\|e'_{i}Z\mathcal{P}(S-\hat{V}_{S})\|_{F}O(\sqrt{\frac{K\lambda_{1}(\Pi'\Pi)}{\tau+\delta_{\mathrm{min}}}})\leq\|e'_{i}Z\|_{F}\|S-\hat{V}_{S}\|_{F}O(\sqrt{\frac{K\lambda_{1}(\Pi'\Pi)}{\tau+\delta_{\mathrm{min}}}})\\
	&=\|e'_{i}\mathscr{D}^{-1/2}_{\tau}\Pi\|_{F}\|S-\hat{V}_{S}\|_{F}O(\sqrt{\frac{K\lambda_{1}(\Pi'\Pi)}{\tau+\delta_{\mathrm{min}}}})\leq\|S-\hat{V}_{S}\|_{F}O(\frac{\sqrt{K\lambda_{1}(\Pi'\Pi)}}{\tau+\delta_{\mathrm{min}}})\\
	&=\|\hat{V}_{S}-\mathcal{P}'V_{S}V'\hat{V}\|_{F}O(\frac{\sqrt{K\lambda_{1}(\Pi'\Pi)}}{\tau+\delta_{\mathrm{min}}})=\|(\hat{V}_{S}\hat{V}'-\mathcal{P}'V_{S}V')\hat{V}\|_{F}O(\frac{\sqrt{K\lambda_{1}(\Pi'\Pi)}}{\tau+\delta_{\mathrm{min}}})\\
	&\leq\|\hat{V}_{S}\hat{V}'-\mathcal{P}'V_{S}V'\|_{F}O(\frac{\sqrt{K\lambda_{1}(\Pi'\Pi)}}{\tau+\delta_{\mathrm{min}}})\\
	&\overset{\mathrm{By~Lemma~}4.6}{=}\|\hat{V}_{2S}-\mathcal{P}'V_{2S}\|_{F}O(\frac{\sqrt{K\lambda_{1}(\Pi'\Pi)}}{\tau+\delta_{\mathrm{min}}})\\
	&\leq(\|\hat{V}_{2S}-\mathcal{P}V_{2S}\|_{F}+\|(\mathcal{P}-\mathcal{P}')V_{2S}\|_{F})O(\frac{\sqrt{K\lambda_{1}(\Pi'\Pi)}}{\tau+\delta_{\mathrm{min}}})=O(\frac{\varpi_{S}\sqrt{K\lambda_{1}(\Pi'\Pi)}}{\tau+\delta_{\mathrm{min}}}).
	\end{align*}
	Then, we have
	\begin{align*}
	&\|e'_{i}(\hat{Z}-Z\mathcal{P})\|_{F}\leq
	O(\frac{K\varpi\lambda_{1}(\Pi'\Pi)}{\tau+\delta_{\mathrm{min}}})+O(\frac{\varpi_{S}\sqrt{K\lambda_{1}(\Pi'\Pi)}}{\tau+\delta_{\mathrm{min}}})\\
	&=O(\frac{\sqrt{K\lambda_{1}(\Pi'\Pi)}}{\tau+\delta_{\mathrm{min}}}\mathrm{max}(\varpi\sqrt{K\lambda_{1}(\Pi'\Pi)},\varpi_{S}))=O(\frac{\varpi_{S}\sqrt{K\lambda_{1}(\Pi'\Pi)}}{\tau+\delta_{\mathrm{min}}})\\
	&=O(\frac{(\tau+\delta_{\mathrm{max}})^{1.5}K\varpi\kappa(\Pi'\Pi)\sqrt{\lambda_{1}(\Pi'\Pi)}}{(\tau+\delta_{\mathrm{min}})^{2}}).
	\end{align*}
	Now we aim to obtain the upper bounds of $\|e'_{i}(\hat{Z}_{*}-Z_{*}\mathcal{P}_{*})\|_{F}$ for CRSC. We begin the proof by providing bounds for several items used in our proof.
	\begin{itemize}
		\item For $1\leq i\leq n$, by Lemmas \ref{P2}, we have $N(i,i)=\frac{1}{\|V(i,:)\|_{F}}\leq\sqrt{\frac{(\tau+\delta_{\mathrm{max}})K\lambda_{1}(\Pi'\Pi)}{\tau+\delta_{\mathrm{min}}}}$ and $N(i,i)\geq\sqrt{\frac{(\tau+\delta_{\mathrm{min}})\lambda_{K}(\Pi'\Pi)}{\tau+\delta_{\mathrm{max}}}}$.
		\item  Recall that $J_{*}=N(\mathcal{I},\mathcal{I})\mathscr{D}^{-1/2}_{\tau}(\mathcal{I},\mathcal{I})$ and $\sqrt{\frac{1}{\tau+\delta_{\mathrm{max}}}}\leq \mathscr{D}^{-1/2}_{\tau}(i,i)\leq \sqrt{\frac{1}{\tau+\delta_{\mathrm{min}}}}$, for $1\leq k\leq K$, we have
		\begin{align*}
		\frac{\sqrt{(\tau+\delta_{\mathrm{min}})\lambda_{K}(\Pi'\Pi)}}{\tau+\delta_{\mathrm{max}}}\leq J_{*}(k,k)\leq \frac{\sqrt{(\tau+\delta_{\mathrm{max}})K\lambda_{1}(\Pi'\Pi))}}{\tau+\delta_{\mathrm{min}}}
		\end{align*}
		Meanwhile, we also have $\|J_{*}\|_{F}\leq\frac{K\sqrt{(\tau+\delta_{\mathrm{max}})\lambda_{1}(\Pi'\Pi))}}{\tau+\delta_{\mathrm{min}}}$.
		\item For $1\leq i\leq n$, since $Y_{*}=VV'_{*,1}(\mathcal{I},:)(V_{*,1}(\mathcal{I},:)V'_{*,1}(\mathcal{I},:))^{-1}=VV^{-1}_{*,1}(\mathcal{I},:)$, we have
		\begin{align*}
		&\|e'_{i}Y_{*}\|_{F}=\|V(i,:)V^{-1}_{*,1}(\mathcal{I},:)\|_{F}\leq\|V(i,:)\|_{F}\|V^{-1}_{*,1}(\mathcal{I},:)\|_{F}\leq \|V(i,:)\|_{F}\frac{\sqrt{K}}{|\lambda_{K}(V_{*,1}(\mathcal{I},:))|}\\
		&= \|V(i,:)\|_{F}\frac{\sqrt{K}}{\lambda^{0.5}_{K}(V_{*,1}(\mathcal{I},:)V'_{*,1}(\mathcal{I},:))}\leq\frac{\tau+\delta_{\mathrm{max}}}{\tau+\delta_{\mathrm{min}}}\sqrt{\frac{K\kappa(\Pi'\Pi)}{\lambda_{K}(\Pi'\Pi)}}.
		\end{align*}
	\end{itemize}
	In Lemma 5.6, we consider permutation matrix $\mathcal{P}_{*}$ for CRSC, let $p_{*}(k)$ be the index of the $k$-th node after considering permutation. Recall that $\hat{J}_{*}=\hat{N}(\mathcal{\hat{I}}_{*},\mathcal{\hat{I}}_{*})D^{-1/2}_{\tau}(\mathcal{\hat{I}}_{*},\mathcal{\hat{I}}_{*})$ and $J_{*}=N(\mathcal{I},\mathcal{I})\mathscr{D}^{-1/2}_{\tau}(\mathcal{I},\mathcal{I})$, for $1\leq k\leq K$, we have
	\begin{align*}
	&|\hat{J}_{*}(k,k)-J_{*}(p_{*}(k),p_{*}(k))|=|e'_{k}\hat{N}(\mathcal{\hat{I}}_{*},\mathcal{\hat{I}}_{*})D^{-1/2}_{\tau}(\mathcal{\hat{I}}_{*},\mathcal{\hat{I}}_{*})e_{k}-e'_{k}\mathcal{P}_{*}N(\mathcal{I},\mathcal{I})\mathscr{D}^{-1/2}_{\tau}(\mathcal{I},\mathcal{I})\mathcal{P}_{*}e_{k}|\\
	&=|e'_{k}\hat{N}(\mathcal{\hat{I}}_{*},\mathcal{\hat{I}}_{*})D^{-1/2}_{\tau}(\mathcal{\hat{I}}_{*},\mathcal{\hat{I}}_{*})e_{k}-e'_{k}\mathcal{P}_{*}N(\mathcal{I},\mathcal{I})\mathcal{P}_{*}D^{-1/2}_{\tau}(\mathcal{\hat{I}}_{*},\mathcal{\hat{I}}_{*})e_{k}+e'_{k}\mathcal{P}_{*}N(\mathcal{I},\mathcal{I})\mathcal{P}_{*}D^{-1/2}_{\tau}(\mathcal{\hat{I}}_{*},\mathcal{\hat{I}}_{*})e_{k}\\
	&~~~-e'_{k}\mathcal{P}_{*}N(\mathcal{I},\mathcal{I})\mathscr{D}^{-1/2}_{\tau}(\mathcal{I},\mathcal{I})\mathcal{P}_{*}e_{k}|\\
	&\leq|e'_{k}(\hat{N}(\mathcal{\hat{I}}_{*},\mathcal{\hat{I}}_{*})-\mathcal{P}_{*}N(\mathcal{I},\mathcal{I})\mathcal{P}_{*})D^{-1/2}_{\tau}(\mathcal{\hat{I}}_{*},\mathcal{\hat{I}}_{*})e_{k}|+|e'_{k}\mathcal{P}_{*}N(\mathcal{I},\mathcal{I})(\mathcal{P}_{*}D^{-1/2}_{\tau}(\mathcal{\hat{I}}_{*},\mathcal{\hat{I}}_{*})-\mathscr{D}^{-1/2}_{\tau}(\mathcal{I},\mathcal{I})\mathcal{P}_{*})e_{k}|\\
	&\leq\frac{1}{\sqrt{\tau+\delta_{\mathrm{min}}}}|e'_{k}(\hat{N}(\mathcal{\hat{I}}_{*},\mathcal{\hat{I}}_{*})-\mathcal{P}_{*}N(\mathcal{I},\mathcal{I})\mathcal{P}_{*})e_{k}|+|e'_{k}(\mathcal{P}_{*}D^{-1/2}_{\tau}(\mathcal{\hat{I}}_{*},\mathcal{\hat{I}}_{*})-\mathscr{D}^{-1/2}_{\tau}(\mathcal{I},\mathcal{I})\mathcal{P}_{*})e_{k}|\mathrm{max}_{1\leq i\leq n}N(i,i)\\
	&\leq\frac{1}{\sqrt{\tau+\delta_{\mathrm{min}}}}|e'_{k}(\hat{N}(\mathcal{\hat{I}}_{*},\mathcal{\hat{I}}_{*})-\mathcal{P}_{*}N(\mathcal{I},\mathcal{I})\mathcal{P}_{*})e_{k}|\\
	&~~~+|e'_{k}(\mathcal{P}_{*}D^{-1/2}_{\tau}(\mathcal{\hat{I}}_{*},\mathcal{\hat{I}}_{*})-\mathscr{D}^{-1/2}_{\tau}(\mathcal{I},\mathcal{I})\mathcal{P}_{*})e_{k}|\sqrt{\frac{(\tau+\delta_{\mathrm{max}})K\lambda_{1}(\Pi'\Pi)}{\tau+\delta_{\mathrm{min}}}}\\
	&\leq \frac{1}{\sqrt{\tau+\delta_{\mathrm{min}}}}\mathrm{max}_{1\leq
		i\leq n}|\hat{N}(i,i)-N(i,i)|+\sqrt{\frac{(\tau+\delta_{\mathrm{max}})K\lambda_{1}(\Pi'\Pi)}{\tau+\delta_{\mathrm{min}}}}\mathrm{max}_{1\leq i\leq n}|D^{-1/2}_{\tau}(i,i)-\mathscr{D}^{-1/2}_{\tau}(i,i)|\\
	&=\frac{1}{\sqrt{\tau+\delta_{\mathrm{min}}}}\mathrm{max}_{1\leq
		i\leq n}|\hat{N}(i,i)-N(i,i)|\\
	&~~~+\sqrt{\frac{(\tau+\delta_{\mathrm{max}})K\lambda_{1}(\Pi'\Pi)}{\tau+\delta_{\mathrm{min}}}}\mathrm{max}_{1\leq i\leq n}|D^{-1/2}_{\tau}(i,i)(1-D^{1/2}_{\tau}(i,i)\mathscr{D}^{-1/2}_{\tau}(i,i))|\\
	&\leq\frac{1}{\sqrt{\tau+\delta_{\mathrm{min}}}}\mathrm{max}_{1\leq
		i\leq n}|\hat{N}(i,i)-N(i,i)|+\frac{\sqrt{(\tau+\delta_{\mathrm{max}})K\lambda_{1}(\Pi'\Pi))}}{\tau+\delta_{\mathrm{min}}}\mathrm{max}_{1\leq i\leq n}|1-D^{1/2}_{\tau}(i,i)\mathscr{D}^{-1/2}_{\tau}(i,i)|\\
	&\leq\frac{1}{\sqrt{\tau}+\delta_{\mathrm{min}}}\mathrm{max}_{1\leq
		i\leq n}|\hat{N}(i,i)-N(i,i)|+\frac{\sqrt{(\tau+\delta_{\mathrm{max}})K\lambda_{1}(\Pi'\Pi))}}{\tau+\delta_{\mathrm{min}}}\|I-D^{1/2}_{\tau}\mathscr{D}^{-1/2}_{\tau}\|\\
	&\mathrm{By~the~proof~of~Lemma~}5.2,\mathrm{with~probability~at~least~}1-o(\frac{K^{4\beta}}{n^{4\alpha-1}})\\
	&\leq\frac{1}{\sqrt{\tau+\delta_{\mathrm{min}}}}\mathrm{max}_{1\leq
		i\leq n}|\hat{N}(i,i)-N(i,i)|+\frac{err_{n}\sqrt{(\tau+\delta_{\mathrm{max}})K\lambda_{1}(\Pi'\Pi))}}{\tau+\delta_{\mathrm{min}}}\\
	&\leq\frac{1}{\sqrt{\tau+\delta_{\mathrm{min}}}}\mathrm{max}_{1\leq
		i\leq n}|\frac{1}{\|\hat{V}(i,:)\|}-\frac{1}{\|V(i,:)\|}|+\frac{err_{n}\sqrt{(\tau+\delta_{\mathrm{max}})K\lambda_{1}(\Pi'\Pi))}}{\tau+\delta_{\mathrm{min}}}\\
	&=\frac{1}{\sqrt{\tau+\delta_{\mathrm{min}}}}\mathrm{max}_{1\leq
		i\leq n}\frac{|\|\hat{V}(i,:)\hat{V}'\|_{F}-\|V(i,:)V'\|_{F}|}{\|\hat{V}(i,:)\|_{F}\|V(i,:)\|_{F}}+\frac{err_{n}\sqrt{(\tau+\delta_{\mathrm{max}})K\lambda_{1}(\Pi'\Pi))}}{\tau+\delta_{\mathrm{min}}}\\
	&\leq\frac{1}{\sqrt{\tau+\delta_{\mathrm{min}}}}\mathrm{max}_{1\leq
		i\leq n}\frac{\|e'_{i}(\hat{V}_{2}-V_{2})\|_{F}}{\|\hat{V}(i,:)\|_{F}\|V(i,:)\|_{F}}+\frac{err_{n}\sqrt{(\tau+\delta_{\mathrm{max}}K\lambda_{1}(\Pi'\Pi))}}{\tau+\delta_{\mathrm{min}}}\\
	&=\frac{1}{\sqrt{\tau+\delta_{\mathrm{min}}}}\mathrm{max}_{1\leq
		i\leq n}\frac{\|e'_{i}(\hat{V}_{2}-V_{2})\|_{F}}{\|\hat{V}(i,:)\hat{V}'-V(i,:)V'+V(i,:)V'\|_{F}\|V(i,:)\|_{F}}+\frac{err_{n}\sqrt{(\tau+\delta_{\mathrm{max}})K\lambda_{1}(\Pi'\Pi))}}{\tau+\delta_{\mathrm{min}}}\\
	&\leq\frac{1}{\sqrt{\tau+\delta_{\mathrm{min}}}}\mathrm{max}_{1\leq
		i\leq n}\frac{\|e'_{i}(\hat{V}_{2}-V_{2})\|_{F}}{|\|\hat{V}(i,:)\hat{V}'-V(i,:)V'\|_{F}-\|V(i,:)V'\|_{F}|\|V(i,:)\|_{F}}+\frac{err_{n}\sqrt{(\tau+\delta_{\mathrm{max}})K\lambda_{1}(\Pi'\Pi))}}{\tau+\delta_{\mathrm{min}}}\\
	&=\frac{1}{\sqrt{\tau+\delta_{\mathrm{min}}}}\mathrm{max}_{1\leq
		i\leq n}\frac{1}{|1-\frac{\|V(i,:)\|_{F}}{\|e'_{i}(\hat{V}_{2}-V_{2})\|_{F}}|\|V(i,:)\|_{F}}+\frac{err_{n}\sqrt{(\tau+\delta_{\mathrm{max}})K\lambda_{1}(\Pi'\Pi))}}{\tau+\delta_{\mathrm{min}}}\\
	&\overset{\mathrm{By~Lemma~}\ref{P2}}{\leq}\frac{\sqrt{(\tau+\delta_{\mathrm{max}})K\lambda_{1}(\Pi'\Pi))}}{\tau+\delta_{\mathrm{min}}}\mathrm{max}_{1\leq
		i\leq n}\frac{1}{|1-\frac{\|V(i,:)\|_{F}}{\|e'_{i}(\hat{V}_{2}-V_{2})\|_{F}}|}+\frac{err_{n}\sqrt{(\tau+\delta_{\mathrm{max}})K\lambda_{1}(\Pi'\Pi))}}{\tau+\delta_{\mathrm{min}}}\\
	&=O(\frac{\sqrt{(\tau+\delta_{\mathrm{max}})K\lambda_{1}(\Pi'\Pi))}}{\tau+\delta_{\mathrm{min}}}),
	\end{align*}
	where  we have used the fact that $err_{n}=O(\frac{\sqrt{\rho n\mathrm{log}(n^{\alpha}K^{-\beta})}}{\tau+\delta_{\mathrm{min}}})\leq O(1)$ when $\tau+\delta_{\mathrm{min}}\geq C\sqrt{\rho n\mathrm{log}(n^{\alpha}K^{-\beta})}$. Then we have $\|\hat{J}_{*}-\mathcal{P}'_{*}J_{*}\mathcal{P}_{*}\|_{F}=O(\frac{K\sqrt{(\tau+\delta_{\mathrm{max}})\lambda_{1}(\Pi'\Pi))}}{\tau+\delta_{\mathrm{min}}})$.
	Then, for $1\leq i\leq n$, since $Z_{*}=Y_{*}J_{*}, \hat{Z}_{*}=\hat{Y}_{*}\hat{J}_{*}$, we have
	\begin{align*}
	&\|e'_{i}(\hat{Z}_{*}-Z_{*}\mathcal{P}_{*})\|_{F}=\|e'_{i}(\mathrm{max}(0, \hat{Y}_{*}\hat{J}_{*})-Y_{*}J_{*}\mathcal{P}_{*})\|_{F}\leq \|e'_{i}(\hat{Y}_{*}\hat{J}_{*}-Y_{*}J_{*}\mathcal{P}_{*})\|_{F}\\
	&=\|e'_{i}(\hat{Y}_{*}-Y_{*}\mathcal{P}_{*})\hat{J}_{*}+e'_{i}Y_{*}\mathcal{P}_{*}(\hat{J}_{*}-\mathcal{P}'_{*}J_{*}\mathcal{P}_{*})\|_{F}\\
	&\leq\|e'_{i}(\hat{Y}_{*}-Y_{*}\mathcal{P}_{*})\|_{F}\|\hat{J}_{*}\|_{F}+\|e'_{i}Y_{*}\mathcal{P}_{*}\|_{F}\|\hat{J}_{*}-\mathcal{P}'_{*}J_{*}\mathcal{P}_{*}\|_{F}\\
	&=\|e'_{i}(\hat{Y}_{*}-Y_{*}\mathcal{P}_{*})\|_{F}\|\hat{J}_{*}-\mathcal{P}'_{*}J_{*}\mathcal{P}_{*}+\mathcal{P}'_{*}J_{*}\mathcal{P}_{*}\|_{F}+\|e'_{i}Y_{*}\mathcal{P}_{*}\|_{F}\|\hat{J}_{*}-\mathcal{P}'_{*}J_{*}\mathcal{P}_{*}\|_{F}\\
	&\leq\|e'_{i}(\hat{Y}_{*}-Y_{*}\mathcal{P}_{*})\|_{F}(\|\hat{J}_{*}-\mathcal{P}'_{*}J_{*}\mathcal{P}_{*}\|_{F}+\|J_{*}\|_{F})+\|e'_{i}Y_{*}\|_{F}\|\hat{J}_{*}-\mathcal{P}'_{*}J_{*}\mathcal{P}_{*}\|_{F}\\
	&\leq O((\frac{\tau+\delta_{\mathrm{max}}}{\tau+\delta_{\mathrm{min}}})^{5}\frac{K^{3.5}\varpi\kappa^{4.5}(\Pi'\Pi)}{\eta})(O(\frac{K\sqrt{(\tau+\delta_{\mathrm{max}})\lambda_{1}(\Pi'\Pi))}}{\tau+\delta_{\mathrm{min}}})+\frac{K\sqrt{(\tau+\delta_{\mathrm{max}})\lambda_{1}(\Pi'\Pi))}}{\tau+\delta_{\mathrm{min}}})\\
	&~~~+\frac{\tau+\delta_{\mathrm{max}}}{\tau+\delta_{\mathrm{min}}}\sqrt{\frac{K\kappa(\Pi'\Pi)}{\lambda_{K}(\Pi'\Pi)}}O(\frac{K\sqrt{(\tau+\delta_{\mathrm{max}})\lambda_{1}(\Pi'\Pi))}}{\tau+\delta_{\mathrm{min}}})\\
	&=O(\frac{(\tau+\delta_{\mathrm{max}})^{5.5}K^{4.5}\varpi\kappa^{4.5}(\Pi'\Pi)\sqrt{\lambda_{1}(\Pi'\Pi)}}{\eta(\tau+\delta_{\mathrm{min}})^{6}}).
	\end{align*}
	By the proof of Lemma 5.6 for CRSC algorithm, we know that $\eta\geq\frac{(\tau+\delta_{\mathrm{min}})^{2}\pi_{\mathrm{min}}}{(\tau+\delta_{\mathrm{max}})^{2}K\lambda_{1}(\Pi'\Pi)}$, we have
	\begin{align*}
	\|e'_{i}(\hat{Z}_{*}-Z_{*}\mathcal{P}_{*})\|_{F}=O(\frac{(\tau+\delta_{\mathrm{max}})^{7.5}K^{5.5}\varpi\kappa^{4.5}(\Pi'\Pi)\lambda^{1.5}_{1}(\Pi'\Pi)}{(\tau+\delta_{\mathrm{min}})^{8}\pi_{\mathrm{min}}}).
	\end{align*}
\end{proof}
\subsection{Proof of Theorem 5.9}
\begin{proof}
	For SRSC, the difference between the row-normalized projection coefficients $\Pi$ and $\hat{\Pi}$ can be bounded by the difference between $Z$ and $\hat{Z}$, for $1\leq i\leq n$, we have
	\begin{align*}	\|e'_{i}(\hat{\Pi}-\Pi\mathcal{P})\|_{F}&=\|\frac{e'_{i}\hat{Z}}{\|e'_{i}\hat{Z}\|_{F}}-\frac{e'_{i}Z\mathcal{P}}{\|e'_{i}Z\mathcal{P}\|_{F}}\|_{F}=\|\frac{e'_{i}\hat{Z}\|e'_{i}Z\|_{F}-e'_{i}Z\mathcal{P}\|e'_{i}\hat{Z}\|_{F}}{\|e'_{i}\hat{Z}\|_{F}\|e'_{i}Z\|_{F}}\|_{F}\\	&=\|\frac{e'_{i}\hat{Z}\|e'_{i}Z\|_{F}-e'_{i}\hat{Z}\|e'_{i}\hat{Z}\|_{F}+e'_{i}\hat{Z}\|e'_{i}\hat{Z}\|_{F}-e'_{i}Z\mathcal{P}\|e'_{i}\hat{Z}\|_{F}}{\|e'_{i}\hat{Z}\|_{F}\|e'_{i}Z\|_{F}}\|_{F}\\
	&\leq \frac{\|e'_{i}\hat{Z}\|e'_{i}Z\|_{F}-e'_{i}\hat{Z}\|e'_{i}\hat{Z}\|_{F}\|_{F}+\|e'_{i}\hat{Z}\|e'_{i}\hat{Z}\|_{F}-e'_{i}Z\mathcal{P}\|e'_{i}\hat{Z}\|_{F}\|_{F}}{\|e'_{i}\hat{Z}\|_{F}\|e'_{i}Z\|_{F}}\\
	&=\frac{\|e'_{i}\hat{Z}\|_{F}|\|e'_{i}Z\|_{F}-\|e'_{i}\hat{Z}\|_{F}|+\|e'_{i}\hat{Z}\|_{F}\|e'_{i}\hat{Z}-e'_{i}Z\mathcal{P}\|_{F}}{\|e'_{i}\hat{Z}\|_{F}\|e'_{i}Z\|_{F}}\\
	&=\frac{|\|e'_{i}Z\|_{F}-\|e'_{i}\hat{Z}\|_{F}|+\|e'_{i}\hat{Z}-e'_{i}Z\mathcal{P}\|_{F}}{\|e'_{i}Z\|_{F}}\leq\frac{2\|e'_{i}(\hat{Z}-Z\mathcal{P})\|_{F}}{\|e'_{i}Z\|_{F}}\\
	&\leq \frac{2\|e'_{i}(\hat{Z}-Z\mathcal{P})\|_{F}}{\mathrm{min}_{1\leq j\leq n}\|e'_{j}Z\|_{F}}\leq O(\|e'_{i}(\hat{Z}-Z\mathcal{P})\|_{F}\sqrt{K(\tau+\delta_{\mathrm{max}})}),
	\end{align*}
	where we have used the fact that $\|e'_{j}Z\|_{F}=\|e'_{j}\mathscr{D}^{-1/2}_{\tau}\Pi\|_{F}=\|\mathscr{D}^{-1/2}_{\tau}(j,j)\Pi(j,:)\|\geq \frac{1}{\sqrt{K(\tau+\delta_{\mathrm{max}})}}$.
	Combine the above result with Lemmas 5.8 and 5.4, we have
	\begin{align*}
	\|e'_{i}(\hat{\Pi}-\Pi\mathcal{P})\|_{F}&=O(\frac{(\tau+\delta_{\mathrm{max}})^{2}K^{1.5}\varpi\kappa(\Pi'\Pi)\sqrt{\lambda_{1}(\Pi'\Pi)}}{(\tau+\delta_{\mathrm{min}})^{2}})\\
	&=O(\frac{(\tau+\delta_{\mathrm{max}})^{3}K^{2}\kappa(\Pi'\Pi)\sqrt{\lambda_{1}(\Pi'\Pi)\mathrm{log}(n^{\alpha}K^{-\beta})}}{\tau(\tau+\delta_{\mathrm{min}})^{2}|\lambda_{K}(\tilde{P})|\lambda_{K}(\Pi'\Pi)\sqrt{\rho}}).
	\end{align*}
	Similarly, for CRSC method, we have $\|e'_{i}(\hat{\Pi}_{*}-\Pi\mathcal{P}_{*})\|_{F}\leq\frac{2\|e'_{i}(\hat{Z}_{*}-Z_{*}\mathcal{P})\|_{F}}{\mathrm{min}_{1\leq j\leq n}\|e'_{j}Z_{*}\|_{F}}$. Recall that $Z_{*}=VV^{-1}_{*,1}(\mathcal{I},:)N(\mathcal{I},\mathcal{I})\mathscr{D}^{-1/2}_{\tau}(\mathcal{I},\mathcal{I})$ and $V_{*,1}(\mathcal{I},:)=N(\mathcal{I},\mathcal{I})V(\mathcal{I},:)$, we have $Z_{*}=VV^{-1}(\mathcal{I},:)\mathscr{D}^{-1/2}_{\tau}(\mathcal{I},\mathcal{I})=VV^{-1}_{\tau,1}(\mathcal{I},:)\equiv Z$, which gives that $\|e'_{i}(\hat{\Pi}_{*}-\Pi\mathcal{P}_{*})\|_{F}=O(\|e'_{i}(\hat{Z}_{*}-Z_{*}\mathcal{P}_{*})\|_{F}\sqrt{K(\tau+\delta_{\mathrm{max}})})$. Then, by Lemmas 5.8 and 5.4, we have
	\begin{align*}
	\|e'_{i}(\hat{\Pi}_{*}-\Pi\mathcal{P}_{*})\|_{F}&=O(\frac{(\tau+\delta_{\mathrm{max}})^{8}K^{6}\varpi\kappa^{4.5}(\Pi'\Pi)\lambda^{1.5}_{1}(\Pi'\Pi)}{(\tau+\delta_{\mathrm{min}})^{8}\pi_{\mathrm{min}}})\\
	&=O(\frac{(\tau+\delta_{\mathrm{max}})^{9}K^{6.6}\kappa^{4.5}(\Pi'\Pi)\lambda^{1.5}_{1}(\Pi'\Pi)\sqrt{\mathrm{log}(n^{\alpha}K^{-\beta})}}{\tau(\tau+\delta_{\mathrm{min}})^{8}|\lambda_{K}(\tilde{P})|\lambda_{K}(\Pi'\Pi)\pi_{\mathrm{min}}\sqrt{\rho}}).
	\end{align*}
\end{proof}

\section{One-Class SVM and SVM-cone algorithm}\label{OneClassSVMandSVMcone}
In this section, we briefly introduce one-class SVM and SVM-cone algorithm given in \cite{MaoSVM}.

As mentioned in Problem 1 in \cite{MaoSVM}, if a matrix $S\in\mathbb{R}^{n\times m}$ has the form $S=HS_{C}$, where $H\in\mathrm{R}^{n\times K}$ with nonnegative entries, no row of $H$ is 0, and $S_{C}\in\mathbb{R}^{K\times m}$ corresponding to $K$ rows of $S$ (i.e., there exists an index set $\mathcal{I}$ with $K$ entries such that $S_{C}=S(\mathcal{I},:)$), and each row of $S$ has unit $l_{2}$ norm. Then problem of inferring $H$ from $S$ is called the ideal cone problem. The ideal cone problem can be solved by one-class SVM applied to the rows of $S$. the $K$ normalized corners in $S_{C}$ are the support vectors found by a one-class SVM:
\begin{align}\label{OneClassSVM}
\mathrm{maximize~}b~~\mathrm{s.t.}~~\textbf{w}'S(i,:)\geq b(\mathrm{~for~}i=1,2,\ldots,n)~\mathrm{and~~}\|\textbf{w}\|_{F}\leq 1.
\end{align}
The solution  $(\textbf{w}, b)$ for the ideal cone problem when $(S_{C}S'_{C})^{-1}\mathbf{1}>0$ is given by
\begin{align}\label{SolutionOfOneClassSVM}
\textbf{w}=b^{-1}\cdot S'_{C}\frac{(S_{C}S'_{C})^{-1}\mathbf{1}}{\mathbf{1}'(S_{C}S'_{C})^{-1}\mathbf{1}},~~~ b=\frac{1}{\sqrt{\mathbf{1}'(S_{C}S'_{C})^{-1}\mathbf{1}}}.
\end{align}
for the empirical case, if we are given a matrix  $\hat{S}\in\mathbb{R}^{n\times m}$ such that all rows of $\hat{S}$  have unit $l_{2}$ norm, infer $H$ from $\hat{S}$ with given $K$ is called the empirical cone problem (i.e., Problem 2 in \cite{MaoSVM}). For the empirical cone problem, we can apply one-class SVM to all rows of $\hat{S}$ to obtain $\textbf{w}$ and $b$'s estimations $\hat{\textbf{w}}$ and $\hat{b}$. Then apply K-means algorithm to rows of $\hat{S}$ that are close to the hyperplane into $K$ clusters, the $K$ clusters can give the estimation of the  index set $\mathcal{I}$. Below is the SVM-cone algorithm given in \cite{MaoSVM}.
\begin{algorithm}
	\caption{\textbf{SVM-cone} \citep{MaoSVM}}
	\label{alg:SVMcone}
	\begin{algorithmic}[1]
		\Require $\hat{S}\in \mathbb{R}^{n\times m}$ with rows have unit $l_{2}$ norm, number of corners $K$, estimated distance corners from hyperplane $\gamma$.
		\Ensure The near-corner index set $\mathcal{\hat{I}}$.
		\State Run one-class SVM on $\hat{S}(i,:)$ to get $\hat{\textbf{w}}$ and $\hat{b}$
		\State Run K-means algorithm to the set $\{\hat{S}(i,:)| \hat{S}(i,:)\hat{\textbf{w}}\leq \hat{b}+\gamma\}$ that are close to the hyperplane into $K$ clusters
		\State Pick one point from each cluster to get the near-corner set $\mathcal{\hat{I}}$
	\end{algorithmic}
\end{algorithm}
As suggested in \cite{MaoSVM}, we can start $\gamma=0$ and incrementally increase it until $K$ distinct clusters are found.

Now turn to our CRSC algorithm and CRSC-equivalence algorithm. Set $\textbf{w}_{1}=b_{1}^{-1}V'_{*,1}(\mathcal{I},:)\frac{(V_{*,1}(\mathcal{I},:)V'_{*,1}(\mathcal{I},:))^{-1}\mathbf{1}}{\mathbf{1}'(V_{*,1}(\mathcal{I},:)V'_{*,1}(\mathcal{I},:))^{-1}}, b_{1}=\frac{1}{\sqrt{\mathbf{1}'(V_{*,1}(\mathcal{I},:)V'_{*,1}(\mathcal{I},:))^{-1}\mathbf{1}}}$, and $\textbf{w}_{2}=b_{2}^{-1}V'_{*,2}(\mathcal{I},:)\frac{(V_{*,2}(\mathcal{I},:)V'_{*,2}(\mathcal{I},:))^{-1}\mathbf{1}}{\mathbf{1}'(V_{*,2}(\mathcal{I},:)V'_{*,2}(\mathcal{I},:))^{-1}}, b_{2}=\frac{1}{\sqrt{\mathbf{1}'(V_{*,2}(\mathcal{I},:)V'_{*,2}(\mathcal{I},:))^{-1}\mathbf{1}}}$ such that $\textbf{w}_{1}$ and $b_{1}$ are solutions of the one-class SVM in Eq (\ref{OneClassSVM}) by setting $S=V_{*,1}$, and $\textbf{w}_{2}$ and $b_{2}$ are solutions of the one-class SVM in Eq (\ref{OneClassSVM}) by setting $S=V_{*,2}$ . By Lemma \ref{WhyUseKmeansInSVMcone}, we see that if node $i$ is a pure node, then we have $V_{*,1}(i,:)\textbf{w}_{1}=b_{1}$, which suggests that in the SVM-cone algorithm, if the input matrix is $V_{*,1}$, by setting $\gamma=0$, we can find all pure nodes, i.e., the set $\{V_{*,1}(i,:)|V_{*,1}(i,:)\textbf{w}_{1}=b_{1}\}$ contain all rows of $V_{*,1}$ respective to pure nodes while including mixed nodes. By Lemma 3.3, we see that these  pure nodes belong to $K$ distinct clusters such that if nodes $i,j$ are in the same clusters, then we have $V_{*,1}(i,:)=V_{*,1}(j,:)$, and this is the reason that we need to apply K-means algorithm on the set obtained in step 2 in the SVM-cone algorithm to obtain the $K$ distinct clusters, and this is also the reason that we said SVM-cone returns the index set $\mathcal{I}$ up to a permutation when the input is $V_{*,1}$ in the explanation of Figure 1 in the main manuscript. Similar arguments hold when the input is $V_{*,2}$ in the SVM-cone algorithm.
\begin{lem}\label{WhyUseKmeansInSVMcone}
	Under $MMSB(n,P,\Pi)$, for $1\leq i\leq n$, if node $i$ is a pure node such that $\Pi(i,k)=1$ for certain $k$, we have
	\begin{align*}
	V_{*,1}(i,:)\textbf{w}_{1}=b_{1}\mathrm{~~~and~~~}V_{*,2}(i,:)\textbf{w}_{2}=b_{2},
	\end{align*}
	Meanwhile, if node $i$ is not a pure node, then the above equalities do not hold.
\end{lem}
\begin{proof}
	We only prove that $V_{*,1}(i,:)\textbf{w}_{1}=b_{1}$ when $\Pi(i,k)=1$, since the second equality can be proved similarly. By Lemma \ref{P1}, we know that when node $i$ is a pure node such that $\Pi(i,k)=1$, $V_{*,1}(i,:)$ can be written as $V_{*,1}(i,:)=e'_{k}V_{*,1}(\mathcal{I},:)$, then we have $V_{*,1}(i,:)\textbf{w}_{1}=b_{1}$ surely. And if $i$ is a mixed node, by Lemma \ref{P1}, we know that $r_{1}(i)>1$ and $\Phi_{1}(i,:)\neq e_{k}$ for any $k=1,2,\ldots, K$, hence $V_{*,1}(i,:)\neq e'_{k}V_{*,1}(\mathcal{I},:)$ if $i$ is mixed, which gives the result.
\end{proof}

\end{document}